\documentclass{elsarticle}

\usepackage[framestyle=fbox,framefit=yes,heightadjust=all,framearound=all]{floatrow}
\usepackage{etex}
\usepackage[normalem]{ulem}
\usepackage{array}
\usepackage[T1]{fontenc} 
\usepackage[utf8]{inputenc} 
\usepackage{cmll}
\usepackage{stmaryrd}
\usepackage{csquotes}
\usepackage{amsfonts}
\usepackage{amssymb}
\usepackage{amsmath}
\usepackage{amsthm}
\DeclareMathVersion{normal2}
\mathversion{normal2}
\mathversion{normal}
\DeclareMathVersion{normal3}
\mathversion{normal3}
\usepackage{MnSymbol}
\mathversion{normal}
\usepackage[normalem]{ulem}
\usepackage{cancel}  
\usepackage{xspace}
\usepackage{amsthm}
\usepackage{mathpartir}
\usepackage{mathtools}
\usepackage{adjustbox}
\usepackage{graphicx}
\usepackage{rotating}
\usepackage{subfigure}
\usepackage{varwidth}
\usepackage{bussproofs}
\usepackage{longtable}
\usepackage{mathtools}
\usepackage[shortlabels]{enumitem}
\usepackage{hyperref}
\usepackage{xcolor}
\definecolor{mygreen}{rgb}{0,0.6,0}
\definecolor{mygray}{rgb}{0.5,0.5,0.5}
\definecolor{mymauve}{rgb}{0.58,0,0.82}
\hypersetup{
    colorlinks,
    linkcolor={red!50!black},
    citecolor={blue!50!black},
    urlcolor={blue!80!black}
}
\usepackage{tikz}
\usetikzlibrary{arrows.meta}
\usetikzlibrary{decorations.markings}
\usetikzlibrary{matrix}
\usetikzlibrary{arrows}
\usetikzlibrary{decorations.pathmorphing}
\usetikzlibrary{shapes}
\usetikzlibrary{arrows}
\usetikzlibrary{positioning} 

\usepackage{circuitikz}
\usepackage{tikz-cd}

\usepackage{enumitem} 

\theoremstyle{plain}
\newtheorem{thm}{Theorem}
\newtheorem{lem}[thm]{Lemma}
\newtheorem{prop}[thm]{Proposition}

\newtheorem{conj}[thm]{Conjecture}
\newtheorem{fact}[thm]{Fact}

\theoremstyle{definition}
\newtheorem{defn}{Definition}
\newtheorem{exmp}{Example}

\theoremstyle{remark}
\newtheorem{rem}{Remark}

\newcommand{\red}{\rightsquigarrow}

\usepackage{comment}

\usepackage{marginnote}
\begin{document}
\begin{frontmatter}

\title{A type-assignment of linear erasure and duplication}

\author[unito]{Gianluca~Curzi}
\ead{curzi@di.unito.it, gianluca.curzi@unito.it}
\ead[urlpa]{http://www.di.unito.it/~curzi}
\author[unito]{Luca~Roversi}
\ead{roversi@di.unito.it, luca.roversi@unito.it}
\ead[urlr]{http://www.di.unito.it/~rover}

\address[unito]{Dipartimento di Informatica -- Universit\`a di Torino}

\begin{abstract}
We introduce  $\mathsf{LEM}$, a type-assignment system for the linear $ \lambda $-calculus that extends second-order 
$\mathsf{IMLL}_2$, i.e., intuitionistic multiplicative Linear Logic, by 
means of logical rules that weaken and contract assumptions, but in a 
purely linear setting. $\mathsf{LEM}$ enjoys both a mildly weakened 
cut-elimination, whose computational cost is cubic, and 
Subject reduction. A translation of $\mathsf{LEM}$ 
into $\mathsf{IMLL}_2$ exists such that the derivations of the former can 
exponentially compress the dimension of the derivations in the latter. 
$\mathsf{LEM}$ 
allows for a modular and compact representation of boolean circuits, directly encoding the fan-out nodes, 
by contraction, and disposing 
garbage, by weakening.
It can also represent natural numbers with terms  very close to
standard Church numerals which, moreover, apply to Hereditarily Finite 
Permutations, i.e. a group structure that exists inside the linear 
$ \lambda $-calculus.
\end{abstract}

\end{frontmatter}

\section{Introduction}
\label{section:Introduction}
Girard introduces \textit{Linear Logic} ($\mathsf{LL}$) in 
\cite{Girard:TCS87} as a refinement of both classical and intuitionistic 
logic. \textsf{LL} decomposes the intuitionistic implication  
\enquote{$\Rightarrow$} into the more primitive linear implication 
\enquote{$\multimap$} and modality \enquote{$\oc$} (of course), 
the latter giving a logical status to weakening and contraction by
means of the so-called  \textit{exponential rules}. 
According to the \emph{Curry-Howard correspondence},
 this  decomposition allows to identify a strictly linear 
component of the functional computations that interacts with the non-linear one, in which duplication and erasure are allowed.

This work focuses on $\mathsf{IMLL}_2$, i.e.~second-order intuitionistic \textit{multiplicative} Linear Logic which, we recall, is free of any kind of exponential rules. 
The Curry-Howard correspondence tightly relates $\mathsf{IMLL}_2$ and 
the linear $\lambda$-calculus, a sub-language of the standard $\lambda$-calculus without explicit erasure and duplication.

Interesting works exist on the expressiveness of both the untyped and the typed linear $ \lambda $-calculus.

Alves et al.~\cite{DBLP:conf/csl/AlvesFFM06} recover the full computational power of G\"odel System $T$ by adding booleans, natural numbers, and a linear iterator to the linear $\lambda$-calculus, the non-linear features coming specifically from the iterator and the numerals. 

Matsuoka investigates the discriminating power of linear $ \lambda $-terms  with types in $ \textsf{IMLL}$, i.e.~intuitionistic multiplicative Linear Logic, proving typed variants of B\"ohm Theorem~\cite{MATSUOKA200737}. We remark that, in this setting, discriminating among linear $ \lambda $-terms relies on \emph{a specific form of weakening} already inside $ \textsf{IMLL} $. 

Another work that exploits the built-in erasure and copying mechanisms of the linear $\lambda$-calculus is by Mairson~\cite{mairsonlinear}. 
With no new constructors, Mairson encodes boolean circuits in the linear 
$ \lambda $-calculus. Moreover, Mairson\&Terui reformulate Mairson's 
results inside $\mathsf{IMLL}_2$ and prove bounds on the complexity of the cut-elimination in sub-systems of \textsf{LL}~\cite{mairson2003computational}.

\paragraph{Contributions}
Starting from Mairson\&Terui's \cite{mairson2003computational}, this work investigates a structural proof-theory, and the related Curry-Howard correspondence,  of $\mathsf{IMLL}_2$  extended with inference rules  for contraction and  weakening. 
\begin{enumerate}
\item 
We introduce the \emph{Linearly Exponential and Multiplicative}  system $\mathsf{LEM}$, giving a logical status to
the erasure and the duplication that~\cite{mairson2003computational} identifies inside the linear $ \lambda $-calculus. 
$\mathsf{LEM}$ is a type-assignment for a \emph{linear} 
$\lambda$-calculus endowed with constructs for weakening and 
contraction, and it is obtained by extending $\mathsf{IMLL}_2$ with 
rules on modal formulas ``$ \shpos A$''.  $\mathsf{LEM}$ can be seen  as a sub-system of $\mathsf{LL}$ 
with a restricted form ``$ \shpos $'' of \enquote{$\oc$}.

\item
We consider a mildly weakened cut-elimination, called ``lazy'',  that  faithfully represents the mechanism of linear erasure and  duplication discussed in~\cite{mairson2003computational}, and we  identify a set of derivations in $\mathsf{LEM}$
that rewrite to cut-free 
ones under that lazy cut-elimination in a cubic number of steps 
(Section~\ref{Cut elimination, complexity, and subject reduction for IMLL2^shpos}).  
Moreover, we show the Subject reduction of $\mathsf{LEM}$ 
(Section~\ref{sec: subject reduction}).

\item 
We prove that the cut-elimination of $\mathsf{IMLL}_2$ can simulate the 
one of $\mathsf{LEM}$ at a cost which can be exponential in the size of 
the given derivation of $\mathsf{LEM}$ 
(Section~\ref{sec: the expressiveness of the system}).
So,  $\mathsf{LEM}$ can speed up the cut-elimination of $\mathsf{IMLL}_2$,  meaning that it compresses in smaller derivations what can be algorithmically expressed in $\mathsf{IMLL}_2$.

\item Hence, we explore the algorithmic expressiveness of $\mathsf{LEM}$
(Section~\ref{section:The expressiveness of IMLLshpos2 and applications}):
\begin{enumerate}
	\item 
	Both $\mathsf{LEM}$ and $\mathsf{IMLL}_2$ can represent boolean circuits.	
	However, the copying mechanism, directly available in $\mathsf{LEM}$, makes the encoding of the fan-out of the nodes of the circuit essentially natural, facilitating the modularity and the readability of the encoding itself.
	Moreover, the erasure in $\mathsf{LEM}$ avoids to accumulate garbage when evaluating a circuit represented by a derivation of $\mathsf{LEM}$, unlike in other proposals.
	\item 
	We show that numerals, structurally related to Church ones, 
	exist in $\mathsf{LEM}$. Their type is
	$ (\shpos \forall\alpha.(\alpha\multimap\alpha) ) \multimap\forall\alpha.(\alpha\multimap\alpha) $ that forbids
	iterations longer than the complexity of the lazy cut-elimination.  Remarkably, the numerals in $\mathsf{LEM}$ admit successor and addition that work as expected, thanks to the Subject reduction. 

    \item 
	Finally, we show that Hereditarily Finite Permutations, which form a group inside the  linear
	$ \lambda $-calculus, inhabit a simple generalization 
	of the here above type of numerals, so possibly connecting 
	$\mathsf{LEM}$ with reversible computations.
\end{enumerate}
The above contributions follow from a fully detailed, and not
at all obvious, technical reworking of Mairson\&Terui's \cite{mairson2003computational} work. 
We propose it as a solid base to further investigations concerning duplication and erasure   in a purely linear setting.
\end{enumerate}
Section~\ref{sec: background} is about (formal) preliminaries. Section~\ref{sec: duplication and erasure in a linear setting} introduces the motivating background and 
Section~\ref{sec:  the system shpos, with, oplus  IMLL2 cut elimination and bound} formally defines $\mathsf{LEM}$.

\paragraph{Acknowledgments}
We are indebted to the anonymous reviewers for their patience and  constructive attitude with which they red and commented on previous versions of this work.
\section{Preliminaries} \label{sec: background}
\subsection{The linear $\lambda$-calculus}

We assume the reader to be familiar with standard $\lambda$-calculus and related concepts like:
(i) the set $ FV(M) $ of the free variables of the $ \lambda $-term $ M $, 
(ii) the meta-level substitution $M[N/x]$ that replaces 
the $ \lambda $-term $ N $ for every free occurrence of the variable $ x $
in $ M $, 
(iii) the contexts $ \mathcal{C}[] $, i.e.~$ \lambda $-terms with a place-holder (hole) $ [] $ 
 that may capture free variables of a $ \lambda $-term plugged into $ [] $,
 (iv) the $\alpha$-equivalence ($=_\alpha$),
(v) the $\beta$-reduction $(\lambda x.M)N\rightarrow_\beta M[N/x]$, 
(vi) the $\eta$-reduction $\lambda x.Mx\rightarrow_\eta M$ that can apply 
if $ x $ is not free in $ M $. Both $\rightarrow_\beta$ and $\rightarrow_\eta$ are considered contextually closed.



By $\rightarrow_\beta^*$ we denote the reflexive and transitive closure of the $ \beta $-reduction, and by  $=_{\beta}$  its reflexive, symmetric and transitive closure.

Also, 
by $\rightarrow_\eta^*$ we denote the reflexive and transitive closure of the $ \eta $-reduction, and by  $=_{\eta}$  its reflexive, symmetric and transitive closure.

Finally, by $\rightarrow_{\beta \eta}$ we denote $\rightarrow_\beta \cup \rightarrow_{\eta}$, and by $\rightarrow^*_{\beta \eta}$ we denote its reflexive and transitive closure.

A $\lambda$-term is in $\beta$-\textit{normal form}, or simply 
($ \beta $-)\emph{normal},
whenever no $ \beta $-reduction applies to it. 
A $\lambda$-term is in $\eta$-\textit{normal form}, or simply \emph{$ \eta $-normal} if no $\eta$-reduction applies to it. Finally, a $\lambda$-term is in $\beta\eta$-\textit{normal form}, or simply  $ \beta \eta $-\emph{normal}, whenever no $ \beta\eta $-reduction applies to it.

A $\lambda$-term is \textit{closed} if $FV(M)=\emptyset$.  

The \textit{size} $\vert M \vert $ of $M$ is the number of nodes in its syntax tree. 

The \emph{linear} $\lambda$-calculus is the $\lambda$-calculus restricted to \emph{linear} $\lambda$-terms: 

\begin{defn}[Linear $\lambda$-terms] 
\label{defn:Linear lambda-terms}
A $\lambda$-term $M$ is \textit{linear} if all of its free variables occur once in it and every proper sub-term $ \lambda x.M' $ of $M$ is such that  $ x $ occurs in  $ M' $ and $M' $ is linear.
\qed
\end{defn}
\noindent
For example, 
$I\triangleq \lambda x.x$ and 
$C\triangleq\lambda x. \lambda y. \lambda z. xzy$ are linear, 
while $K \triangleq \lambda x. \lambda y. x$ and 
$S\triangleq\lambda x. \lambda y. \lambda z. xz(yz) $ are not.

\vspace{\baselineskip}
\noindent
To our purposes, we shall adopt the following notion of value:
\begin{defn}[Values]
\label{defn:Values among linear lambda-terms}
A \emph{value} is every linear $ \lambda $-term which is both 
($ \beta $-)\emph{normal} and \emph{closed}. 
\qed
\end{defn}
\noindent
We shall generally use $ V $ and $ U $ to range over values.

\begin{fact}[Stability]
\label{fact:Stability}
Linear $\lambda$-terms are stable under $\beta$-reduction, 
i.e.~$M$ linear and $M \rightarrow_\beta N $ imply $N$ is linear. 
Analogously, linear $\lambda$-terms are stable under $\eta$-red\-uct\-ion, i.e.~$M$ linear and $M \rightarrow_\eta N$ imply $N$ is linear. In both cases,   $FV(N)= FV(M)$.
\qed
\end{fact}
\noindent
Finally, we shall write $M \circ N$ in place of $\lambda z. M(Nz)$.

\begin{figure}[ht]
	\begin{mathpar}
		\inferrule*[Right= $ax$]
		{\\}
		{x: A \vdash x:A} 
		\and
		\inferrule*[Right= $cut$]
		{\Gamma \vdash N: A \\ \Delta, x:A \vdash M:C}
		{\Gamma, \Delta \vdash M[N/x]:C}
		\\
		\inferrule*[Right= $\multimap$R]
		{\Gamma, x:A \vdash M:B}
		{\Gamma \vdash \lambda x. M : A \multimap B} 
		\and
		\inferrule*[Right= $\multimap$L]
		{\Gamma \vdash N: A \\ \Delta, x:B \vdash M: C}
		{\Gamma, \Delta, y: A \multimap B \vdash M[yN/x]: C}
		\\
		\inferrule*[Right=$\forall$R]
		{\Gamma \vdash M: A\langle \gamma/\alpha \rangle 
		 \\ \gamma \not \in FV(\operatorname{rng}(\Gamma))}
		{\Gamma \vdash M: \forall \alpha.A} 
		\and
		\inferrule*[Right= $\forall$L]
		{\Gamma, x: A\langle B/ \alpha \rangle \vdash M:C}
		{\Gamma, x: \forall \alpha. A \vdash M:C}
	\end{mathpar}
	\caption{$\mathsf{IMLL}_2$ as a type-assignment system.}
	\label{figure:IMLL2} 
\end{figure}

\subsection{The systems $\mathsf{IMLL}_2$ and $\mathsf{IMLL}$} \label{subsec: IMLL2}
We assume familiarity with basic proof-theoretical notions and with Linear Logic (see~\cite{girard1987linear,troelstra2000basic}.)
\emph{Second-order Intuitionistic Multiplicative Linear Logic} ($\mathsf{IMLL}_2$), seen as a type-assignment for the linear $
 \lambda $-calculus, is in Figure~\ref{figure:IMLL2}, where, we remark, the only logical operators are the universal quantifier ``$ \forall$'' and the
	linear implication ``$ \multimap $''. $\mathsf{IMLL}_2$ derives \emph{judgments} $\Gamma \vdash M: A$, i.e.~a \emph{type} $ A $ for the linear $ \lambda $-term $ M $ from the \emph{context} $\Gamma$.
%
%
A \textit{type} is a (type) variable $\alpha$, or an \emph{implication} $A \multimap B$, or a \emph{universal quantification} $\forall \alpha. A$, where
$A$ and $B$ are types. 
The set of free type variables of $ A $ is $ FV(A) $. If $FV(A)= \emptyset$, then $ A $ is \textit{closed}. If $FV(A)= \lbrace \alpha_1, \ldots, \alpha_n \rbrace$, then  \textit{a closure} $\overline{A}$ of $A$ is $\forall \alpha_1. \cdots. \forall \alpha_n. A$, not necessarily linked to a specific order of $ \alpha_1,\ldots,\alpha_n $.
The standard meta-level substitution of a type $ B $ for every free occurrence of $ \alpha $ in $ A $ is $A \langle B/ \alpha \rangle$. 
The \textit{size} $\vert A \vert$ of the type $A$ is the number of nodes in its syntax tree. 
A \textit{context} $\Gamma$ has form $x_1: A_1, \ldots, x_n:A_n$, with $ n\geq 0 $, i.e.~it is a finite multiset of \textit{assumptions} $x: A$, where $ x $ is a
 $ \lambda $-variable. 
The \emph{domain} $\operatorname{dom}(\Gamma)$ of $ \Gamma $ is 
$\{x_1, \ldots, x_n\}$ and its range $ \operatorname{rng}(\Gamma) $
is $\{A_1, \ldots, A_n\}$.
The size $\vert \Gamma \vert$ of $ \Gamma $ is $\sum^n_{i=1}\vert A_i \vert$.  
Typically, names for contexts are $\Gamma, \Delta$ or $\Sigma$.

Since $ \textsf{IMLL}_2 $ gives types to linear $ \lambda $-terms,  
$\multimap$L is necessarily subject to the \emph{linearity constraint} 
$\operatorname{dom}(\Gamma) \cap \operatorname{dom}(\Delta)= \emptyset$. 
We range over the derivations of $\mathsf{IMLL}_2$ by $\mathcal{D}$. 
The \textit{size} $\vert \mathcal{D} \vert$ of $\mathcal{D}$ is the number of the rule instances that $\mathcal{D}$ contains. 
We say that $\Gamma \vdash M: B$ is \textit{derivable} if a derivation $\mathcal{D}$ exists that concludes with the judgment $ \Gamma \vdash M:B$, and we also say that $\mathcal{D}$ is a derivation of $ \Gamma \vdash M:B$. In that case we write $\mathcal{D}\triangleleft \Gamma \vdash M: B$ saying that 
$M$ is an \textit{inhabitant} of $B$ or that $B$ is \textit{inhabited} by $M$ 
from $\Gamma$. 
The cut-elimination steps for $\mathsf{IMLL}_2$ are standard and both
cut-elimination and confluence hold for it \cite{troelstra2000basic}. 

\textit{Propositional Intuitionistic Multiplicative Linear Logic} ($\mathsf{IMLL}$) is $\mathsf{IMLL}_2$ without $\forall$R and $\forall$L. 
From Hindley \cite{hindley1989bck}, we recall that $\mathsf{IMLL}$, 
thus $\mathsf{IMLL}_2$, gives a type to every linear $\lambda$-term. 
The converse holds as well, due to the above \textit{linearity constraint} 
on $\multimap$L, so the class of linear $\lambda$-terms is exactly the one 
of all typable $ \lambda $-terms in $\mathsf{IMLL}_2$. 
It follows that second-order does not allow to type more terms but it is nevertheless useful to assign uniform types to structurally related 
$ \lambda $-terms.

We conclude by recalling standard definitions of types in $\mathsf{IMLL}_2$:

\begin{defn}[Basic datatypes]
	\label{eqn:datatype unity and product}
The \emph{unity} type is $ \textbf{1} \triangleq \forall \alpha. (\alpha \multimap \alpha) $ with constructor $ I \triangleq \lambda x.x $, i.e.~the identity,
	and destructor $ \mathtt{let}\ M \mathtt{\ be \ }I \mathtt{\ in \ }N \triangleq MN $;

The \emph{tensor product} type
	$ A \otimes B \triangleq 
	\forall \alpha. (A \multimap B \multimap \alpha) \multimap \alpha $ 
	with constructor 
	$  \langle M, N \rangle \triangleq \lambda z. z\,M\,N $
	and destructor $ \mathtt{let}\ M \mathtt{\ be \ }x,y \mathtt{\ in \ }N
	\triangleq M(\lambda x. \lambda y. N) $.
	
Both binary tensor product and pair extend to their obvious $n$-ary versions
$A^n= \underbrace{A \otimes \ldots\otimes A}_{n}$ and $M^n \triangleq  \langle \underbrace{ M, \ldots, M}_{n}\rangle$.
\qed
\end{defn}
\begin{rem}
	Every occurrence of unity, ($ n $-ary) tensor and $ n $-tuple in the coming sections will be taken from Definition~\ref{eqn:datatype unity and product}.
\end{rem}
\noindent	
Finally, Definition~\ref{eqn:datatype unity and product} talks about \emph{datatypes} because, by introducing a specific syntax for constructors and destructors, we implicitly adopt a pattern matching mechanism
to operate on $ \lambda $-terms typed with those types.
\section{Duplication and erasure for the linear $\lambda$-calculus}
\label{sec: duplication and erasure in a linear setting}
As a motivational background we discuss  erasure and duplication in the linear $\lambda$-calculus both in an untyped and in a type-assignment setting.

\subsection{The untyped setting}
The linear $\lambda$-calculus forbids any form of \emph{direct} duplication 
of $ \lambda $-terms, by means of multiple occurrences of the same variable, 
or of erasure, by omitting occurrences of bound variables in a $ \lambda $-term. 
Nevertheless, erasure and duplication can be simulated. 
Concerning the former, a first approach has been developed by Klop \cite{klop1980combinatory}, and can be called   \enquote{erasure by garbage collection}.  It consists on   accumulating unwanted data during computation in place of erasing it. 
For example, $K'=\lambda xy.\langle x,y \rangle$ represents the classical 
$K= \lambda xy.x$, the second component of $\langle x,y \rangle$ 
being garbage. Another approach is by Mackie,  and can be called 
\enquote{erasure by data consumption} \cite{Mackie2018}. It involves a step-wise erasure 
process that proceeds by $\beta$-reduction, according to the following definition: 

\begin{defn}[Erasability]
\label{defn: erasure} 
A linear $\lambda$-term $M$ is \emph{erasable} 
if $\mathcal{C}[M]\rightarrow_\beta^* I$, for some  context $\mathcal{C}[]$ 
such that $ \mathcal{C}[M] $ is linear.
\qed 
\end{defn}
\begin{exmp}
The context $\mathcal{C}[]=(\lambda z.[])III$ erases 
$\lambda xy.zxy$ because, filling 
$ [] $ by $\lambda xy.zxy$, we obtain a closed linear $\lambda$-term that reduces to $I$.
\qed
\end{exmp}
\noindent
In~\cite{mairsonlinear}, Mackie proves that all closed linear $\lambda$-terms can be erased by means of very simple contexts.
\begin{lem}[\cite{mairsonlinear}]
\label{lem: linear terms are solvable} 
Let $ M $ be any closed linear $\lambda$-term. Then there exists $n \geq 0$ such that $M I\overset{n }{\ldots}I\rightarrow_\beta^* I$.
\end{lem}
\noindent
The above result is closely related to solvability (see~\cite{barendregt1984lambda}): 
``\textit{A $\lambda$-term $M$ in the standard $\lambda$-calculus is said \textit{solvable} if, for some $n$, there exist $\lambda$-terms 
$N_1, \ldots, N_n$ such that $M N_1 \ldots N_n =_\beta I$.}''
Lemma~\ref{lem: linear terms are solvable} states that every closed linear $\lambda$-term is solvable by linear contexts. 
\par 
In fact, the notion of erasability can be addressed in a more general setting.
\begin{defn}[Erasable sets]\label{defn: erasable sets} Let $X$ be a set of  linear $\lambda$-terms. We say that $X$ is an \textit{erasable set} if a linear $\lambda$-term $\mathtt{E}_X$ exists such that $\mathtt{E}_X \, M \rightarrow^*_{\beta \eta } I$, for all $M \in X$. We call $\mathtt{E}_X$ \textit{eraser} of $X$.
\end{defn}
\begin{prop} \label{prop: erasable if and only if closed} A finite set  $X$ of  linear $\lambda$-terms is  erasable  if and only if all terms in $X$ are closed. 
\end{prop}
\begin{proof}
Let $X$ be a finite set of linear $\lambda$-terms. To prove the left-to-right direction,  suppose  $X$ is erasable.   By definition, there exists a linear $\lambda$-term $\mathtt{E}_X$ such that $\mathtt{E}_X \, M \rightarrow_{\beta\eta}^* I$, for all $M \in X$. Since $ I $ is closed, by Fact~\ref{fact:Stability} each $ M\in X $ must be closed too. Let us now suppose that all terms in $X$ are closed, and let $M_1, \ldots, M_n$ be such terms.  By Lemma~\ref{lem: linear terms are solvable}, for every $i \leq n$ there exists a $k_i\geq 0$ such that $M_i I\overset{k_i }{\ldots}I\rightarrow_\beta^* I$.  It suffices to  set $\mathtt{E}_{X}\triangleq \lambda x. x I\overset{k }{\ldots}I $, where $k= \max_{i=1}^n k_i$. 
\end{proof}
\noindent
Recall from Definition~\ref{eqn:datatype unity and product} that  $\langle M, N \rangle\triangleq \lambda z.z MN$. In the same spirit of Definition~\ref{defn: erasable sets}, we now investigate duplicability in the linear $\lambda$-calculus. 
\begin{defn}[Duplicable sets]\label{defn: duplicable sets} Let $X$ be a  set of   linear $\lambda$-terms. We say that $X$ is a \textit{duplicable set} if a linear $\lambda$-term $\mathtt{D}_X$ exists such that $\mathtt{D}_X \, M \rightarrow^*_{\beta \eta} \langle M,  M \rangle$ and $FV(\mathtt{D}_X)\cap FV(M) =\emptyset$, for all $M \in X$. We call $\mathtt{D}_X$ \textit{duplicator} of $X$.
\end{defn}
\begin{prop}\label{prop: duplicable implies closed} If a finite set  $X$ of  linear $\lambda$-terms is  duplicable  then all terms in $X$ are closed. 
\end{prop}
\begin{proof}
Let $X$ be a finite set of  linear $\lambda$-terms, and suppose  $X$ is  duplicable.  By definition,  there exists a linear $\lambda$-term $\mathtt{D}_X$ such that $\mathtt{D}_X \, M \rightarrow_{\beta\eta}^* \langle M,  M  \rangle$, for all $M \in X$. Since both $M$ and  $\mathtt{D}_{X}$ are linear $\lambda$-terms, and $FV(\mathtt{D}_X)\cap FV(M) =\emptyset$, we have that $\mathtt{D}_X\, M$ is linear, for all  $M \in X$. If there were a variable occurring free in a term $M\in X$, then it would occur twice in $\langle M, M \rangle$, contradicting  Fact~\ref{fact:Stability}.
\end{proof}
\noindent
We conjecture that the converse holds as well, as long as we restrict to sets of distincts $\beta\eta$-normal forms. Indeed, duplication in a linear setting ultimately relies on the following linear version of the  general separation theorem for the standard $\lambda$-calculus proved by Coppo et al.~\cite{coppo1978semi}:
\begin{conj}[General separation] \label{conj: linear separation} Let $X=\lbrace M_1, \ldots, M_n \rbrace$ be a set of distinct closed linear $\lambda$-terms in  $\beta\eta$-normal form. Then, for all  $N_1, \ldots, N_n$  closed linear $\lambda$-terms, there exists a closed linear $\lambda$-term $F$ such that $F\, M_i =_{\beta \eta}N_i$, $\forall i \leq n$.
\end{conj}
\noindent
Now, let $X= \lbrace M_1, \ldots, M_n \rbrace$ be a finite set of distinct closed  linear $\lambda$-terms in $\beta\eta$-normal form.  If Conjecture~\ref{conj: linear separation} were true,  by fixing  $N_i\triangleq \langle M_i, M_i \rangle$ for all $i \leq  n$, there would exists a closed  linear $\lambda$-term  $\mathtt{D}_X$ such that $\mathtt{D}_X \, M_i \rightarrow_{\beta \eta}^* \langle M_i, M_i\rangle$. 
So, we could connect linear erasure and duplication to standard  $\lambda$-calculus notions:
\begin{equation*}
\begin{split}
\textit{solvability} & \textrm{ implies } \textit{linear erasability}\\
\textit{separation} & \textrm{ implies } \textit{linear duplication} 
\enspace .
\end{split}
\end{equation*}
This topic is left to future work (see Section~\ref{section:Conclusions}).

\subsection{The typed setting}
Erasure and duplication are less direct and liberal in 
$ \textsf{IMLL}_2 $ which assigns types to linear $ \lambda $-terms. 
Specifically, it is possible to erase or duplicate all values (Definition~\ref{defn:Values among linear lambda-terms}) of what we call  \enquote{ground type}, i.e.~Mairson\&Terui's  notion of closed $\Pi_1$-type~\cite{mairson2003computational},  whose formal definition will be recalled shortly. A typical example of  ground type in $\mathsf{IMLL}_2$  is the one representing  booleans.
The standard second-order intuitionistic formulation of booleans (i.e.~$\forall \alpha .\alpha \multimap \alpha \multimap \alpha$)   is 
meaningless for $\mathsf{IMLL}_2$ due to the lack of free weakening. 
Mairson\&Terui \citep{mairson2003computational} define them as:
\begin{align}
&\mathbf{B}\triangleq\forall \alpha. \alpha \multimap \alpha \multimap \alpha \otimes \alpha 
\label{eqn: boolean data type}
&&\mathtt{tt}\triangleq	\lambda x. \lambda y. \langle x,y \rangle 
&&\mathtt{ff}\triangleq \lambda x. \lambda y. \langle y,x \rangle 
\end{align} 
where the values \enquote{truth} \texttt{tt}
and the \enquote{falsity} \texttt{ff} implement the
\enquote{erasure by garbage collection}: the first element of the pair is the \enquote{real} output, while the second one is garbage.
Starting from~(\ref{eqn: boolean data type}), Mairson shows in~\cite{mairsonlinear} that $\mathsf{IMLL}$ is expressive enough to encode boolean functions. Mairson and Terui  reformulate that encoding 
in $\mathsf{IMLL}_2$ in order to prove 
results about the complexity of cut-elimination 
\cite{mairson2003computational}. The advantage of $\mathsf{IMLL}_2$ is to assign uniform types to the $\lambda$-terms representing boolean functions. An \emph{eraser} $\mathtt{E}_{\mathbf{B}}$ and a \emph{duplicator} $\mathtt{D}_{\mathbf{B}}$ are the keys to obtain the encoding:
\begin{align}
\label{eqn: erasure booleans}
\mathtt{E}_{\mathbf{B}}& \triangleq
\lambda z. \mathtt{let\ }zII \mathtt{\ be \ } x,y \mathtt{ \ in \ }(\mathtt{let \ }y \mathtt{ \ be \ } I \mathtt{ \ in \ }x) :\mathbf{B}\multimap \mathbf{1}
\\
\label{eqn: duplication booleans}
\mathtt{D}_{\mathbf{B}}& \triangleq
\lambda z. \pi_1(z\langle \mathtt{tt}, \mathtt{tt}\rangle\langle \mathtt{ff}, \mathtt{ff}\rangle) : \mathbf{B}\multimap \mathbf{B}\otimes \mathbf{B}
\\
\label{eqn: boolean projection}
\pi_1 &\triangleq
\lambda z. \mathtt{let\ }z \mathtt{\ be \ }x,y \mathtt{\ in\ }(\mathtt{let \ }\mathtt{E}_{\mathbf{B}} \, y \mathtt{\ be \ }I \mathtt{\ in \ }x) :
(\mathbf{B}\otimes \mathbf{B}) \multimap\mathbf{B}
\end{align}
with $ \mathbf{B} $ as in \eqref{eqn: boolean data type} and $\pi_1$ the linear $\lambda$-term projecting the first element of a pair.\\
 Switching to type-assignment setting we get uniform copying and erasing mechanisms of the whole \textit{class} of values of a given ground type. Note that the type-theoretical constraints let the erasure of a typed linear $\lambda$-term make use of something more than mere stacks of identities as in Lemma~\ref{lem: linear terms are solvable}. Also, note that \textit{both} the possible results $\langle \mathtt{tt}, \mathtt{tt} \rangle$ 
and $\langle \mathtt{ff}, \mathtt{ff} \rangle$ of duplications are built-in 
components of  $ \texttt{D}_{\textbf{B}} $.
In accordance with the given input, $\mathtt{D}_{\mathbf{B}}$  selects 
the right pair representing the result by \textit{erasing} the 
unwanted one. Such a \enquote{linear} form of \emph{duplication by selection and erasure} is a step-by-step elimination of useless data until the desired result shows up. 
\par
The analysis of~(\ref{eqn: erasure booleans}) and~(\ref{eqn: duplication booleans}) leads to the following formal notions: 
\begin{defn} [Duplicable and erasable types in $\mathsf{IMLL}_2$]  
\label{defn: duplicable and erasable types} 
Let $A$  be a type in $\mathsf{IMLL}_2$. It is a \textit{duplicable type}  
if a linear $\lambda$-term $\mathtt{D}_A: A \multimap A \otimes A$ exists 
such that $\mathtt{D}_{A} \,V \rightarrow_{\beta \eta}^* 
\langle V, V \rangle$, for every value $V$ of $A$.

Moreover, $A$ is an \textit{erasable type} if a linear $\lambda$-term $\mathtt{E}_A:A \multimap \mathbf{1}$ exists such that
$\mathtt{E}_{A} \, V \rightarrow_{\beta \eta}^* I$,
for every value  $V$ of $A$.

We call $\mathtt{D}_A$ \emph{duplicator} of $A$ and $\mathtt{E}_{A}$ its \emph{eraser}. 
\qed
\end{defn}
\noindent
Duplicators and erasers in Definition~\ref{defn: duplicable and erasable types} apply to values of a given type, i.e.~\textit{closed} and normal inhabitants. This is not a loss of generality because Proposition~\ref{prop: erasable if and only if closed} and Proposition~\ref{prop: duplicable implies closed} say that  only closed terms can be duplicated or erased linearly.

\begin{defn}[$\Pi_1$, $\Sigma_1$-types \cite{mairson2003computational} and Ground types]
\label{defn:Pi and Sigma types} 
The following mutually defined grammars generate $\Pi_1$ and  $\Sigma_1$-types:
\begin{align*}
\Pi_1&:= \alpha  \ \vert \ \Sigma_1      
\multimap \Pi_1 \ \vert \ \forall \alpha. \Pi_1  \\
\Sigma_1&:= \alpha \ \vert \ \Pi_1 
\multimap \Sigma_1 
\end{align*} 
We call \emph{ground types} the \emph{closed} $\Pi_1$-types.
\qed
\end{defn}
\noindent
We note that the universal quantifier $\forall$ occurs only positively 
in a $\Pi_1$-type, hence in ground types.
\par
The booleans $\mathbf{B}$ in~\eqref{eqn: boolean data type}, 
the unit $\mathbf{1}$ and the tensor $A\otimes B$ as in Definition~\ref{eqn:datatype unity and product} are ground types, if $A$ and $B$ are.
In fact, following \cite{mairson2003computational}, tensors and units can 
occur \emph{also} to the left-hand side of a linear implication 
\enquote{$ \multimap $}, even in negative positions. The reason is that 
we can ignore them in practice, thanks to the isomorphisms:
\begin{align*}
((A \otimes B) \multimap C) \multimapboth (A \multimap B \multimap C) 
&&
(\mathbf{1}\multimap C) \multimapboth C
\enspace .
\end{align*} 
Ground types represent \emph{finite} data types, while the values with a ground type represent their data.
\begin{fact}\label{rem: every normal term has Pi1 type} 
Every closed linear $\lambda$-term $ M $ has a ground type.
\qed
\end{fact}
\begin{proof}
Every closed  linear $\lambda$-term $M$  is typable  in $\mathsf{IMLL}$ (see~\cite{hindley1989bck}). Types in $\mathsf{IMLL}$ are  quantifier-free instances
of  $\Pi_1$-types. Hence, $M$ has also a $\Pi_1$-type $A$ in $\mathsf{IMLL}_2$. 
Let $FV(A)= \lbrace\alpha_1, \ldots, \alpha_n\rbrace $. Since $M$ inhabits $A$, it also inhabits 
$\overline{A} =\forall \alpha_1.\cdots.\forall \alpha_n. A$, which  is a 
closed $\Pi_1$-type, i.e.~a ground type in $\mathsf{IMLL}_2$.
\end{proof}
\noindent
The class of ground types is  a subset of both the classes of duplicable and erasable types. 
\begin{thm}[\cite{mairson2003computational}]
\label{thm: pi1 types are erasable}
Every ground type is erasable.
\end{thm}
\begin{proof} 
The proof follows from proving two statements
by simultaneous induction: (i)  For every $\Pi_1$-type $A$ with free type variables $\alpha_1, \ldots, \alpha_n$ there exists 
a linear $ \lambda $-term $\mathtt{E}_A$  such 
that $ \vdash \mathtt{E}_A: A[\mathbf{1}/\alpha_1, \ldots,\mathbf{1} /\alpha_n] $ $\multimap \mathbf{1}$, and (ii) for every $\Sigma_1$-type $A$ with free type variables $\alpha_1, \ldots, \alpha_n$ there exists a linear $ \lambda $-term $\mathtt{H}_A$ such that 
$ \vdash \mathtt{H}_A:A[\mathbf{1}/\alpha_1, \ldots,\mathbf{1}/ \alpha_n]$.
\end{proof}

\begin{thm}
\label{thm: pi1 are duplicable}
Every inhabited ground type is duplicable.
\end{thm}
\noindent
Mairson and Terui sketch the proof of Theorem~\ref{thm: pi1 are duplicable} in
\cite{mairson2003computational}.
\ref{sec: the d-soundness theorem DICE}, which we see as an integral and relevant part of this work, develops it in every detail. 

\section{The system $\mathsf{LEM}$}\label{sec:  the system shpos, with, oplus  IMLL2 cut elimination and bound}
Theorems~\ref{thm: pi1 types are erasable} and~\ref{thm: pi1 are duplicable}  say that the ground types can be weakened and contracted in $\mathsf{IMLL}_2$. We here logically internalize those kinds of weakening a contraction in the deductive system $\mathsf{LEM}$ (Linearly Exponential and Multiplicative). It extends $\mathsf{IMLL}_2$ 
with inference rules for the modality 
``$\shpos$''  that closely recall the exponential rules in Linear Logic.
\begin{defn}[Types of $\mathsf{LEM}$]
\label{defn:Types for IMLL2shpos}
Let $\mathcal{X}$ be a denumerable set of type variables. The following grammar~\eqref{eqn: grammar mathcalAe} generates the \textit{exponential types}, while the grammar~\eqref{eqn: grammar A} generates the \textit{linear types}:
\begin{align}
\sigma := &\ A \ \vert \  \shpos \sigma \label{eqn: grammar mathcalAe}\\
A:= &\ \alpha \ \vert \ \sigma \multimap A \ \vert \ \forall \alpha.A  \label{eqn: grammar A} 
\end{align}
where $\alpha \in \mathcal{X}$ and,  in the last clause of the grammar~\eqref{eqn: grammar mathcalAe}, i.e.~the one introducing $ \shpos \sigma $, \emph{the type $ \sigma$ must be closed and without negative occurrences of $\forall$}. The set of all types generated by the grammar~(\ref{eqn: grammar mathcalAe})   will be denoted $\Theta_\shpos$.   A type is \textit{strictly exponential} if it is of the form $\shpos \sigma$. A \textit{strictly exponential context} is a context containing only  strictly exponential types and, similarly, a  \textit{linear context} contains only linear types. Finally, $A\langle B/\alpha \rangle$ is the standard meta-level substitution of $B$, for every occurrence of $ \alpha$ in $A$. 
\qed

%
\end{defn}
\noindent
\begin{rem}
The modality ``$\shpos$'' identifies where the ground types  (Definition~\ref{defn:Pi and Sigma types}) occur in the 
grammars~(\ref{eqn: grammar mathcalAe}) and~(\ref{eqn: grammar A})
 because it applies to closed types that are free from negative occurrences  of $\forall$. 
So, the occurrences of $\shpos \sigma$ identify where contraction and weakening rules can apply in the derivations of $\mathsf{LEM}$.
\qed
\end{rem}
\noindent
Also, we observe that syntactically replacing the Linear Logic modality \enquote{$\oc$} for 
\enquote{$\shpos$} in~(\ref{eqn: grammar mathcalAe}) and~(\ref{eqn: grammar A}) yields  a subset of Gaboardi\&Ronchi's \textit{essential types}~\cite{gaboardi2009light}, introduced to prove Subject reduction in a variant of Soft Linear Logic. Essential types forbid the occurrences of modalities in the right-hand side of an implication, such as in $A \multimap \oc B$.

%

$\mathsf{LEM}$ will be defined as  a type-assignment for the term calculus $\Lambda_\shpos$, that is essentially the standard linear $\lambda$-calculus with  explicit and type-dependent constructs for erasure and  duplication, i.e.~$\mathtt{discard}_\sigma$ and $\mathtt{copy}_\sigma ^V$, the latter being also decorated with a value $V$. These new constructs are able to copy and discard values only, i.e.~closed and normal linear $\lambda$-terms. 

\begin{defn}[Terms and reduction of $\mathsf{LEM} $] 
\label{defn:lambda-terms and values of IMLLshpos2} 
Let $\mathcal{V}$ be a denumerable set of variables. The \textit{terms} of $\mathsf{LEM}$ are given by the grammar:
\begin{align}
M, N := \ & x\ 
     \vert\ \lambda x.M\ 
     \vert\ MN\ 
     \vert\ \mathtt{discard}_{\sigma}\, M \mathtt{\ in\ } N\ 
     \vert\ \mathtt{copy}^{V}_{\sigma}\, M \mathtt{\ as\ }x,y 
            \mathtt{\ in\ } N
&  \label{eqn: term assignment shposIMLL2}
\end{align} 
where $x, y\in \mathcal{V}$, $ V $ is a value (Definition~\ref{defn:Values among linear lambda-terms}, Section~\ref{sec: background}), and $ \sigma \in \Theta_\shpos$. The set of all terms of $\mathsf{LEM}$  will be denoted $\Lambda_{\shpos}$.
The set of the free variables of a term, and the notion of size are standard for variables, abstractions, and applications. 
The extension to the new constructors are:
\begin{align}
\nonumber
FV(\mathtt{  discard }_{\sigma}\, M\mathtt{ \   in \ }N)= 
& \ FV(M)\cup FV(N)
\\
\nonumber
FV(\mathtt{  copy}^{V}_{\sigma}\, M \mathtt{ \   as \ }x, y \mathtt{ \   in \ }N)= & \ FV(M) \cup (FV(N) \setminus \lbrace x, y \rbrace)
\\
\nonumber
\vert \mathtt{  discard }_{\sigma}\, M\mathtt{ \   in \ }N\vert = & \ \vert M\vert + \vert N\vert +1 
\\
\label{align: size of copy}
\vert \mathtt{  copy}^{V}_{\sigma}\, M \mathtt{ \   as \ }x, y \mathtt{ \   in \ }N\vert = & \ \vert M\vert  + \vert N \vert + \vert V \vert + 1 \enspace .
\end{align}
A term $M$ in~\eqref{eqn: term assignment shposIMLL2} is \textit{linear} 
if all of its free variables occur once in it and every proper sub-term 
of $M$ with form $ \lambda x.N $  (resp.~$ \mathtt{copy}^{V}_{\sigma}\, M' \mathtt{\ as\ }y,z \mathtt{\ in\ } N $) is such that $ x$ (resp.~$y,z$)  must occur in $ N $ and $ N $ is linear. \textit{Henceforth, we use linear terms only.}
 
The notions of meta-level substitution and context are as usual.

 The \textit{one-step reduction relation} $\rightarrow$ is a binary relation on terms. It is defined by the reduction rules in Figure~\ref{fig: term reduction rules} and by the commuting conversions in Figure~\ref{fig: term commuting conversions}. It applies in any context. Its reflexive and transitive closure is denoted $\rightarrow^*$. A term is said a (or is in) \textit{normal form} if no reduction step applies to it. 
\qed
\end{defn}
\begin{figure}[t]
\centering
\begin{equation}\nonumber
\begin{aligned}
\scalebox{0.7}{$  (\lambda x. M)N  $}&\scalebox{0.7}{$ \  \rightarrow  \ M[N/x] $}\\
\scalebox{0.7}{$  \mathtt{discard}_{\sigma}\, V \mathtt{\ in\ }M $} & \scalebox{0.7}{$ \  \rightarrow \ M $}\\
\scalebox{0.7}{$  \mathtt{copy}^{U}_{\sigma}\, V \mathtt{\ as\ } y, z \mathtt{\ in\ }M $} & \scalebox{0.7}{$ \  
	\rightarrow  \ M[V/y, V/z] $}
\end{aligned}
\end{equation}
\caption{Reduction on terms.}
\label{fig: term reduction rules}
\end{figure}

\begin{figure}[t]
	\centering
	\begin{equation}\nonumber
	\begin{aligned}
	\scalebox{0.7}{$(\mathtt{discard}_{\sigma}\, M \mathtt{\ in\ } N)P $}&\scalebox{0.7}{$ \ \rightarrow\ \mathtt{discard}_{\sigma}\, M  \mathtt{\ in\ }(NP) $}\\
    \scalebox{0.7}{$\mathtt{discard}_{\sigma}\, (\mathtt{discard}_{\tau}\, M \mathtt{\ in\ } N) \mathtt{\ in \ }P $}&\scalebox{0.7}{$\  \rightarrow \  \mathtt{discard}_{\tau}\, M  \mathtt{\ in\ }(\mathtt{discard}_{\sigma}\, N \mathtt{\ in\ }P)$}\\
	\scalebox{0.7}{$\mathtt{copy}^{V}_{\sigma}\, (\mathtt{discard}_{\tau}\, M \mathtt{\ in\ } N)\mathtt{\ as\ }y,z \mathtt{\ in\ }P $}& \scalebox{0.7}{$\ \rightarrow \  \mathtt{discard}_{\tau}\, M  \mathtt{\ in\ }(\mathtt{copy}^{V}_{\sigma}\, N \mathtt{\ as\ }y,z \mathtt{\ in\ }P)$}
	\\
	\scalebox{0.7}{$(\mathtt{copy}^{V}_{\sigma}\, M \mathtt{\ as\ }y,z \mathtt{\ in\ }N)P $}&\scalebox{0.7}{$\  \rightarrow \  \mathtt{copy}^{V}_{\sigma}\, M \mathtt{\ as\ }y,z  \mathtt{\ in\ }(NP) $}\\
	\scalebox{0.7}{$\mathtt{discard}_{\sigma}\, (\mathtt{copy}^{V}_{\tau}\, M \mathtt{\ as\ }y,z \mathtt{\ in\ }N) \mathtt{\ in \ }P $}&\scalebox{0.7}{$\ \rightarrow  \ \mathtt{copy}^{V}_{\tau}\, M \mathtt{\ as\ }y,z  \mathtt{\ in\ }(\mathtt{discard}_{\sigma}\, N \mathtt{\ in\ }P)$}\\
	\scalebox{0.7}{$\mathtt{copy}^{U}_{\sigma}\, (\mathtt{copy}^{V}_{\tau}\, M \mathtt{\ as\ }y,z \mathtt{\ in\ }N)\mathtt{\ as\ }y',z' \mathtt{\ in\ }P $}& \scalebox{0.7}{$\ \rightarrow \  \mathtt{copy}^{V}_{\tau}\, M \mathtt{\ as\ } y,z  \mathtt{\ in\ }(\mathtt{copy}^{U}_{\sigma}\, N \mathtt{\ as\ }y',z' \mathtt{\ in\ }P)$}
	\end{aligned}
	\end{equation}
	\caption{Commuting conversions on terms.}
	\label{fig: term commuting conversions}
\end{figure}
\noindent

\begin{figure}[htbp]
\begin{mathpar}
\scalebox{0.8}{$
\inferrule*[Right= $ax$]{\\}{x: A \vdash x:A}$} \and
\scalebox{0.8}{$\inferrule*[Right= $cut$]{\Gamma \vdash N: \sigma \\ \Delta, x: \sigma \vdash M:\tau}{\Gamma, \Delta \vdash M[N/x]: \tau}$} \\
\scalebox{0.8}{$\inferrule*[Right= $\multimap$R]{\Gamma, x:\sigma \vdash M:B}{\Gamma \vdash \lambda x. M : \sigma \multimap B }$} \and
\scalebox{0.8}{$\inferrule*[Right= $\multimap$L]{\Gamma \vdash N: \sigma \\ \Delta, x:B  \vdash M: \tau}{\Gamma, \Delta, y: \sigma \multimap B \vdash M[yN/x]: \tau}$}\\
\scalebox{0.8}{$\inferrule*[Right=$\forall$R]{\Gamma \vdash M: A \langle \gamma / \alpha 
\rangle\\ \gamma \not \in FV(\operatorname{rng}(\Gamma))}{\Gamma \vdash M: \forall \alpha.A}$} \and
\scalebox{0.8}{$\inferrule*[Right= $\forall$L]{ \Gamma, x: A\langle B/ \alpha \rangle  \vdash M:\tau}{\Gamma, x: \forall \alpha. A \vdash M:\tau}$}\\
\scalebox{0.8}{$\inferrule*[Right=$p$]{x_1:\shpos \sigma_1, \ldots, x_n:\shpos \sigma_n \vdash M:\sigma }{x_1: \shpos \sigma_1, \ldots, x_n: \shpos \sigma_n \vdash  M: \shpos \sigma }$} \and
\scalebox{0.8}{$\inferrule*[Right=$d$]{\Gamma, x: \sigma \vdash M:\tau}{\Gamma, y: \shpos  \sigma  \vdash M[y/x]:\tau}$}\\
\scalebox{0.8}{$\inferrule*[Right=$w$]{\Gamma \vdash M:\tau}{\Gamma, x: \shpos  \sigma \vdash \mathtt{ discard}_{\sigma}\,   x\mathtt{ \   in \ }M:\tau}$} \and 
\scalebox{0.8}{$\inferrule*[Right=$c$]{\Gamma, y: \shpos  \sigma, z: \shpos \sigma \vdash M:\tau \\ \vdash V: \sigma}{\Gamma, x: \shpos \sigma \vdash \mathtt{ copy} ^{V}_{\sigma}\, x \mathtt{ \   as \ } y,z \mathtt{ \   in \ }M:\tau}$}
\end{mathpar}
\caption{The system $\mathsf{LEM}$.}
\label{fig: the system shpos add IMLL2} 
\end{figure}
\noindent
Both the type and the term annotations in the constructs  $\mathtt{discard}_\sigma$  and $\mathtt{copy}^V_\sigma$ will become meaningful once we introduce the type-assignment system. The value $V$ will be an inhabitant of $\sigma$, a necessary condition in order to faithfully express the mechanism of linear duplication.

The structure of the types in Definition~\ref{defn:Types for IMLL2shpos} drives the definition of $ \mathsf{LEM} $.
\begin{defn}[The system $\mathsf{LEM} $]
\label{defn:The type assigment IMLLshpos2}
It is the type-assignment system for the term calculus $\Lambda_\shpos$ (Definition~\ref{defn:lambda-terms and values of IMLLshpos2}) in Figure~\ref{fig: the system shpos add IMLL2}. 
It extends $\mathsf{IMLL}_2$ with the rules \emph{promotion} $p$, \emph{dereliction} $d$, \emph{weakening} $w$ and \emph{contraction} $c$.
As usual, $\multimap$R, $\forall$R, and $p$ are right rules
while
$\multimap$L, $\forall$L, $d$, $w$, and $c$ are left ones.
\qed
\end{defn}
\noindent
First, we observe that $ax$ cannot introduce exponential types, like
in the \emph{essential types} of the type systems in \cite{gaboardi2009light}. This is the base for proving:
\begin{prop}[Exponential context from exponential conclusion]
\label{prop: exponential context} 
If $\mathcal{D} \triangleleft \Gamma \vdash M: \shpos \sigma$ is a derivation in 
$\mathsf{LEM}$, then $\Gamma$ is a strictly exponential context.
\end{prop}
\begin{proof}
	By structural induction on the derivation $\mathcal{D}$.
\end{proof}
\noindent
Also, we observe that $p$, $d$, $w$, $c$ of $\mathsf{LEM}$ are reminiscent of the namesake Linear Logic exponential rules, but they only apply to types $ \shpos\sigma $ that~\eqref{eqn: grammar mathcalAe}  (Definition~\ref{defn:Types for IMLL2shpos}) generates, i.e.~closed types with no negative occurrences of universal quantifiers.  
The rule $c$ has one premise more than the contraction rule in Linear Logic. This premise \enquote{witnesses} that the type $\sigma$ we want to contract is inhabited by at least one value $V$. This is because  $ c $ expresses the mechanism of linear contraction discussed in the previous section for $\mathsf{IMLL}_2$.
As we shall see in Theorem~\ref{thm: translations for LAML}, the term $\mathtt{copy}^V_\sigma$ that $c$ introduces is a (very) compact notation for duplicators of ground types, 
whose detailed description is in~\ref{sec: the d-soundness theorem DICE}. Roughly, the  duplicator of a ground type $A$ is a  linear $\lambda$-term that, taking a value $U$ of type $A$ as argument,  implements the following three main operations:
\begin{enumerate}
	\item \label{enumerate: duplication step 1}
	expand $ U $ to its $ \eta $-long normal form 
	$ U_{A} $ 
	whose dimension is bounded by the dimension of the type $A$. 
	This is why using a modality to identify types for which 
	this can be done is so important;
	
	\item \label{enumerate: duplication step 2}
    compile $ U_{A} $ to a linear $ \lambda $-term
    $ \lceil U_{A} \rceil $ which encodes $ U_{A} $ as 
    a boolean tuple;
	
	\item \label{enumerate: duplication step 3}
    copy and decode $ \lceil U_{A} \rceil $, obtaining the duplication 
    $ \langle U_{A}, U_{A} \rangle $ of $ U_{A} $ as final result.  This duplication 
    is by means of the term $\mathtt{dec}^s_{A}$ in~\ref{subsection: dec and DBn}.   It nests a series of  \texttt{if}-\texttt{then}-\texttt{else} constructs which is a look-up table, possibly quite big, that stores all the pairs of normal inhabitants of $A$.
    Each of them represents a possible outcome of the duplication.
    Given a boolean tuple $ \lceil U_{A} \rceil $  in input, the nested \texttt{if}-\texttt{then}-\texttt{else} select the corresponding pair 
    $ \langle U_{A}, U_{A} \rangle $, erasing all the remaining \enquote{candidates}. The inhabitation condition for $A$ stated in Theorem~\ref{thm: pi1 are duplicable} ensures that the default pair $ \langle V, V \rangle $ exists as a sort of \enquote{exception}. We ``throw'' it each time the boolean tuple that $\mathtt{dec}^s_{A}$ receives as input does not encode any term.
\end{enumerate}
\noindent
Point~\ref{enumerate: duplication step 3} is the one implementing Mairson\&Terui's \enquote{\emph{duplication by selection and erasure}} discussed in Section~\ref{sec: duplication and erasure in a linear setting}. It involves the component of the duplicator which is exponential in the size of $A$.  Therefore, as we shall see in Theorem~\ref{thm: exponential compression},  the construct $\mathtt{copy}_\sigma^V$ exponentially compresses the linear duplication mechanism encoded in a duplicator.

We conclude this section by commenting about how \enquote{$ \shpos $} and \textsf{LL}'s \enquote{$\oc$} differ. Intuitively, the latter allows to duplicate or erase logical structure, or terms, \emph{at once}, which is the standard way to computationally interpret contraction and weakening of a logical system.
The modality \enquote{$ \shpos $} identifies duplication and erasure processes with a more \emph{constructive} nature. The duplication proceeds step-by-step among a whole set of possible choices in order to identify those ones that cannot contain the copies of the term we are interested to duplicate, until it eventually reaches what it searches. Then, it exploits erasure. Erasing means eroding step-by-step a derivation or a term, according to the type that drives its construction.

\section{Basic computational properties of $\mathsf{LEM}$}
\label{sec: basic computational properties of IMLLshpos}
The observations in the previous sections lead us to set the reduction rules on terms, which allow duplication and erasure for values only, as in Figure~\ref{fig: term reduction rules}.
Those reductions are more restrictive than the cut-elimination
steps that we could perform on $\mathsf{LEM}$ if we look at
it as it was a pure logical system, i.e.~not a type-assignment. Since the cut-elimination of $\mathsf{LEM}$
works as in $\mathsf{LL}$, once replaced \enquote{!} for \enquote{$ \shpos $}, we can observe the effect of moving the $ cut $ in:
\begin{equation}\label{eqn: non linear duplication}
\AxiomC{\vdots}
\noLine
\UnaryInfC{$y: \shpos A, z: \shpos \mathbf{1} \vdash yz:  \mathbf{1}$}
\RightLabel{$p$}
\UnaryInfC{$y: \shpos A, z: \shpos \mathbf{1} \vdash yz: \shpos  \mathbf{1}$}
\AxiomC{$\vdots$}
\noLine
\UnaryInfC{$\Gamma, x_1: \shpos \mathbf{1}, x_2: \shpos \mathbf{1} \vdash M:\ \tau$}
\AxiomC{\vdots}
\noLine
\UnaryInfC{$\vdash I: \mathbf{1}$}
\RightLabel{$c$}
\BinaryInfC{$\Gamma, x: \shpos \mathbf{1} \vdash \mathtt{ copy} ^{I}_{\mathbf{1}}\, x \mathtt{ \   as \ } y,z \mathtt{ \   in \ }M: \tau$}
\RightLabel{$cut$}
\BinaryInfC{$\Gamma, y: \shpos A, z: \shpos \mathbf{1} \vdash \mathtt{ copy} ^{I}_{\mathbf{1}}\, yz \mathtt{ \   as \ } y,z \mathtt{ \   in \ }M: \tau$}
\DisplayProof
\end{equation}
upward, in order to eventually eliminate it.
The move would require to duplicate the \emph{open} term $yz$, erroneously 
yielding a non linear term. So, at the proof-theoretical level, moves 
of the cut rule exist that cannot correspond to any reduction on terms. 
In order to circumvent the here above misalignment, we proceed as follows:
\begin{itemize}
	\item 
	We define the \emph{lazy cut-elimination steps}. Their introduction rules out any attempt to eliminate instances 
	of cuts like~\eqref{eqn: non linear duplication}.
	The apparent drawback is to transform cuts like~\eqref{eqn: non linear duplication} into \emph{deadlocks}, i.e. into instances of $ cut $
	that we cannot eliminate.
	
	\item
	Deadlocks are not a problem. Once defined \emph{lazy types},
	we can show that a \emph{lazy cut-elimination strategy} exists
		such that it eliminates all the cut rules that may
		occur in a derivation of a lazy type. 
		The cost of the elimination is cubical
		(Theorem~\ref{thm: cut elimination for downarrow IMLL2}).
\end{itemize}
Last, we show that the reduction on terms in Figure~\ref{fig: term reduction rules} and Figure~\ref{fig: term commuting conversions} enjoy Theorem~\ref{thm:Subject Reduction for IMLL2shpos}, i.e.~Subject reduction.

\subsection{Cut-elimination and  its cubical complexity}
\label{Cut elimination, complexity, and subject reduction for IMLL2^shpos}
\begin{defn}[The cuts of $\mathsf{LEM} $]
\label{defn: definitions for cut} 
Let $ (X,Y) $ identify an instance:
\begin{equation}\label{eqn: cut elimination format}
\AxiomC{\vdots}
\RightLabel{$X$}
\UnaryInfC{$\Gamma \vdash N:\sigma$}
\AxiomC{\vdots}
\RightLabel{$Y$}
\UnaryInfC{$\Delta, x:\sigma \vdash M: \tau$}
\RightLabel{$ cut $}
\BinaryInfC{$\Gamma,\Delta \vdash M[N/x]: \tau$}
\DisplayProof
\end{equation}
of the rule $ cut $ that occurs in a given derivation 
$ \mathcal{D} $ of 
$\mathsf{LEM}$, where $X$ and $Y$ are two of the rules in Figure~\ref{fig: the system shpos add IMLL2}. 
\emph{Axiom cuts} involve $ ax $, and are of the form $(X,ax)$ or $(ax,Y)$, for some $ X$ and $ Y $. 
\emph{Exponential cuts} are ($p$,$d$), ($p$,$w$), and ($p$,$c$). 
\emph{Principal cuts} are $(\multimap$R$,\multimap$L$)$,
$(\forall$R$,\forall$L$)$ and every exponential cut.
\emph{Symmetric cuts} contain axiom and principal cuts.
Every symmetric cut that is not exponential is \emph{multiplicative}. 
\emph{Commuting cuts} are all the remaining instances of $ cut $, not mentioned here above, $(p, p)$ included, for example.
\par
A \textit{lazy cut} is every instance of the cut~\eqref{eqn: cut elimination format} which is both exponential and such that 
$N$ is a value.
\par
A \emph{deadlock} is every instance of the cut~\eqref{eqn: cut elimination format} which is both exponential and such that $\Gamma\neq \emptyset$. Otherwise, it is \emph{safe}.
\qed
\end{defn}


\begin{figure}[t]
\begin{equation*}
\scalebox{0.6}{$
\begin{aligned}
\AxiomC{$\mathcal{D}$}
	\noLine
	\UnaryInfC{$\Gamma, x:\sigma \vdash M: A$}
	\RightLabel{$\multimap R$}
	\UnaryInfC{$\Gamma \vdash \lambda x.M: \sigma\multimap A$}
	\AxiomC{$\mathcal{D}'$}
	\noLine
	\UnaryInfC{$\Delta\vdash N:\sigma$}
	\AxiomC{$\mathcal{D}''$}
	\noLine
	\UnaryInfC{$\Theta, z:A  \vdash P:\tau$}
	\RightLabel{$\multimap L$}
	\BinaryInfC{$\Delta,\Theta, y:\sigma\multimap A \vdash P[yN/z]:\tau$}
	\RightLabel{$cut$}
	\BinaryInfC{$\Gamma, \Delta, \Theta \vdash P[(\lambda x.M)N/z]:\tau$}
	\DisplayProof
&\red \ 
\AxiomC{$\mathcal{D}'$}
	\noLine
	\UnaryInfC{$\Delta\vdash N:\sigma$}
	\AxiomC{$\mathcal{D}$}
	\noLine
	\UnaryInfC{$\Gamma, x:\sigma \vdash M:A$}
	\RightLabel{$cut$}
	\BinaryInfC{$\Gamma, \Delta \vdash M[N/x]:A$}
	\AxiomC{$\mathcal{D}''$}
	\noLine
	\UnaryInfC{$\Theta, z:A \vdash P:\tau$}
	\RightLabel{$cut$}
	\BinaryInfC{$\Gamma, \Delta, \Theta \vdash P[M[N/x]/z]:\tau$}
	\DisplayProof\\ \\ 
		\AxiomC{$\mathcal{D}$}
	\noLine
	\UnaryInfC{$\Gamma \vdash M: A\langle\gamma/\alpha\rangle$}
	\RightLabel{$\forall R$}
	\UnaryInfC{$\Gamma \vdash M: \forall\alpha. A$}
	\AxiomC{$\mathcal{D}'$}
	\noLine
	\UnaryInfC{$\Delta, x: A\langle B/\alpha\rangle \vdash N:\tau$}
	\RightLabel{$\forall L$}
	\UnaryInfC{$\Delta, x: \forall\alpha. A \vdash N:\tau$}
	\RightLabel{$cut$}
	\BinaryInfC{$\Gamma,\Delta \vdash N[M/x]:\tau$}
	\DisplayProof
& \red \ 
	\AxiomC{$\mathcal{D} \langle B /\alpha\rangle$}
	\noLine
	\UnaryInfC{$\Gamma\vdash M:A\langle B/\alpha\rangle$}
	\AxiomC{$\mathcal{D}'$}
	\noLine
	\UnaryInfC{$\Delta, x:A\langle B/\alpha\rangle \vdash N:\tau$}
	\RightLabel{$cut$}
	\BinaryInfC{$\Gamma, \Delta \vdash N[M/x]:\tau$}
	\DisplayProof\\ \\ 
		\AxiomC{$\mathcal{D}$}
	\noLine
	\UnaryInfC{$\vdash V: \sigma$}
	\RightLabel{$p$}
	\UnaryInfC{$\vdash V: \shpos \sigma$}
	\AxiomC{$\mathcal{D}'$}
	\noLine
	\UnaryInfC{$\Gamma, x: \sigma \vdash M:\tau$}
	\RightLabel{$d$}
	\UnaryInfC{$\Gamma, y: \shpos \sigma \vdash M[y/x]:\tau$}
	\RightLabel{$cut$}
	\BinaryInfC{$\Gamma \vdash M[ V/x]:\tau$}
	\DisplayProof
& \red \ 
	\AxiomC{$\mathcal{D}$}
	\noLine
	\UnaryInfC{$\vdash V:\sigma$}
	\AxiomC{$\mathcal{D}'$}
	\noLine
	\UnaryInfC{$\Gamma, x:\sigma \vdash M:\tau$}
	\RightLabel{$cut$}
	\BinaryInfC{$\Gamma \vdash M[V/x]:\tau$}
	\DisplayProof\\\\ 
\AxiomC{$\mathcal{D}$}
\noLine
\UnaryInfC{$\vdash V:\sigma$}
\RightLabel{$p$}
\UnaryInfC{$\vdash  V:  \shpos \sigma$}
\AxiomC{$\mathcal{D}'$}
\noLine
\UnaryInfC{$\Delta \vdash M:\tau$}
\RightLabel{$w$}
\UnaryInfC{$\Delta, x: \shpos  \sigma \vdash \mathtt{discard}_{\sigma}\, x \mathtt{\ in\ }M:\tau$}
\RightLabel{$cut$}
\BinaryInfC{$\Delta \vdash \mathtt{discard}_{\sigma}\, V \mathtt{\ in\ }M:\tau$}
\DisplayProof
& \red \ 
\AxiomC{$\mathcal{D}'$}
\noLine
\UnaryInfC{$\Delta \vdash M:\tau$}
\DisplayProof\\ \\
\AxiomC{$\mathcal{D}$}
\noLine
\UnaryInfC{$\vdash V: \sigma$}
\RightLabel{$p$}
\UnaryInfC{$\vdash  V: \shpos \sigma$}
\AxiomC{$\mathcal{D}_1$}
\noLine
\UnaryInfC{$\Delta, y: \shpos \sigma, z:\shpos \sigma \vdash M: \tau$}
\AxiomC{$\mathcal{D}_2$}
\noLine
\UnaryInfC{$\vdash U :\sigma$}
\RightLabel{$c$}
\BinaryInfC{$\Delta, x:\shpos \sigma \vdash \mathtt{copy}^{U }_{\sigma}\,  x \mathtt{\ as\ } y, z \mathtt{\ in\ }M:\tau$}
\RightLabel{$cut$}
\BinaryInfC{$\Delta \vdash \mathtt{copy}^{U} _{\sigma}\, 
V \mathtt{\ as\ } y, z \mathtt{\ in \ }M:\tau$}
\DisplayProof
&\red \ 
\AxiomC{$\mathcal{D}$}
\noLine
\UnaryInfC{$\vdash V : \sigma$}
\RightLabel{$p$}
\UnaryInfC{$\vdash  V  : \shpos \sigma$}
\AxiomC{$\mathcal{D}$}
\noLine
\UnaryInfC{$\vdash V  :\sigma$}
\RightLabel{$p$}
\UnaryInfC{$\vdash  V : \shpos \sigma$}
\AxiomC{$\mathcal{D}_1$}
\noLine
\UnaryInfC{$\Delta, y: \shpos \sigma, z: \shpos \sigma \vdash M:\tau$}
\RightLabel{$cut$}
\BinaryInfC{$\Delta, z: \shpos \sigma \vdash M[V/y]:\tau$}
\RightLabel{$cut$}
\BinaryInfC{$\Delta \vdash M[V/y, V/z]:\tau$}
\DisplayProof
\end{aligned}
$}
\end{equation*}
\caption{\emph{Lazy} cut-elimination rules for the principal cuts $(\multimap$R$,\multimap$L$)$, $(\forall$L$,\forall$R$)$, 
$(p, d)$, $(p, w)$, and $(p, c)$.}
\label{fig: cut elimination exponential} 
\end{figure}

\noindent 
The lazy cut-elimination rules that we introduce here below 
are the standard ones, but restricted to avoid the elimination of non lazy instances of the
exponential cuts $(p, d)$, $(p, w)$ and $(p, c)$. 

\begin{defn}[Lazy cut-elimination rules]
\label{defn:Lazy cut-elimination rules}
Figure~\ref{fig: cut elimination exponential} 
introduces the \emph{lazy cut-elimination rules} for the principal 
cuts.
The elimination rules for commuting and axiom cuts are standard, so we omit
them all from Figure~\ref{fig: cut elimination exponential};
the (possibly) less obvious commuting ones can be recovered 
from the reductions on terms in Figure~\ref{fig: term commuting conversions}.
We remark that the elimination of the principal cuts
$(\forall$R$,\forall$L$)$ and $(p, d)$ does not modify the subject 
of their concluding judgment. So, we call them \textit{insignificant} as every other cut-elimination rule non influencing their concluding subject.
Given a derivation $ \mathcal{D} $, we write $\mathcal{D}\red \mathcal{D}'$ if $\mathcal{D}$ rewrites  to some $\mathcal{D}$ by 
one of the above rules.
\qed
\end{defn}
\noindent
Lazy cut-elimination is a way of preventing the erasure and the duplication of terms that are not values, and hence to restore a correspondence between cut-elimination and term reduction.  However, one can run into derivations containing exponential cuts that will never turn into lazy cuts, like the deadlock in~\eqref{eqn: non linear duplication}. The solution we adopt is to identify a set of judgments whose derivations can be rewritten into cut-free ones by a sequence of lazy cut-elimination steps.

\begin{defn}[Lazy types, lazy judgments and lazy derivations]
\label{defn: lazyness} 
We say that $\sigma$ is a \textit{lazy type} if it contains no \emph{negative} occurrences of $\forall$. 
Also, we say that $x_1:\sigma_1, \ldots, x_n:\sigma_n  \vdash M:\tau$ 
is a \textit{lazy judgment} if $\tau$ is a lazy type and 
$\sigma_1$, \ldots, $\sigma_n$ contain no \emph{positive}  
occurrences of  $\forall$. 
Last, $\mathcal{D}\triangleleft \Gamma \vdash 
M: \tau$ is called a \textit{lazy derivation}  if $\Gamma \vdash M: \tau$ 
is a lazy judgment.
\qed
\end{defn}
\noindent
Lemma~\ref{fact: quantification propagation} and 
\ref{lem: lazy derivation properties} here below, as well as 
Definition~\ref{defn: bound mio} and
\ref{defn:Height of an inference rule},  are the last preliminaries
to show the relevance of lazy cuts that occur in lazy derivations.

\begin{lem} {\ }
\label{fact: quantification propagation}
\begin{enumerate}[(1)]
\item 
\label{enum: shpos is lazy} 
Every type of the form $\shpos \sigma$ is closed and lazy.
\item 
\label{enum: closed implies positive forall} 
Every closed type has at least a positive quantification. 
\item 
\label{enum: rules c,d,w,forallL and some p are not lazy} 
Let $ \rho $ be any instance of 
$\forall \mathrm{L}$, $d$, $w$, $c$, and $ p $, the latter with a
non empty context. The conclusion of $ \rho $ is not lazy.
\item 
\label{enum: quantification propagation} 
Let $ \rho $ be any instance of $ ax $, $ \multimap\mathrm{R}$, 
$ \multimap\mathrm{L}$, $ \forall\mathrm{R}$, and $ p $, the latter with an empty context. If the conclusion of $ \rho $ is lazy, then, every  premise of $ \rho $ is lazy.
\item 
\label{enum: cut-free and lazy implies all judgements lazy}  
If $\mathcal{D}$ is a cut-free and lazy derivation of $\mathsf{LEM}$, 
then  all its judgments are lazy and no occurrences of $\forall \mathrm{L}$, $d$, $w$, $c$, and $ p $, the latter with a
non empty context, can exist in $ \mathcal{D} $.
\end{enumerate}
\end{lem}
\begin{proof}
Point~\ref{enum: shpos is lazy} holds by Definition~\ref{defn:Types for IMLL2shpos}. 
Point~\ref{enum: closed implies positive forall} is by a structural induction on types. 
Concerning Point~\ref{enum: rules c,d,w,forallL and some p are not lazy}, the conclusions of  $d$, $w$, $c$, and $ p $ contain $ \shpos\sigma $. This is a closed type,  hence,  by Point~\ref{enum: closed implies positive forall}, such conclusions are not lazy judgments. Moreover, 
$\forall \mathrm{L}$   introduces  a positive occurrence of
$ \forall $ in the context of its conclusion, so that this latter cannot be a lazy judgment. 
Point~\ref{enum: quantification propagation} is a case analysis on every listed inference rule.
%
As for Point~\ref{enum: cut-free and lazy implies all judgements lazy},
we can proceed by structural induction on $ \mathcal{D} $. By definition,
the conclusion of $ \mathcal{D} $ is a lazy judgment. 
Point~\ref{enum: rules c,d,w,forallL and some p are not lazy} excludes
that one among $\forall \mathrm{L}$, $d$, $w$, $c$, and $ p $ (with a
non empty context) may be the last rule of $ \mathcal{D} $. So, 
only one among $ ax $, $ \multimap\mathrm{R}$, 
$ \multimap\mathrm{L}$, $ \forall\mathrm{R}$, and $ p $ (with an empty context) can be the concluding rule, say $ r $, of $ \mathcal{D} $.
Point~\ref{enum: quantification propagation} implies that 
all the premises of $ r $ are lazy. Hence, we can apply the inductive hypothesis to the derivations of the premises of $ r $  and conclude.
\end{proof}

\begin{defn}[Size of a derivation]
	\label{defn: bound mio}
	The \emph{size} $\vert\mathcal{D}\vert$ of a derivation $\mathcal{D}$ in $\mathsf{LEM}$ is defined by induction:
	\begin{enumerate}
		\item 
		If $\mathcal{D}$ is $ax$ then $\vert\mathcal{D}\vert=1$.
		\item 
		If $\mathcal{D}$ is a derivation $\mathcal{D}'$ that concludes by a
		rule with a single premise, then $\vert\mathcal{D}\vert= \vert\mathcal{D}'\vert+1$.
		\item 
		If $\mathcal{D}$ composes two derivations $\mathcal{D}'$ and $\mathcal{D}'' $ by a rule with two premises, but different from $c$, 
		then $\vert\mathcal{D}\vert= \vert\mathcal{D}'\vert+ \vert\mathcal{D}''\vert+1$. 
		\item 
		\label{enum:on standard dimension of c}
		If $\mathcal{D}$ composes two derivations $\mathcal{D}'$ and $\mathcal{D}'' $ by the rule $c$,  
		then $\vert\mathcal{D}\vert= \vert\mathcal{D}'\vert+ \vert\mathcal{D}''\vert+3$. 
		\qed
	\end{enumerate}
\end{defn}

\begin{rem}
Adding ``3'' instead of the possibly expected ``1'' in the clause~\eqref{enum:on standard dimension of c} of Definition~\ref{defn: bound mio} highlights the non linearity that the instances of $ c $ 
introduce in the course of the lazy cut-elimination on \textsf{LEM} of Definition~\ref{defn:Lazy cut-elimination strategy} below.
\qed
\end{rem}

\begin{lem} 
\label{lem: lazy derivation properties}
Let $\mathcal{D}\triangleleft x_1: \sigma_1, \ldots, x_n: \sigma_n \vdash M: \sigma$ be a cut-free and lazy derivation.
\begin{enumerate}[(1)]
\item 
\label{enum: cut free and lazy implies linear normal form} 
$M$ is a linear $\lambda$-term in normal form.
\item  
\label{enum: size term bounded size types LEML} 
$\vert M \vert \leq \sum_{i=1}^n \vert \sigma_i \vert +  \vert \sigma \vert $.
\item 
\label{enum: relation lazy cut free derivation and term} 
$\vert \mathcal{D} \vert= \vert M \vert +k$, where $k$  is the number of $\forall$ and $\shpos$ occurring in   $\sigma_1, \ldots, \sigma_n, \sigma$.
\item 
\label{enum: lazy cut free derivations are smaller if terms are smaller} 
If $\mathcal{D}'\triangleleft x_1: \sigma_1, \ldots, x_n: \sigma_n \vdash N: \sigma$ is  a lazy and cut-free derivation, then $\vert N \vert \leq \vert M \vert$ implies $\vert \mathcal{D}'\vert \leq \vert \mathcal{D} \vert $.
\item 
\label{enum: finite number of normal forms} 
The set of values with lazy type $ \sigma $ is finite.
\end{enumerate}
\end{lem}
\begin{proof}
The assumptions on $ \mathcal{D}$ allow to apply Lemma~\ref{fact: quantification propagation}.\ref{enum: cut-free and lazy implies all judgements lazy}
which implies that every judgment in $\mathcal{D}$ is lazy and
free of $\forall$L, $d$, $w$, $c$ and $ p $ (with a non empty context). 
Hence, we can prove Points~\ref{enum: cut free and lazy implies linear normal form}-\ref{enum: relation lazy cut free derivation and term} by induction on the structure of $\mathcal{D}$. Point~\ref{enum: lazy cut free derivations are smaller if terms are smaller} is a corollary of
Point~\ref{enum: relation lazy cut free derivation and term}. Point~\ref{enum: finite number of normal forms} is a corollary of 
Point~\ref{enum: size term bounded size types LEML}.
\end{proof}

\begin{defn}[Height of an inference rule]
\label{defn:Height of an inference rule}
Let $\mathcal{D}\triangleleft \Gamma\vdash M:\sigma$ be a derivation and 
$ r $ a rule instance in it. The \emph{height of $r$}, written $ h(r) $,  is the number of rule instances from the conclusion 
$ \Gamma\vdash M:\sigma $ of $ \mathcal{D} $ upward to the conclusion of $ r $. The \textit{height of $\mathcal{D}$}, written $h(\mathcal{D})$, is the largest $h(r)$ among its rule instances.
\qed
\end{defn}

\noindent
Lemma~\ref{lem: existence of the safe exponential cut} 
and~\ref{lem: above the cut} assure that we can eliminate exponential lazy cuts from a lazy derivation.

\begin{lem}[Existence of a lazy $ cut $] 
\label{lem: existence of the safe exponential cut} 
Let $\mathcal{D}$ be a lazy derivation with only exponential cuts in it.
At least one of those cuts is \emph{safe}.
\end{lem}
\begin{proof}
Let $\Gamma \vdash M: \tau$ be the conclusion of $ \mathcal{D} $.
By contradiction, let us suppose that every occurrence of (exponential)
$ cut $ in $ \mathcal{D} $ is a deadlock.
At least one of them, say $ c_m $, has minimal height $ h(c_m) $, i.e.~no other $ cut $ occurs in the sequence of rule instances, say $ r_1,\ldots r_n $, from the conclusion of $ c_m $ down to the
one of $ \mathcal{D} $.
Since $ c_m $ is a deadlock, its leftmost premise has form $ \Delta \vdash N:\shpos\sigma $, 
where $ \Delta\neq \emptyset $. 
By Proposition~\ref{prop: exponential context}, $\Delta$ is strictly exponential and the whole $ \Delta \vdash N:\shpos\sigma $ is a non lazy
judgment by Lemma~\ref{fact: quantification propagation}.\ref{enum: shpos is lazy} and Lemma~\ref{fact: quantification propagation}.\ref{enum: closed implies positive forall}. 
The contraposition of
Lemma~\ref{fact: quantification propagation}.\ref{enum: quantification propagation} 
implies that the non lazy judgment in $ c_m $ can only be transformed to 
a non lazy judgment by every $ r_i $, with $ 1\leq i\leq n $, letting the 
conclusion of $ \mathcal{D} $ non lazy, so
contradicting the assumption.
Hence, $ c_m $ must be safe.
\end{proof}

\begin{lem}[Eliminating a lazy $ cut $]
\label{lem: above the cut}
Let $\mathcal{D}$ be a lazy derivation with only exponential cuts in it.
A lazy derivation $\mathcal{D}^*$ exists such that 
both $\mathcal{D} \red \mathcal{D}^*$, by reducing a lazy cut, and 
$\vert\mathcal{D}^*\vert < \vert\mathcal{D}\vert$.
\end{lem}
\begin{proof}
Lemma~\ref{lem: existence of the safe exponential cut}
implies that $ \mathcal{D} $ contains at least an exponential cut which is  safe. Let us take $(p, X)$ with maximal height $ h((p, X)) $ among those safe  instances of $ cut $. 
So, if $ (p, X) $ has form:
\begin{equation}
\label{eqn: above the cut}
\AxiomC{$\mathcal{D}'$}
\noLine
\UnaryInfC{$\vdash N: \sigma$}
\RightLabel{$p$}
\UnaryInfC{$\vdash  N: \shpos \sigma$}
\AxiomC{\vdots}
\RightLabel{$X$}
\UnaryInfC{$\Delta, x:\shpos \sigma \vdash M:\tau$}
\RightLabel{$cut$ \enspace ,}
\BinaryInfC{$ \Delta \vdash M[N/x]:\tau$}
\DisplayProof
\end{equation}
\noindent
then $\mathcal{D}'$ is a lazy derivation because $\shpos \sigma$ is a lazy type by  Lemma~\ref{fact: quantification propagation}.\ref{enum: shpos is lazy}.  Since $\mathcal{D}'$ is lazy and can only contain exponential cuts, by Lemma~\ref{lem: existence of the safe exponential cut} and by maximality of  $ h((p, X)) $,   it is forcefully cut-free.
So, by Lemma~\ref{lem: lazy derivation properties}.\ref{enum: cut free and lazy implies linear normal form}, we have that $ N $ is a value, 
i.e.~$(p, X)$ is lazy and we can reduce it to obtain $ \mathcal{D}^*$. 
If $ X $ in~\eqref{eqn: above the cut} is $d$ or $w$, then it is 
simple to show
that $\vert\mathcal{D}^*\vert<\vert\mathcal{D}\vert$. 
Let $ X $ be $ c $. Then, \eqref{eqn: above the cut} is:
\begin{equation}
\label{eqn:above the cut (p,c)}
\AxiomC{$\mathcal{D}'$}
\noLine
\UnaryInfC{$\vdash V: \sigma$}
\RightLabel{$p$}
\UnaryInfC{$\vdash  V: \shpos \sigma$}
\AxiomC{$\mathcal{D}''$}
\noLine
\UnaryInfC{$\Delta, y: \shpos \sigma, z:\shpos \sigma \vdash M': \tau$}
\AxiomC{$\mathcal{D}'''$}
\noLine
\UnaryInfC{$\vdash U: \sigma$}
\RightLabel{$c$}
\BinaryInfC{$\Delta, x:\shpos \sigma \vdash \mathtt{copy}^{U}_{\sigma}\,  x \mathtt{\ as\ }y,z \mathtt{\ in\ }M':\tau$}
\RightLabel{$cut$ \enspace ,}
\BinaryInfC{$\Delta \vdash \mathtt{copy}^{U} _{\sigma}\,   V \mathtt{\ as\ }y,z \mathtt{\ in \ }M':\tau$}
\DisplayProof
\end{equation}
with $\mathcal{D}'''$ lazy and cut-free for the same reasons as $ \mathcal{D}' $ is. So, \eqref{eqn:above the cut (p,c)} can reduce to:
\begin{equation}
\label{eqn:reducing above the cut (p,c)}
\AxiomC{$\mathcal{D}'$}
\noLine
\UnaryInfC{$\vdash V : \sigma$}
\RightLabel{$p$}
\UnaryInfC{$\vdash  V  : \shpos \sigma$}
\AxiomC{$\mathcal{D}'$}
\noLine
\UnaryInfC{$\vdash V  :\sigma$}
\RightLabel{$p$}
\UnaryInfC{$\vdash  V : \shpos \sigma$}
\AxiomC{$\mathcal{D}''$}
\noLine
\UnaryInfC{$\Delta, y: \shpos \sigma, z: \shpos \sigma \vdash M':\tau$}
\RightLabel{$cut$ \enspace .}
\BinaryInfC{$\Delta, z: \shpos \sigma \vdash M'[V/y]:\tau$}
\RightLabel{$cut$}
\BinaryInfC{$\Delta \vdash M'[V/y, V/z]:\tau$}
\DisplayProof
\end{equation}
\noindent
By Lemma~\ref{lem: lazy derivation properties}.\ref{enum: finite number of normal forms}, we can safely  assume that $U$ is a value with largest size among values of type $\sigma$.
So, Lemma~\ref{lem: lazy derivation properties}.\ref{enum: size term bounded size types LEML} implies $\vert V \vert \leq \vert U \vert$, 
from which $\vert \mathcal{D}' \vert \leq \vert \mathcal{D}''' \vert$,
by Lemma~\ref{lem: lazy derivation properties}.\ref{enum: lazy cut free derivations are smaller if terms are smaller}.
By applying Definition~\ref{defn: bound mio}.\ref{enum:on standard dimension of c} to 
\eqref{eqn:above the cut (p,c)} and
\eqref{eqn:reducing above the cut (p,c)}, we have $\vert \mathcal{D}^* \vert < \vert \mathcal{D} \vert$.
\end{proof}

\begin{defn}[Lazy cut-elimination strategy]
\label{defn:Lazy cut-elimination strategy}
Let $ \mathcal{D} $ be a lazy derivation of $ \textsf{LEM} $.
Let a \emph{round} on $ \mathcal{D} $ be defined as follows:
\begin{enumerate}[\{1\}]
	\item \label{enum: round comm}  
	Let eliminate all the commuting instances of $ cut $.
	\item \label{enum: round sym} 
	If any instance of $ cut $ remains, it is necessarily symmetric.
	Let reduce a multiplicative cut, if any.
	Otherwise, let reduce a lazy exponential cut, if any.
\end{enumerate}
The \emph{lazy cut-elimination strategy} iterates 
rounds, starting from $ \mathcal{D} $, until instances of $ cut $
exist in the obtained derivation.
\qed
\end{defn}

\begin{thm}[Lazy cut-elimination has a cubic bound]
\label{thm: cut elimination for downarrow IMLL2}
Let $\mathcal{D}$ be a lazy derivation. 
The lazy cut-elimination can reduce $\mathcal{D}$ to a cut-free $\mathcal{D}^*$ in $\mathcal{O}(\vert \mathcal{D} \vert^3)$ steps.
\end{thm}
\begin{proof}
Let $H(\mathcal{D})$ be the sum of the heights $h(\mathcal{D}')$ 
of all  sub-derivations $\mathcal{D}'$ of 
$\mathcal{D}$ whose conclusion is an instance of $ cut $.
We proceed by induction on the lexicographically order of the pairs 
$\langle\vert \mathcal{D}\vert , H(\mathcal{D}) \rangle$. 
%
%
To show that the lazy cut-elimination strategy in Definition~\ref{defn:Lazy cut-elimination strategy} terminates, we
start by applying a round to $ \mathcal{D} $, using 
step~\ref{enum: round comm}.
%
%
Every commuting cut-elimination step  just moves an instance of $ cut $ upward, strictly decreasing $H(\mathcal{D})$ and leaving
$ \vert \mathcal{D}\vert $ unaltered. 
Let us continue by applying
	step~\ref{enum: round sym} of the round.
As usual, $ \vert \mathcal{D}\vert $ shrinks when eliminating a multiplicative cut.
If only exponential instances of $ cut $ remain,  by Lemma~\ref{lem: above the cut} we can rewrite 
$ \mathcal{D} $ to $ \mathcal{D}' $ by reducing a lazy exponential cut in such a way that 
$ \vert \mathcal{D}'\vert  < \vert \mathcal{D} \vert $. Therefore,  the lazy cut-elimination strategy terminates with a cut-free derivation $ \mathcal{D}^* $.
\par
We now exhibit a bound on the number of 
cut-elimination steps from $ \mathcal{D} $ to $ \mathcal{D}^* $.
Generally speaking, we can represent a lazy strategy as:
\begin{align}
\label{align:cut-elimination bound diagram}
\mathcal{D}
= \mathcal{D}_0
\underbrace{\longrightarrow}_{cc_0}
\mathcal{D}'_0
\red
\mathcal{D}_1
\,\cdots \underbrace{\longrightarrow}_{cc_i}
\mathcal{D}'_{i}
\red
\mathcal{D}_{i+1}
\underbrace{\longrightarrow}_{cc_{i+1}}\cdots\,
\mathcal{D}'_{n-1}
\red
\mathcal{D}_n
\underbrace{\longrightarrow}_{cc_{n}}
\mathcal{D}'_{n}
= \mathcal{D}^*
\enspace ,
\end{align}
where every $ cc_j $ denotes the number of 
commuting cuts applied from derivation $ \mathcal{D}_{j} $ to derivation
$ \mathcal{D}'_{j} $, for every $ 0\leq j\leq n $.
A bound on every $ cc_j $ is $ \vert\mathcal{D}_{j}\vert^2 $ because every
instance of rule in $ \mathcal{D}_{j} $ can, in principle, be commuted with every other. 
The first part of the proof implies
$ \vert\mathcal{D}_{j}\vert = \vert\mathcal{D}'_{j}\vert $,
for every $ 0\leq j\leq n $.
Lemma~\ref{lem: above the cut} implies
$ \vert\mathcal{D}'_{j}\vert > \vert\mathcal{D}_{j+1}\vert $,
for every $ 0\leq j\leq n-1 $.
So, $ n\leq \vert\mathcal{D}\vert $ and the total number of cut-elimination steps in 
\eqref{align:cut-elimination bound diagram} is 
$ O(\vert\mathcal{D}\vert\cdot\vert\mathcal{D}\vert^2) $.
\end{proof}

\begin{rem}
The cubic bound on the lazy strategy keeps holding also in the case we apply the \emph{lazy} cut-elimination to \emph{non lazy} derivations. Of course, in that case, deadlocks may remain in the final derivation where no instance of $cut$ can be further eliminated.
\qed
\end{rem}


\subsection{Subject reduction theorem}
\label{sec: subject reduction}
%
The proof of the Subject reduction requires some typical preliminaries.
\begin{lem}[Substitution]
\label{lem: substitution type variables} 
If $\Gamma \vdash M: \tau$ then $\Gamma[A/ \alpha] \vdash M: \tau[A/ \alpha]$, for every \emph{linear} type $ A $.
\end{lem}
\begin{lem}[Generation] \label{lem: generation}  {\ }  
\begin{enumerate}[(1)]
\item \label{lem:generation 1} 
If $\mathcal{D} \triangleleft \Gamma \vdash \lambda x. M:\tau$, then 
$\tau= \shpos^n   \forall \vec{ \alpha}. (\sigma \multimap A)$, where 
$\shpos^n\triangleq \shpos \overset{n}{\ldots}\shpos$ and 
$\vec{\alpha}= \alpha_1,\ldots, \alpha_m$, for some $n\geq0$ and $m\geq 0$.
\item  \label{lem:generation 5} 
If $\mathcal{D} \triangleleft \Delta, x: \forall \alpha.A \vdash P:\tau$, then an instance $ r $ of $\forall\mathrm{L}$ exists in $\mathcal{D}$ 
with 
conclusion $\Delta', x: \forall \alpha.A \vdash P':\tau'$,
for some $\Delta', P'$ and $\tau'$. 
I.e., $ r $ introduces $x:\forall \alpha.A$.
\item  \label{lem:generation 6}
If $\mathcal{D} \triangleleft \Delta, x: \sigma \multimap B 
    \vdash P[xN/y]: \tau$, 
then an instance $ r $ of $\multimap\mathrm{L}$ 
exists in $\mathcal{D}$ with
conclusion $\Delta', x: \sigma \multimap B \vdash P'[xN'/y] : \tau'$,
for some $\Delta', P', N'$ and $\tau'$.
I.e., $ r $ introduces $x:\sigma \multimap B$.
\item  \label{lem:generation 2} 
If $\mathcal{D} \triangleleft \Gamma \vdash \lambda x. M : \forall  \alpha. A$, then a derivation $\mathcal{D}'$ exists
which is $\mathcal{D}$ with some rule permuted in order to obtain an instance of $\forall\mathrm{R}$ as last rule of $\mathcal{D}'$.
\item  \label{lem:generation 3} 
If $\mathcal{D}\triangleleft \Gamma \vdash \lambda x. P: \sigma \multimap B$, then a derivation $\mathcal{D}'$ exists
which is $\mathcal{D}$ with some rule permuted in order to obtain an instance of $\multimap\mathrm{R}$ as last rule of $\mathcal{D}'$.
\item  \label{lem:generation 7}
If $\mathcal{D} \triangleleft \Delta, x: \shpos \sigma \vdash P[xN/y] : 
\tau$, then an instance $ r $ of $ d $ exists in $\mathcal{D}$ 
with 
conclusion  $\Delta', x: \shpos \sigma \vdash P'[xN'/y]:\tau'$,
for some $\Delta', P', N'$ and $\tau'$.
I.e.,  $d$ introduces $x:\shpos  \sigma$.
\item \label{lem:generation 4} 
If $\mathcal{D}  \triangleleft \Gamma \vdash M: \shpos \sigma$, then a derivation $\mathcal{D}'$ exists
which is $\mathcal{D}$ with some rule permuted in order to get an instance of $p$ as last rule of $\mathcal{D}'$.
\item \label{lem:generation 8}
If $\mathcal{D} \triangleleft \Delta, x: \shpos  \sigma \vdash 
\mathtt{discard}_{\sigma}\, x \mathtt{\ in \ }P: \tau$, then
an instance $ r $ of $ w $ exists in $\mathcal{D}$ with 
conclusion 
$ \Delta', x: \shpos\sigma 
  \vdash 
  \mathtt{discard}_{\sigma}\, x \mathtt{\ in \ }P': \tau' $,
  for some $ \Delta', P' $ and $ \tau' $.
I.e., $r$ introduces $x:\shpos \sigma$.
\item \label{lem:generation 9}
If $\mathcal{D} \triangleleft \Delta, x: \shpos \sigma \vdash 
	\mathtt{copy}_{\sigma}^{U} \, x  \mathtt{\ as \ } x_1, x_2 \mathtt{\ in \ }P: \tau$, then
	an instance $ r $ of $ c $ exists in $\mathcal{D}$ with 
	conclusion 
	$\Delta', x: \shpos \sigma \vdash 
	\mathtt{copy}_{\sigma}^{U} \, x  \mathtt{\ as \ } 
	x_1, x_2 \mathtt{\ in \ }P': \tau'$,
	for some $ \Delta', P' $ and $ \tau' $.
	I.e., $r$ introduces $x:\shpos \sigma$.
\end{enumerate}
\end{lem}
\begin{proof}
We can adapt the proof by 
Gaboardi\&Ronchi in~\cite{gaboardi2009light} to $\mathsf{LEM}$
because the types in Definition~\ref{defn:Types for IMLL2shpos}
are a sub-set of Gaboardi\&Ronchi's \emph{essential types}.
In particular, Point\!~\textit{\ref{lem:generation 4}} relies on
Proposition~\ref{prop: exponential context}.
\end{proof}

\begin{thm}[Subject reduction] 
\label{thm:Subject Reduction for IMLL2shpos} 
 If $\Gamma 
\vdash M: \tau$  and $M \rightarrow M^\prime$, then $\Gamma 
\vdash M^\prime: \tau$.
\end{thm}
\begin{proof}
We proceed by structural induction on $ \mathcal{D} $.
The crucial case of $M \rightarrow M^\prime$
(Figure~\ref{fig: term reduction rules}) is when
$(\lambda x. P)Q$ exists in $M$ and $\mathcal{D}$ contains:
\begin{prooftree}
\AxiomC{$\mathcal{D}' \triangleleft \Delta \vdash \lambda x. P: \sigma$}
\AxiomC{$\mathcal{D}'' \triangleleft \Sigma, y: \sigma \vdash N[ yQ/z ]:\tau$}
\RightLabel{$ cut $ \enspace .}
\BinaryInfC{$\mathcal{D} \triangleleft \Delta, \Sigma \vdash N[(\lambda x.P)Q/z]:\tau$}
\end{prooftree}
Lemma~\ref{lem: generation}.\textit{\ref{lem:generation 1}} implies that  $\sigma= \shpos^n \forall \vec{ \alpha}. (\sigma_1  \multimap C)$, where $\shpos^n\triangleq \shpos \overset{n}{\ldots}\shpos$ and $\vec{\alpha}= \alpha_1,\ldots, \alpha_m$, for some $n\geq0$ and $m\geq 0$. 
Lemma~\ref{lem: generation}.\textit{\ref{lem:generation 5}},  
\ref{lem: generation}.\textit{\ref{lem:generation 6}}, 
and \ref{lem: generation}.\textit{\ref{lem:generation 7}}, 
imply that $\mathcal{D}''$ has form: 
\begin{equation}
\label{eqn:conclusion of derivation}
\AxiomC{\vdots}
\noLine
\UnaryInfC{$ \Sigma_1'' \vdash Q'':\sigma_1'$}
\AxiomC{$\vdots$}
\noLine
\UnaryInfC{$\Sigma_2'', z: C' \vdash N'': \tau''$}
\RightLabel{$\multimap$L}
\BinaryInfC{$\Sigma''_1, \Sigma''_2, y''':\sigma_1' \multimap C' \vdash N''[y'''Q''/z]: \tau''$}
\noLine
\UnaryInfC{$\vdots $}
\noLine
\UnaryInfC{$\Sigma', y'': \forall \vec{\alpha}. (\sigma_1 \multimap C )\vdash N'[y''Q'/z]: \tau'$}
\RightLabel{$d$}
\UnaryInfC{$\Sigma', y': \shpos \forall \vec{\alpha}.( \sigma_1 \multimap C) \vdash N'[y'Q'/z]: \tau'$}
\noLine
\UnaryInfC{$\vdots $}
\noLine
\UnaryInfC{$\Sigma, y:\shpos ^n \forall \vec{\alpha}. (\sigma_1 \multimap C) \vdash N[yQ/z]: \tau$}
\DisplayProof
\end{equation}
where $\sigma_1'= \sigma_1[A_1/\alpha_1, \ldots, A_m /\alpha_m]$ and $C'=C[A_1/\alpha_1, \ldots, A_m /\alpha_m]$, for some $\Sigma', \Sigma''_1, \Sigma''_2,y', y'',y''', N', N'', Q',Q'',  \tau',\tau'',  A_1, \ldots, A_m$.
Lemma~\ref{lem: generation}.\textit{\ref{lem:generation 2}},  
\ref{lem: generation}.\textit{\ref{lem:generation 3}} and 
\ref{lem: generation}.\textit{\ref{lem:generation 4}} imply that, permuting some of its rules,
$\mathcal{D}'$ can be reorganized as:
\begin{prooftree}
\AxiomC{$\vdots$}
\noLine
\UnaryInfC{$\Delta, x:\sigma_1 \vdash P: C$}
\RightLabel{$\multimap$R}
\UnaryInfC{$\Delta \vdash \lambda x. P:\sigma_1 \multimap C$}
\RightLabel{$\forall$R  \enspace ,}
\doubleLine
\UnaryInfC{$\Delta \vdash \lambda x.P: \forall \vec{\alpha}. (\sigma_1 \multimap C)$}
\doubleLine
\RightLabel{$p$}
\UnaryInfC{$\Delta \vdash \lambda x.P: \shpos^n \forall \vec{\alpha}. (\sigma_1 \multimap C)$}
\end{prooftree}
where the concluding instances of $ p $ are necessary if $n>0$ and
are legally introduced because $\Delta$ is strictly exponential as consequence of Proposition~\ref{prop: exponential context}
that we can apply to the judgment beacause $ \shpos\sigma $ is
strictly exponential as well.
Moreover, Lemma~\ref{lem: substitution type variables} assures that a derivation of $\Delta, x:\sigma_1' \vdash P: C'$ exists
because $\alpha_1, \ldots, \alpha_m$ are not free in $\Delta$. Therefore:
\begin{prooftree}
\AxiomC{$\vdots$}
\noLine
\UnaryInfC{$\Sigma''_1 \vdash Q'': \sigma'_1$}
\AxiomC{$\vdots$}
\noLine
\UnaryInfC{$\Delta, x:\sigma'_1 \vdash P: C'$}
\RightLabel{$cut$}
\BinaryInfC{$\Delta, \Sigma''_1 \vdash P[Q''/x]: C'$}
\AxiomC{$\vdots$}
\noLine
\UnaryInfC{$\Sigma''_2, z:C' \vdash N'': \tau''$}
\RightLabel{$cut$ \enspace .}
\BinaryInfC{$\Delta, \Sigma''_1, \Sigma''_2 \vdash N''[P[Q''/x]/z]: \tau''$}
\noLine
\UnaryInfC{$\vdots$}
\noLine
\UnaryInfC{$\Delta, \Sigma \vdash N[P[Q/x]/z]: \tau$}
\end{prooftree}
\noindent
which concludes with the same rules as in~\eqref{eqn:conclusion of derivation}. A similar proof exists, which relies on
Lemma~\ref{lem: generation}.\textit{\ref{lem:generation 8}}, 
or Lemma~\ref{lem: generation}.\textit{\ref{lem:generation 9}},
when reducing
$\mathtt{discard}_{\sigma}$ $V\mathtt{\ in\ }$ $M$, or 
$\mathtt{copy}^{U}_{\sigma}\, V \mathtt{\ as\ }y,z \mathtt{\ in\ }M$. All the remaining cases are straightforward. 
\end{proof}
\section{Translation of $\mathsf{LEM}$ into $\mathsf{IMLL}_2$ and exponential compression}\label{sec: the expressiveness of the system}
The system $\mathsf{LEM}$ provides a logical status to copying 
and erasing operations that exist in $\mathsf{IMLL}_2$. In what follows, we 
show that a translation $(\_)^\bullet$ from $\mathsf{LEM}$ into $\mathsf{IMLL}_2$ exists 
that \enquote{unpacks} both the constructs $\mathtt{discard}_{\sigma}$ and  
$\mathtt{copy}_{\sigma}^{V}$ by turning them into, respectively, an  eraser 
and a duplicator of ground types. Then,  we show that the reduction steps in 
Figure~\ref{fig: term reduction rules} and the commuting conversions in 
Figure~\ref{fig: term commuting conversions} can be simulated by the 
$\beta \eta$-reduction of the linear $\lambda$-calculus, as long as we 
restrict to terms of $\Lambda_\shpos$ typable in $\mathsf{LEM}$. 
Last, we discuss the complexity of the translation, and we prove that 
every term typable in $\mathsf{LEM}$ is mapped to a linear 
$\lambda$-term whose size can be is exponential in the one of the original 
term.

 We start with defining the translation from $\mathsf{LEM}$ to $\mathsf{IMLL}_2$.
\begin{defn} [From $ \mathsf{LEM} $ to $ \mathsf{IMLL}_2 $]
\label{defn: translation IMLL2 shpos with oplus into IMLL2} 
The map  
$(\_)^\bullet: \mathsf{LEM} \longrightarrow\mathsf{IMLL}_2$ 
translates a derivation 
$\mathcal{D} \triangleleft\Gamma \vdash_{\mathsf{LEM}}  M: \tau $ into a derivation $\mathcal{D}^\bullet \triangleleft \Gamma^\bullet \vdash_{\mathsf{IMLL}_2} M^\bullet :\tau^\bullet$:
\begin{enumerate}
\item For all types $\sigma \in \Theta_\shpos$:
\begin{align*}
\alpha^\bullet &\triangleq \alpha\\
(\tau \multimap A)^\bullet &\triangleq \tau^\bullet \multimap A^\bullet \\
(\forall \alpha . A)^\bullet &\triangleq \forall \alpha. A^\bullet \\
(\shpos \tau)^\bullet & \triangleq \tau^\bullet \enspace .
\end{align*}
\item For all contexts $\Gamma= x_1: \sigma_1, \ldots, x_n: \sigma_n$,  we set $\Gamma^\bullet\triangleq  x_1: \sigma_1 ^\bullet, \ldots, x_n: \sigma_n^\bullet$;
\item
\label{enum: translation on terms} 
For all typable terms $ M\in \Lambda_\shpos$:
\begin{align*}
x^\bullet  &\triangleq x
\\
(\lambda x. P )^\bullet &\triangleq \lambda x. P^\bullet 
\\
(PQ)^\bullet &\triangleq P^\bullet Q^\bullet
\\
(\mathtt{discard}_{\sigma}\, P \mathtt{\ in\ }Q)^\bullet
&\triangleq
  \mathtt{let\ } \mathtt{E}_{\sigma^\bullet}\, P^\bullet \mathtt{\ be\ }I \mathtt{\ in\ }Q^\bullet 
\\
(\mathtt{copy}^{V}_{\sigma}\, P \mathtt{\ as\ }x_1,x_2 \mathtt{\ in\ 
}Q)^\bullet &\triangleq \mathtt{let\ }\mathtt{D}^{V^\bullet} _{\sigma^\bullet}\, P^\bullet 
\mathtt{\ be\ } x_1,x_2 \mathtt{\ in \ }Q^\bullet \enspace ,
\end{align*}
where $\mathtt{E}_{\sigma^\bullet}$ and $\mathtt{D}^{V^\bullet} _{\sigma^\bullet}$ (see Remark~\ref{rem: duplicator} in~\ref{sec: the d-soundness theorem DICE}) are the eraser and the duplicator of $\sigma^\bullet$ which is both
ground, because $\sigma$ is closed and with no negative occurrences of $\forall$, 
and inhabited by $V^\bullet$. 

\item 
The definition of $ (\_)^\bullet $ extends to any derivation 
$\mathcal{D}\triangleleft\Gamma\vdash M:\sigma$ of $\mathsf{LEM}$
in the obvious way, following
the structure of $ M^{\bullet} $. 
Figure \ref{fig: translation inference rules} collects the most interesting cases.
\qed
\end{enumerate}   
\end{defn}

\begin{rem}
Both $\mathtt{E}_{\sigma^\bullet}$ and 
$\mathtt{D}^{V^\bullet} _{\sigma^\bullet}$ in 
point~\ref{enum: translation on terms} of 
Definition~\ref{defn: translation IMLL2 shpos with oplus into IMLL2} here above exist by Theorem~\ref{thm: pi1 types are erasable} and Theorem~\ref{thm: pi1 are duplicable}.
\qed
\end{rem}

\begin{figure}[t]
\centering
\begin{equation}
\begin{aligned}\nonumber
\left(
\scalebox{0.6}{
\AxiomC{$\mathcal{D}$}
\noLine
\UnaryInfC{$x_1: \shpos \sigma_1, \ldots, x_n: \shpos \sigma_n \vdash M: \sigma$}
\RightLabel{$p$}
\UnaryInfC{$x_1: \shpos \sigma_1, \ldots, x_n: \shpos \sigma_n \vdash M: \shpos \sigma$}
\DisplayProof } 
\right)^\bullet
\triangleq &
\left(
\scalebox{0.6}{
\AxiomC{$\mathcal{D}$}
\noLine
\UnaryInfC{$x_1: \shpos \sigma_1, \ldots, x_n: \shpos \sigma_n \vdash M :\sigma$}
\DisplayProof } 
\right)^\bullet
\\
\\
\left(
\scalebox{0.6}{
\AxiomC{$\mathcal{D}$}
\noLine
\UnaryInfC{$\Gamma, x: \sigma \vdash M:\tau$}
\RightLabel{$d$}
\UnaryInfC{$\Gamma, y: \shpos \sigma \vdash M[y/x]: \tau$}
\DisplayProof} 
\right)^\bullet
\triangleq &
\scalebox{0.6}{
\AxiomC{}
\RightLabel{ax}
\UnaryInfC{$  y:\sigma^\bullet \vdash y: \sigma^\bullet$}
\def\extraVskip{0pt}
\AxiomC{$\left( \begin{aligned}\begin{gathered} \mathcal{D} \\ \Gamma, x: \sigma \vdash M: \tau \end{gathered}\end{aligned}\right)^\bullet$}
\noLine
\UnaryInfC{\hspace{2.5cm}}
\def\extraVskip{3pt}
\RightLabel{$cut$}
\BinaryInfC{$\Gamma^\bullet, y: \sigma^\bullet \vdash M^\bullet  [y/x]:\tau^\bullet$}
\DisplayProof} 
\\
\\
\left(
\scalebox{0.6}{
\AxiomC{$\mathcal{D}$}
\noLine
\UnaryInfC{$\Gamma \vdash M:\tau$}
\RightLabel{$w$}
\UnaryInfC{$\Gamma , x: \shpos \sigma \vdash \mathtt{discard }_{\sigma}\, x \mathtt{\  in\  }M:\tau$}
\DisplayProof}
\right)^\bullet
\triangleq&
\scalebox{0.6}{
\AxiomC{$\vdots$}
\noLine
\UnaryInfC{$x: \sigma^\bullet \vdash \mathtt{E}_{\sigma^\bullet}\, x: \mathbf{1}$}
\AxiomC{$ \left( \begin{aligned}\begin{gathered}  \mathcal{D} \\ \Gamma \vdash M :\tau \end{gathered}\end{aligned} \right)^\bullet$}
\noLine
\UnaryInfC{\vdots}
\noLine
\UnaryInfC{$\Gamma^\bullet, y: \mathbf{1} \vdash \mathtt{let\ }y \mathtt{\ be\ }I \mathtt{\ in\ }M^\bullet :\tau^\bullet$}
\RightLabel{$cut$}
\BinaryInfC{$\Gamma^\bullet, x:\sigma^\bullet  \vdash \mathtt{let\ }\mathtt{E}_{\sigma^\bullet}\, x \mathtt{\ be\ }I \mathtt{\ in\ }M^\bullet :\tau^\bullet$}
\DisplayProof}
\\
\\
\left( 
\scalebox{0.6}{
\AxiomC{$\mathcal{D}_1$}
\noLine
\UnaryInfC{$\Gamma, x_1: \shpos \sigma, x_2: \shpos \sigma \vdash M: \tau$}
\AxiomC{$\mathcal{D}_2$}
\noLine
\UnaryInfC{$\vdash V: \sigma$}
\RightLabel{$c$}
\BinaryInfC{$\Gamma, x: \shpos \sigma \vdash \mathtt{copy}^{V}_{\sigma} \, x \mathtt{\ as \ } x_1, x_2 \mathtt{ \ in\ } M: \tau$}
\DisplayProof}
\right)^\bullet
\triangleq & 
\scalebox{0.6}{
\def\defaultHypSeparation{\hskip 0.4cm}
\def\ScoreOverhang{1pt}
\AxiomC{$\mathcal{D}_2^\bullet$}
\noLine
\UnaryInfC{$\vdash V^\bullet : \sigma^\bullet$}
\noLine
\UnaryInfC{$\vdots$}
\noLine
\UnaryInfC{$x:\sigma^\bullet \vdash \mathtt{D}^{V^\bullet } _{\sigma^\bullet} \, x: \sigma^\bullet \otimes \sigma^\bullet$}
\AxiomC{$\left( \begin{aligned} \begin{gathered}\mathcal{D}_1 \\   \Gamma, x_1: \shpos \sigma, x_2: \shpos \sigma  \vdash M :  \tau \end{gathered}\end{aligned} \right)^\bullet$}
\noLine
\UnaryInfC{\vdots}
\noLine
\UnaryInfC{$\Gamma^\bullet, y :\sigma^\bullet \otimes \sigma^\bullet \vdash  \mathtt{let\ }y \mathtt{\ be\ }x_1,x_2 \mathtt{\ in\ }M^\bullet: \tau^\bullet$}
\RightLabel{$cut$}
\BinaryInfC{$\Gamma^\bullet, x: \sigma^\bullet \vdash \mathtt{let\ }\mathtt{D}_{\sigma^\bullet}^{V^\bullet} \, x \mathtt{\ be\ }x_1,x_2 \mathtt{\ in\ }M^\bullet :\tau^\bullet$}
\DisplayProof}
\end{aligned}
\end{equation}
\caption{The translation of the rules $p$, $d$,  $w$ and $c$.}
\label{fig: translation inference rules} 
\end{figure}

The simulation theorem requires some preliminaries.

\begin{lem} \label{lem: translation on values}For every typable value $V$:
\begin{enumerate}[(1)]
\item \label{enum: 1 translation on values} $V^\bullet =V$.
\item \label{enum: 2 translation on values} $V$ has type $\sigma$ if and only if $V^\bullet$ has type $\sigma^\bullet$.
\end{enumerate}
\end{lem}
\begin{proof}
Straightforward consequence of Definition~\ref{defn: translation IMLL2 shpos with oplus into IMLL2}.
\end{proof}

\begin{lem} \label{lem: bullet commutes well with substitution in LAML} For all terms $M, N \in \Lambda_{\shpos}$ typable in $\mathsf{LEM}$, $M^\bullet [N^\bullet/x]= (M[N/x])^\bullet$.
\end{lem}
\begin{proof}
We can proceed by a standard structural induction on $ M $.
\end{proof}
\noindent
We now show that every reduction on terms typable in $\mathsf{LEM}$ can be simulated in the linear $\lambda$-calculus by means of the $\beta \eta$-reduction relation.  We recall that
Subject reduction holds on every typable $M \in \Lambda_\shpos$
(Theorem~\ref{thm:Subject Reduction for IMLL2shpos}). 
Moreover, every linear $\lambda$-term that has a type in $\mathsf{IMLL}$, has one in $\mathsf{IMLL}_2$ (see~\cite{hindley1989bck}). So, we state the simulation theorem 
for terms, rather than the related derivations.
\begin{thm}[Simulation]
\label{thm: translations for LAML} Let $\mathcal{D}\triangleleft \Gamma \vdash M: \sigma$ be a derivation in  $\mathsf{LEM}$. If $M_1 \rightarrow M_2$ then $M_1^\bullet\rightarrow^*_{\beta \eta}M_2^\bullet$:
\begin{center}
\begin{tikzcd}
M_1\arrow[d, dashed] \arrow[r]& M_2 \arrow[d,dashed] \\
M_1^\bullet \arrow[r, "*"  pos=1, "\beta \eta " ' pos=1.05] & M_2^{  \bullet}
\end{tikzcd}
\enspace.
\end{center}
\end{thm}
\begin{proof} We can proceed by structural induction on $M_1$.
One of the most interesting cases is when
$M_1$ is $ (\lambda x. P)Q$ and $M_2= P[Q/x]$. 
Lemma~\ref{lem: bullet commutes well with substitution in LAML}
implies $((\lambda y. P)Q)^\bullet= (\lambda y. P^\bullet)Q^\bullet \rightarrow_\beta P^\bullet[Q^\bullet/x] =(P[Q/x])^\bullet$. 
If $M_1$ is $\mathtt{discard}_{\sigma}\, V \mathtt{\ in\ }N$ and
$M_2$ is $N$, then $V$ is a value of type $\sigma$. By Lemma~\ref{lem: translation on values}.\ref{enum: 2 translation on values}, $V^\bullet$ is a value of $\sigma^\bullet$.  Hence:
\allowdisplaybreaks
\begin{align*}
(\mathtt{discard}_{\sigma}\, V \mathtt{\ in\ }N)^\bullet &
\triangleq \mathtt{let\ } \mathtt{E}_{\sigma^\bullet}\, V^\bullet \mathtt{\ be\ }I \mathtt{\ in\ }N^\bullet
\rightarrow^*_\beta  N^\bullet
\end{align*}
\noindent
by Theorem~\ref{thm: pi1 types are erasable}.
If $M_1$ is $\mathtt{copy}^{U}_{\sigma}\, V \mathtt{\ as\ }x_1,x_2 \mathtt{\ in\ }N$ and $M_2$ is $N[V/x_1, V/x_2]$, 
then $U$ and $V$ are both values of type $\sigma$. By Lemma~\ref{lem: translation on values}.\ref{enum: 2 translation on values}, $U^\bullet$ and $V^\bullet$ are both values of type $\sigma^\bullet$. 
Hence:
\allowdisplaybreaks
\begin{align*}
(\mathtt{copy}^{U}_{\sigma}\, V \mathtt{\ as\ }x_1,x_2 \mathtt{\ in\ 
}N)^\bullet&\triangleq  \mathtt{let\ }\mathtt{D}^{U^\bullet } _{\sigma^\bullet}\, V^\bullet
\mathtt{\ be\ } x_1,x_2 \mathtt{\ in \ }N^\bullet \\
&\rightarrow^*_{\beta \eta}  \mathtt{let\ }\langle V^\bullet, V^\bullet \rangle \mathtt{\ be\ } x_1,x_2 \mathtt{\ in \ }N^\bullet  &&\text{Theorem}~\ref{thm: pi1 are duplicable} \\
&\rightarrow^*_\beta N^\bullet[V^\bullet/x_1, V^\bullet/x_2]\\
&\triangleq (N^\bullet[V^\bullet/x_1]) [ V^\bullet/x_2]\\
&=    ((N[V/x_1]) [ V/x_2])^\bullet &&\text{Lemma}~\ref{lem: bullet commutes well with substitution in LAML}\\
&\triangleq (N[V/x_1, V/x_2])^\bullet   \enspace .
\end{align*}
\end{proof}
\noindent
We conclude by estimating the impact of the translation  on the size of terms  produced by  \enquote{unpacking} the constructs $\mathtt{discard}_{\sigma}$ and $\mathtt{copy}^V_{\sigma}$. 
We need to bound the dimension of erasers and duplicators with ground type $A$. We rely on the map $(\_ )^-$ 
(Definition~\ref{defn: stripping forall} 
in~\ref{sec: the d-soundness theorem DICE})
that, intuitively, strips every occurrence of $\forall$ away from a
given type (Remark~\ref{rem: duplicator} in~\ref{sec: the d-soundness theorem DICE}.)

\begin{lem}[Size of duplicators and erasers] 
\label{lem: size of duplicator}  For every ground type $A$:
\begin{enumerate}[(1)]
\item  \label{enum: size of eraser}$\vert \mathtt{E}_{A} \vert \in \mathcal{O}(\vert A^- \vert)$.
\item\label{enum: size of duplicator} If $V$ is a value of $A$, then  $\vert \mathtt{D}_{A}^{V} \vert \in \mathcal{O}(  2^{\vert A^- \vert^2})$.
\end{enumerate}
\end{lem}
\begin{proof}
Point~\ref{enum: size of eraser} is straightforward by looking at the proof of Theorem~\ref{thm: pi1 types are erasable}. Concerning Point~\ref{enum: size of duplicator}, from 
\ref{sec: the d-soundness theorem DICE} we know that $\mathtt{D}_{A}^{V}$ is
$ \mathtt{dec}^s_{A} \circ  
\mathtt{enc}_{A}^s  \circ
\mathtt{sub}^s_{A} $,
where  $s= \mathcal{O}(\vert  A^- \vert \, \cdot \, \log \vert A^- \vert )$. The components of $\mathtt{D}_{A}^{V}$ with a
size not linear in $\vert A^- \vert $ are $\mathtt{dec}^s_{A}$ and  $\mathtt{enc}^s_{A}$. The $ \lambda $-term $\mathtt{dec}^s_{A}$ (see point~\ref{enumerate: duplication step 3} in Section~\ref{sec:  the system shpos, with, oplus  IMLL2 cut elimination and bound}) nests occurrences of
\texttt{if}-\texttt{then}-\texttt{else} each containing $2^{s}$ pairs of normal inhabitants of $A$, every of which,
by Lemma~\ref{lem: pi1type bound}, has size bounded by 
$\vert A^- \vert$. 
Similarly, $\mathtt{enc}^s_{A}$ alternates instances of
$\lambda$-terms $\mathtt{abs}^s$ and $\mathtt{app}^s$ which, 
again, nest occurrences of \texttt{if}-\texttt{then}-\texttt{else} 
every one with $2^{s}$ instances of boolean strings of size $s$. 
The overall complexity of $\mathtt{D}_{A}^{V}$ is $\mathcal{O}(s \cdot 2^{s})= \mathcal{O}(2^{\vert A^-  \vert^2})$. 
\end{proof}
\noindent

\begin{thm}[Exponential Compression]
\label{thm: exponential compression} 
Let  $\mathcal{D} \triangleleft \Gamma \vdash M: \sigma$ be a derivation in 
$\mathsf{LEM}$. Then,  
$\vert M^\bullet \vert = \mathcal{O}(2^{\vert M \vert^k})$, 
for some $k \geq 1$.
\end{thm}
\begin{proof} The proof is by structural induction on $M$. The  interesting case is when $M$ is 
$\mathtt{copy}^{V}_{\sigma}\, P \mathtt{\ as\ }x_1,x_2 \mathtt{\ in\ }Q$. 
Since $M$ is typable,  $V$ has type $\sigma$, that is  closed and free from negative occurrences of $\forall$, hence lazy. 
By Lemma~\ref{lem: lazy derivation properties}.\ref{enum: finite number of normal forms}, it is safe to assume that $V$ is a value with largest size among values of type $\sigma$.
By Lemma~\ref{lem: translation on values}, $V=V^\bullet$ is also the largest value of type $\sigma^\bullet$ in $\mathsf{IMLL_2}$. Finally, by Lemma~\ref{lem: pi1type bound}, this implies that $V$  is a $\eta$-long normal form of type $\sigma^\bullet$.  Now, by using  Definition~\ref{eqn:datatype unity and product}, we have   $\vert M^\bullet \vert= \vert \mathtt{let\ }\mathtt{D}^{V} _{\sigma^\bullet}\, P^\bullet \mathtt{\ be\ } x_1,x_2 \mathtt{\ in \ }Q^\bullet\vert = \vert \mathtt{D}^{V} _{\sigma^\bullet}\vert + \vert  P^\bullet \vert+ \vert Q^\bullet \vert +4$.  On the one hand, by induction hypothesis, we obtain  $\vert   P^\bullet \vert = \mathcal{O}(2^{\vert P \vert^{k'}})$ and  $\vert   Q^\bullet \vert = \mathcal{O}(2^{\vert Q \vert^{k''}})$, for some $k', k'' \geq 1$.  On the other hand, by applying both Lemma~\ref{lem: size of duplicator} and  Lemma~\ref{lem: pi1type bound},  we have $\vert \mathtt{D}^{V} _{\sigma^\bullet}\vert=\mathcal{O}(2^{\vert A^- \vert^2}) =\mathcal{O}(2^{2 \cdot \vert V \vert^2})$. 
%
Therefore, there exists $k\geq 1$ such that $\vert M^\bullet \vert = \mathcal{O}(2^{(\vert V \vert + \vert P \vert + \vert Q \vert +1)^k})=\mathcal{O}(2^{\vert M \vert^k })$.
\end{proof}
\section{The expressiveness of $\mathsf{LEM}$ and applications}
\label{section:The expressiveness of IMLLshpos2 and applications}
Theorem~\ref{thm: translations for LAML} says that $\mathsf{LEM}$  is not  \enquote{algorithmically} more expressive than  $ \mathsf{IMLL}_2 $. Nonetheless, terms with type in  $\mathsf{LEM}$,  and their evaluation mechanisms, exponentially compress the corresponding linear $ \lambda $-terms and evaluations in $ \mathsf{IMLL}_2 $ (Theorem~\ref{thm: exponential compression}).  The goal of this section is to explore  the benefits of this compression.


\begin{figure}[ht]
\centering
 \subfigure[\scriptsize{From left,  input,  internal,  fan-out, and  output nodes.}	\label{figure:node of boolean circuits} ]{
\scalebox{.65}{
 \begin{tikzpicture}[
 node/.style={circle,minimum size=1pt,fill=white},
 node distance=0.25cm
 ]

 \node[circle, inner sep =1pt, draw](op){\scriptsize{$\textsl{op}^n$}};
 \node[circle]  at  ($(op)+(0, +1)$)(ab){\scriptsize{\ldots $n\geq 0$\ldots}} ;
   \coordinate[right = of ab](abri){} ;
   \coordinate[left = of ab](able){} ;
   \draw[->](able)to[out=270, in=135](op);
      \draw[->](abri)to[out=270, in=45](op);
       \draw[->](op)--($(op)+(0, -1)$) ;

 \node[node,  draw, left=3cm  of op](in){\scriptsize{$\textrm{x}$}};
   \draw[->](in)--($(in)+(0, -1)$) ;  
   \node[circle, inner sep= 1pt, draw, right=3cm of op, fill=black](fan){};
  \node[node]  at ($(fan)+(0, -1)$)(befan){\scriptsize{\ldots $n\geq 0$\ldots}} ;
   \node[circle] at  ($(fan)+(0, +1)$)(abfan){} ;
   \coordinate[right = of befan](berifan){} ;
   \coordinate[left = of befan](belefan){} ;
     \draw[->, thick](fan)to[out=225, in=90](belefan);
      \draw[->,thick](fan)to[out=315, in=90](berifan);
         \draw[->](abfan)to[out=270, in=90](fan);        
   \node[node,  draw, right=4cm  of belefan](out){\scriptsize{$\mathrm{y}$}};
   \draw[<-](out)--($(out)+(0, 1)$) ;     
 \end{tikzpicture}
 }}

\subfigure[\scriptsize{The 2-bits full-adder boolean circuit}\label{figure:2-bits full adder}]{
\scalebox{.65}{
 \begin{tikzpicture}[
 auto=left,
 empty/.style = {minimum size=0cm,fill=white},
 el/.style = {inner sep=2pt, align=left, sloped},
 node/.style={circle,minimum size=.9cm,fill=white},
 sqnode/.style={rectangle,inner sep=2pt,fill=white},
 node distance=1cm
 ]
 \draw(0, 1)node(q){};
 \node[node,draw            ] (b1)  {$ \mathrm{x}_1 $};
 \node[node,draw,below=of b1] (b2)  {$ \mathrm{x}_2 $};
 \node[node,draw,below=of b2] (cin) {$ \mathrm{x}_{\textrm{in}} $};
 
 \node[circle, fill=black, inner sep=1pt, right=of b1, draw ](fop1)  {};
 \node[circle, fill=black, inner sep=1pt, right=of b2, draw ](fop2){};

 \node[node,draw,right=2cm of b1] (xor1)  {$\bigoplus$};
 \node[circle, fill=black, inner sep=1pt,  right=of xor1] (foxor1) {};
 \node[node,draw,right=2cm of b2] (and1)  {$\bigwedge$};
 
 \node[node,draw,right= of foxor1] (xor2) {$\bigoplus$};
 \draw ($ (xor2)+(0,-3.825) $) node[node,draw] (and2) {$\bigwedge$};
  \node[circle, fill=black, inner sep=1pt, draw, left= of and2] (focin) {};
  
  \node[node, below=of xor2](preor){};
 \node [node,draw, right= of preor] (or) {$\bigvee$ };

 \node[node, right=of xor2, draw] (output1){$ \mathrm{y}_1 $};
 \node[node, right=of or, draw]  (output2){$ \mathrm{y}_2 $};
 \draw (xor2) to[out=0, in=180]node[above]{\scriptsize{$s$}} (output1);
  \draw (or) to[out=0, in=180]node[above]{\scriptsize{$y_{\text{out}}$}} (output2);
 \draw ($ (output1)+(1.75,0) $) node[sqnode] (s) {\scriptsize{$(x'_1\oplus x'_2)\oplus y'$}};
 \draw ($ (output2)+(2.50,0) $) 
   node[sqnode] (s) {\scriptsize{$(x''_1\wedge x''_2)\vee((x'_1\oplus x'_2)\wedge y'')$}};

 \draw (b1)   to[out=0,in=180]node[above]{\scriptsize{$x_1$}} (fop1);
 \draw (b2)   to[out=0,in=180]node[above]{\scriptsize{$x_2$}} (fop2);
 \draw (cin)  to[out=0,in=180]node [above]{\scriptsize{$y_{\text{in}}$}}(focin);
 
 \draw[thick](fop1)  to[out=0,in=180] node[above  ] {\scriptsize{$ x'_1  $}} (xor1);
 \draw[thick](fop1)  to[out=0,in=135] node[below,pos=0.2] {\scriptsize{$ x''_1\ \ $}} (and1);
 
 \draw[ thick] (fop2)  to[out=0,in=225] node[above,pos=0.2] {\scriptsize{$ x'_2\ \ $}} (xor1);
 \draw[ thick] (fop2)  to[out=0,in=180] node[below        ] {\scriptsize{$ x''_2 $}} (and1);
 
 \draw[ thick] (focin) to[out=0,in=225] node[above, pos=0.1] {\scriptsize{$ y'\ \ $}} (xor2);
 \draw[ thick] (focin) to[out=0,in=180] node[below] {\scriptsize{$ y''$}} (and2);

 \draw (xor1)   to[out=  0,in=180]node[above]{\scriptsize{$z_1$}}                                     (foxor1);
 \draw[ thick] (foxor1) to[out=  0,in=180] node[above]       {\scriptsize{$ z_1' $}}  (xor2);
 \draw[ thick] (foxor1) to[out=0,in=135] node[below,pos=0.1] {\scriptsize{$z''_1 \ \ $}} (and2);

 \draw (and1) to[out=0,in=180] node[above,pos=0.1] {\scriptsize{$ z_2 $}} (or);
 \draw (and2) to[out=0,in=225] node[below, right] {\scriptsize{$z_3 $}} (or);
 \end{tikzpicture}
} 
}
\subfigure[\scriptsize{The 3-bits majority boolean circuit}
\label{figure:maj3}]{
\scalebox{.65}{
\begin{tikzpicture}[
auto=left,
empty/.style = {minimum size=0cm,fill=white},
el/.style = {inner sep=2pt, align=left, sloped},
node/.style={circle,minimum size=.9cm,fill=white},
sqnode/.style={rectangle,inner sep=2pt,fill=white},
node distance=1cm
]
 \draw(0, 1)node(q){};
\node[node,draw            ] (b1)  {$ \mathrm{x}_1 $};
\node[node,draw,below=of b1] (b2)  {$ \mathrm{x}_2 $};

\node[circle, fill=black, inner sep=1pt, right=of b1, draw] (fob1){}; 
\node[circle, fill=black, inner sep=1pt, right=of b2, draw ] (fob2){};

\node[node,draw,right=of fob1] (o1)  {$\bigvee$};
\node[node,draw,right=of fob2] (a1)  {$\bigwedge$};
\node[node,draw,below=of b2  ] (b3)  {$ \mathrm{x}_3 $};

\coordinate[right=of a1] (foa1);

\node[node,draw,right=of foa1] (o2)  {$\bigvee$};
\node[node,draw,below=of o2] (a2)  {$\bigwedge$};

\node[circle, fill=black, inner sep=1pt, right=of a1, draw ] (foa1){};
\node[node, below= of a1](beforefob3){};
\node[circle, fill=black, inner sep=1pt, right= of beforefob3, draw] (fob3){};
\node[circle, fill=black, inner sep=1pt, right=of o2, draw ]   (foo2){};

\node[node,draw,right=of foo2] (a3)  {$\bigwedge$};
\node[node,draw,above=of a3] (o3)  {$\bigvee$};
\node[circle, fill=black, inner sep=1pt, left= of o3, draw] (foo1){};
\draw ($ (o3)+(2,-1) $) node[node,draw] (a) {$\bigwedge$};

\node[node, right = of a,draw ]  (output1){$\mathrm{y_1}$};
\node[node, right=of a2, draw] (output2){$\mathrm{y_2}$};
\draw (a) to[out=0, in=180]node[above]{\scriptsize{$m$}} (output1);
 \draw (a2) to[out=0, in=180]node[above]{\scriptsize{$g$}} (output2);
\draw ($(output1)+(2,0)$)node[sqnode] (grabage) {\scriptsize{$ \textsl{maj}_3(x_1',x_2',x_3')$}};
\draw ($(output2)+(2,0)$)node[sqnode] (out) {\scriptsize{$ \textsl{min}_3\{x_1'',x_2'',x_3''\} $}};

\draw (b1)   to[out=  0,in=180]node[above]{\scriptsize{$x_1$}} (fob1);
\draw[thick] (fob1) to[out=  0,in=180] node[above] { \scriptsize{$x_1'$}} (o1);
\draw[thick] (fob1) to[out=  0,in=135] node[below, pos=0.2] { \scriptsize{$x_1''\ \ $}} (a1);
\draw (b2)  to[out=  0,in=180] node[above]{\scriptsize{$x_2$}}(fob2);
\draw[thick] (fob2)  to[out= 0,in=225] node[above, pos=0.2] {\scriptsize{ $x_2'\ \ \ $}} (o1);
\draw[thick] (fob2)  to[out= 0,in=180] node[below] { \scriptsize{$x_2''$}}(a1);

\draw (o1)  to[out=  0,in=180] (foo1);
\draw (a1)  to[out=  0,in=180]node[above]{\scriptsize{$y_2$}} (foa1);
\draw[thick] (foa1)  to[out=0,in=135] node[below, pos=0.2] {\scriptsize{$ y_2''\ \ $}} (a2);
\draw [thick](foa1)  to[out=0,in=180] node[above] {\scriptsize{$ y_2'$}} (o2);
\draw (b3)  to[out= 0,in=180]node[above]{\scriptsize{$x_3$}} (fob3);
\draw[thick] (fob3)  to[out=  0,in=225] node[above, pos=0.2] {\scriptsize{$ x_3'\ \ \ $}}  (o2);
\draw[thick] (fob3)  to[out=  0,in=180] node[below] {\scriptsize{$ x_3''$}} (a2);

\draw (o1)  to[out= 0,in=180]node[above]{\scriptsize{$y_1$}} (foo1);
\draw[thick] (foo1)  to[out= 0,in=180] node[above] {\scriptsize{$ y_1'$}} (o3);
\draw[thick] (foo1)  to[out= 0,in=135] node[below, pos=0.2] {\scriptsize{$ y_1''\ \ $}} (a3);
\draw (o2)  to[out=  0,in=180]node[above]{\scriptsize{$y_3$}} (foo2);
\draw[thick] (foo2)  to[out=  0,in=225] node[above, pos=0.2] {\scriptsize{ $y_3'\ \ \  $}}  (o3);
\draw[thick] (foo2)  to[out=  0,in=180] node[below] {\scriptsize{$ y_3''$}} (a3);

\draw (o3)  to[out=  0,in=135]node[above, right]{\scriptsize{$\ z_1$}} (a);
\draw (a3)  to[out=  0,in=225]node[below, right]{\scriptsize{$\ z_2$}} (a);
\end{tikzpicture}
} 
}
\caption{Nodes of boolean circuits and some examples. 
Writing, for example,
$ \textsl{maj}_3\{x_1',x_2',x_3'\} $ in place of 
$ \textsl{maj}_3\{x_1,x_2,x_3\} $ would be equivalent. 
The current notation just highlights which is the component of the fan-out nodes that an output depends on.}
\label{fig: some examples of boolean circuits}
 \end{figure}
 
\begin{figure}
\scalebox{0.85}{
\bgroup
\def\arraystretch{1.5}%
\begin{tabular}{l|ll} 
 $\neg$ & $\mathtt{not}  \triangleq\lambda b. \lambda x. \lambda y. b yx $&$:\mathbf{B} \multimap\mathbf{B}$
 \\ \hline
$\wedge^0$   &$\mathtt{and}^0 \triangleq \lambda x.\lambda y. \langle x,y \rangle$ &$:\mathbf{B}$
\\
$\wedge^1$   &$\mathtt{and}^1\triangleq I$ &$:\mathbf{B}\multimap \mathbf{B}$
\\
$\wedge^2$&  $ \mathtt{and}^2  \triangleq 
\lambda x_1. \lambda x_2.   \pi_1(x_1 x_2\, \mathtt{ff})$ &$:
\mathbf{B}\multimap  \mathbf{B} \multimap \mathbf{B}$
\\
$\wedge^{n+2} $ &  $\mathtt{and}^{n+2}\triangleq \lambda x_1\ldots  x_{n+1}x_{n+2}. \mathtt{and}^2\, (\mathtt{and}^{n+1}\, x_1\,   \ldots\, x_{n+1})\, x_{n+2}$ &$:\mathbf{B}\multimap \overset{n+2}{\ldots}  \multimap \mathbf{B}\multimap\mathbf{B}$
\\ \hline
$\vee^0$ &$ \mathtt{or}^0 \triangleq \lambda x.\lambda y. \langle y,x \rangle$ &$:\mathbf{B}$\\
$\vee^1$   &$\mathtt{or}^1 \triangleq  I $ &$:\mathbf{B}\multimap \mathbf{B}$
\\
$\vee^2$ & $\mathtt{or}^2  \triangleq \lambda x_1. \lambda x_2. \pi_1(x_1 \mathtt{tt}\, x_2)$ &$:
\mathbf{B}\multimap \mathbf{B} \multimap \mathbf{B}$
\\
$\vee^{n+2} $ &  $\mathtt{or}^{n+2}\triangleq \lambda x_1\ldots  x_{n+1}x_{n+2}. \mathtt{or}^2\, (\mathtt{or}^{n+1}\, x_1\,   \ldots\, x_{n+1})\, x_{n+2}$ &$: \mathbf{B}\multimap \overset{n+2}{\ldots}  \multimap \mathbf{B}\multimap\mathbf{B}$\\ \hline
\textsl{fo}$^0$& $\mathtt{out}^0\triangleq \lambda x. \mathtt{discard}_{\mathbf{B}}\, x \mathtt{ \ in \ } I$ &$: \shpos \mathbf{B}\multimap \mathbf{1}$\\
\textsl{fo}$^1$& $\mathtt{out}^1\triangleq  I$ &$: \shpos \mathbf{B}\multimap \mathbf{B}$ \\
\textsl{fo}$^2$& $\mathtt{out}^2\triangleq \lambda x. 
\texttt{copy}_{\mathbf{B}}^{\mathtt{tt}}\, x\ \texttt{as}\ x_1,x_2\ \texttt{in}\ \langle x_1, x_2\rangle$ &$:\shpos \mathbf{B}\multimap \shpos \mathbf{B}\otimes  \shpos \mathbf{B}$ \\
\textsl{fo}$^{n+2}$& $\mathtt{out}^{n+2}\triangleq \lambda x. 
\texttt{copy}_{\mathbf{B}}^{\mathtt{tt}}\, x\ \texttt{as}\ x_1,x_2\ \texttt{in}\ 
\langle 
\mathtt{out}^{n+1}\, x_1,
x_2
\rangle $ &$: \shpos \mathbf{B}\multimap \shpos \mathbf{B} \otimes \overset{n+2}{\ldots} \otimes \shpos \mathbf{B}$ \\
\end{tabular}
\egroup 
}
\caption{Encoding of boolean functions and fan-out.}
\label{fig: enc boolean functions and fanout}
\end{figure}

\subsection{Boolean circuits in $\mathsf{LEM}$}
\label{sec: encoding boolean circuits}
We encode boolean circuits as terms of $\mathsf{LEM}$
(Definition~\ref{defn: from boolean circuits to terms}) and 
we prove a simulation result 
(Proposition~\ref{prop: simulation for boolean circuits}).

The encoding is inspired by Mairson\&Terui~\cite{mairsonlinear}. Other encodings of the boolean circuits have been shown in  Terui~\cite{DBLP:conf/lics/Terui04}, 
Mogbil\&Rahli~\cite{DBLP:conf/lfcs/MogbilR07} and 
Aubert~\cite{DBLP:journals/corr/abs-1201-1120} by considering  the  \textit{unbounded} proof-nets for the multiplicative fragment \textsf{MLL} of Linear Logic. Unbounded proof-nets are an efficient language able to express $n$-ary tensor products  by single nodes and to characterize parallel computational complexity classes such as \textsf{NC}, \textsf{AC}, and $\textsf{P}/_{\operatorname{poly}}$. The contribution of this work to these encodings is in the use of   \texttt{copy} and \texttt{discard}  to directly express the fan-out nodes that  allow a more compact and  modular representation of circuits. In particular, as compared to~\cite{DBLP:conf/lics/Terui04,DBLP:conf/lfcs/MogbilR07,DBLP:journals/corr/abs-1201-1120}, our encoding is able to  get rid of the garbage  that accumulates in the course of the simulation.

We start by briefly recalling the basics of boolean circuits from 
Vollmer~\cite{Vollmer:1999:ICC:520668}.

\begin{defn}[Boolean circuits] 
A \textit{boolean circuit} $ C $ is a finite, directed and acyclic graph with $n$ \emph{input nodes}, $m$ \emph{output  nodes}, \emph{internal nodes} and \textit{fan-out nodes} as in Figure~\ref{figure:node of boolean circuits}. 
The incoming (resp.~outgoing) edges of a node are \textit{premises} 
(resp.~\textit{conclusions}). The \textit{fan-in} of an internal node is the number of its premises. 
Labels for the $n$ input nodes of $ C $ are  $\mathrm{x}_1, \ldots, \mathrm{x}_n$ and those ones for the $m$ outputs are $\mathrm{y}_1, \ldots, \mathrm{y}_m$. 
Each internal node with fan-in $n \geq 0$ has an $n$-ary boolean function $\textsl{op}^n$ as its label, provided that if $n=0$, then $\textsl{op}^n$ is a boolean constant in $\lbrace 0,1 \rbrace$. 
The fan-out nodes have no label. 
Input and internal nodes are \textit{logical nodes} and their conclusions are \textit{logical edges}. If $\nu$ and $\nu'$ are logical nodes, then $\nu'$ is a \textit{successor} of $\nu$ if a directed path from $\nu$ to $\nu'$ exists which crosses no logical node.
The \emph{size} $ |C| $ of $C$ is the number of nodes. Its \emph{depth} $ \delta(C) $ is the length of the longest path from an input node to a output node. 
A \textit{basis} $\mathcal{B}$ is a set of boolean functions. A boolean circuit $C$ is \textit{over a basis} $\mathcal{B}$  if the label of every of its internal nodes belong to $\mathcal{B}$ only. The standard unbounded fan-in basis is $\mathcal{B}_1= \lbrace \neg, (\wedge^n)_{n \in \mathbb{N}},  (\vee^n)_{n \in \mathbb{N}}\rbrace$. 
\qed
\end{defn}
\noindent
When representing boolean circuits as terms we label edges by $\lambda$-variables, we omit their orientation, we assume that every fan-out always has a logical edge as its premise and we draw non-logical edges, i.e.~conclusions of fan-out nodes, as thick lines.
Figures~\ref{figure:2-bits full adder} and~\ref{figure:maj3} are examples.
A 2-bits full-adder is in the first one. 
It takes two bits $ x_1, x_2 $ and a carrier $ y_{\textrm{in}} $ as inputs.
Its outputs are the sum $ s = (x'_1\oplus x'_2)\oplus y' $ and the 
carrier $ y_{\text{out}} = (x''_1\wedge x''_2)\vee((x'_1\oplus x'_2)\wedge y'')$, where $ \oplus $ is the exclusive or that we can
obtain by the functionally complete functions in $ \mathcal{B}_1 $.
Figure~\ref{figure:maj3} is the 3-bits majority function $ \textsl{maj}_3(x_1,x_2,x_3)$. 
It serially composes three occurrences of the boolean circuit that switches two inputs $ x_1 $ and $ x_2 $ in order to put the greatest on the topmost output and the smallest on the bottommost one, under the convention that $ 0$ is smaller than $1$.
%
%
%
%
%
So, the 3-bits majority circuit first sorts its input bits and then checks if the topmost two, i.e. the majority, are both set to $1$.
The lowermost output is garbage. 

Translating boolean circuits as terms of $\mathsf{LEM}$ requires to encode the boolean functions in $\mathcal{B}_1$ and the fan-out nodes. 
Figure~\ref{fig: enc boolean functions and fanout} reports them, where $\mathtt{tt}$ and $\mathtt{ff}$ encode the boolean values in~\eqref{eqn: boolean data type}, and $\pi_1$ is the projection in~\eqref{eqn: boolean projection}. 
As a typographical convention, $\mathtt{i}\in \lbrace \mathtt{tt}, \mathtt{ff} \rbrace$ will code the boolean constant $i \in \lbrace 0,1 \rbrace$, 
and $\mathtt{op}^n$ the $n$-ary boolean function $\textsl{op}^n$, according to Figure~\ref{fig: enc boolean functions and fanout}. 
We shorten $\mathtt{and}^0$, $\mathtt{or}^0$,  $\mathtt{and}^2$, $\mathtt{or}^2$, and $\mathtt{out}^2$ as $\mathtt{tt}$, $\mathtt{ff}$, $\mathtt{and}$, $\mathtt{or}$, and $\mathtt{out}$, respectively. The encoding of  the binary exclusive or $\oplus$ is $\mathtt{xor}$.

We recall that boolean circuits are a model of parallel computation, while the $\lambda$-calculus models sequential computations. Mapping the former into the latter requires some technicalities. The notion of \emph{level} allows to topologically sort the structure of the boolean circuits in order to preserve the node dependencies:
\begin{defn}[Level] The \textit{level} $l$ of a logical node $\nu$ in a boolean circuit $C$ is:
\begin{enumerate}
\item $ 0 $ if $\nu$ has no  successors, and
\item $\max \lbrace l_1, \ldots, l_k  \rbrace +1$ if $\nu$ has  successors $\nu_1, \ldots, \nu_k$ with levels $l_1, \ldots, l_k$.
\end{enumerate}
The \textit{level} of a logical edge is the level of the logical node it is the conclusion of. 
The \textit{level} of a boolean circuit is the greatest level of its logical nodes.
\qed
\end{defn}
We define a level-by-level translation of unbounded fan-in boolean circuits over $\mathcal{B}_1$ into  terms typable in $\mathsf{LEM}$ taking inspiration from Schubert~\cite{schubert2001complexity}:
\begin{defn}[From boolean circuits to terms]\label{defn: from boolean circuits to terms}
Let $C$ be a boolean circuit with $n$ inputs and $m$ outputs.  We define the term $\mathtt{level}^l_C$ by induction on $l-1$:
\begin{enumerate}
\item $\mathtt{level}_{C}^{-1} \triangleq \langle x_1, \ldots, x_n  \rangle$, where $x_1, \ldots, x_n$ are  the variables labelling the logical edges of level $0$.
\item $\mathtt{level}_{C}^{l} \triangleq (\lambda x_1\ldots x_n x_{n+1}\ldots x_{m}. \mathtt{let}\, (\mathtt{out}^{k_{1}}\, x_{1}) \ \mathtt{\ be \ }y^1_1, \ldots, y^1_{k_1} \mathtt{ \ in \ } \ldots$\\
$\mathtt{let}\, (\mathtt{out}^{k_{n}}\, x_{n})  \mathtt{\ be \ }y^n_1, \ldots, y^n_{k_n} \mathtt{ \ in \ }  \mathtt{level}_{C}^{l-1})\, B_1 \ldots B_m$, where:
\begin{enumerate}
\item $x_1,\ldots ,x_n,x_{n+1},\ldots, x_{m}$ are the variables labelling the logical edges of level  $ l $, 
\item  for all $1 \leq j \leq n$,   $x_j$ is the premise of a fan-out node with conclusions labelled with   $y^j_1, \ldots, y^j_{k_j}$ (see Figure~\ref{fig: level fan out e internal nodes}).
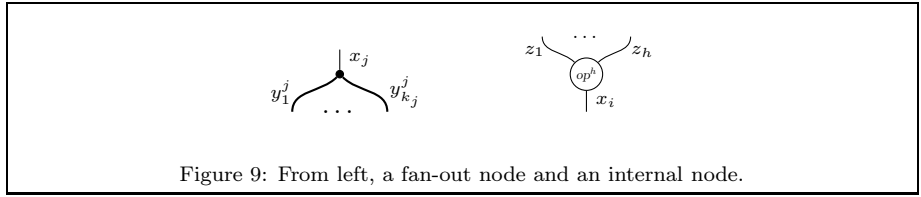
\begin{figure}
\centering
\begin{tikzpicture}[node/.style={circle,minimum size=1pt,fill=white},
 node distance=0.25cm]
  \node[circle, inner sep= 1pt, draw, fill=black](fan){};
  \node[node]  at ($(fan)+(0, -0.50)$)(befan){\ldots} ;
   \node[circle] at  ($(fan)+(0, +0.50)$)(abfan){} ;
   \coordinate[right = of befan](berifan){} ;
   \coordinate[left = of befan](belefan){} ;
     \draw[-, thick](fan)to[out=225, in=90]node[left]{\scriptsize{$y^j_1$}\ \ }(belefan);
      \draw[-,thick](fan)to[out=315, in=90]node[right]{\ \scriptsize{$y^j_{k_j}$}}(berifan);
         \draw[-](abfan)to[out=270, in=90]node[right]{\scriptsize{$x_j$}}(fan);   

 \node[circle, inner sep =1pt, draw, right=3cm of fan](op){\scalebox{0.5}{$\textsl{op}^h$}};
 \node[circle]  at  ($(op)+(0, +0.50)$)(ab){\scriptsize{\ldots}} ;
   \coordinate[right = of ab](abri){} ;
   \coordinate[left = of ab](able){} ;
   \draw[-](able)to[out=270, in=135]node[left]{\scriptsize{$z_1$}}(op);
      \draw[-](abri)to[out=270, in=45]node[right]{\scriptsize{$\, z_h$}}(op);
       \draw[-](op)--node[right]{\scriptsize{$x_i$}} ($(op)+(0, -0.50)$) ;
\end{tikzpicture}
\caption{From left, a fan-out node and an internal node.}
\label{fig: level fan out e internal nodes}
\end{figure}
\item for all $1 \leq i \leq m$,  if $x_i$ is the variable labeling the conclusion of an internal node $\textsl{op}^h$  with premises labeled  by $z_1,  \ldots,  z_h$, respectively (see Figure~\ref{fig: level fan out e internal nodes}), then $B_i\triangleq \mathtt{op}^h\, z_1\,  \ldots\,  z_h$. If $x_i$ is the variable labelling the conclusion of an input node  then $B_i \triangleq x_i$.
\end{enumerate}
\end{enumerate}
Last, if the input nodes have conclusions labeled by 
$x_1, \ldots, x_{n}$, respectively, and if $C$ has level $l$,
then we define  
$\lambda(C) \triangleq  \lambda x. \mathtt{let}\ x \mathtt{  \ be \ }x_1, \ldots, x_{n} \mathtt{ \ in \ }  \mathtt{level}_C^l$.
\qed
\end{defn}
\begin{exmp}[$2$-bits full adder] The level-by-level translation of the boolean circuit $C$  in Figure~\ref{figure:2-bits full adder} is the following:
\label{exmp:2-bits full-adder}
\allowdisplaybreaks
\begin{align*}
\mathtt{level}_{C}^{-1}&\triangleq \langle s, y_{\text{out}} \rangle\\
\mathtt{level}_{C}^0&\triangleq (\lambda s. \lambda y_{\text{out}}. \mathtt{level}^{-1}_C)(\mathtt{xor}\, z'_1\, y')(\mathtt{or}\, z_2\, z_3)\\
\mathtt{level}_{C}^1&\triangleq (\lambda z_2. \lambda z_3. \mathtt{level}_{C}^0)(\mathtt{and}\, x''_1\, x''_2)(\mathtt{and}\, z''_1\, y'') \\
\mathtt{level}_{C}^2&\triangleq(\lambda z_1.\lambda y_{\text{in}}. \mathtt{let}\,(\mathtt{out}\, z_1)\  \mathtt{be}\ z'_1, z''_1 \ \mathtt{ \ in \ }\\
&
\phantom{\ \triangleq\ \qquad }
\mathtt{let}\,(\mathtt{out}\, y_{\text{in}}) \mathtt{ \ be\  }y', y'' \mathtt{ \ in \ }\mathtt{level}_C ^1     )(\mathtt{xor}\, x'_1\, x'_2)\, y_{\text{in}}\\
\mathtt{level}_{C}^3&\triangleq (\lambda x_1. \lambda x_2.  \mathtt{let}\,(\mathtt{out}\, x_1) \mathtt{ \ be\  }x_1', x_1'' \mathtt{ \ in \ } \\
&
\phantom{\ \triangleq\ \qquad }
\mathtt{let}\,(\mathtt{out}\, x_2) \mathtt{ \ be\  }x_2', x_2'' \mathtt{ \ in \ } \mathtt{level}_C^2)\, x_1\, x_2 
\end{align*}
where we set $\lambda(C)\triangleq \lambda x. \mathtt{let }\, x \mathtt{\ be \ }  x_1,  x_2,  y_{\text{in}} \mathtt{ \ in \ } \mathtt{level}_C ^3$ which reduces to:
\begin{align*}
&\lambda x. \mathtt{let }\, x \mathtt{\ be \ }  x_1,  x_2,  y_{\text{in}} \mathtt{ \ in \ } (\mathtt{let}\, (\mathtt{out}\, x_1) \mathtt{\ be \ }x'_1, x''_1 \mathtt{ \ in \ }\\
&\mathtt{let}\, (\mathtt{out} \, x_2) \mathtt{\ be \ }x'_2, x''_2 \mathtt{ \ in \ }( \mathtt{let}\, (\mathtt{out}\, y_{\text{in}})\mathtt{\ be \ }y', y'' \mathtt{ \ in \ }\\
&\mathtt{let}\,(\mathtt{out}\, (\mathtt{xor}\, x'_1\, x'_2)) \mathtt{ \ be \ } z'_1, z''_2  \mathtt{ \ in\  } \langle  \mathtt{xor}\, z'_1\, y', \mathtt{or}\, (\mathtt{and}\, x''_1\, x''_2)(\mathtt{and}\, z''_1\, y'') \rangle \  ))
\enspace .
\qed
\end{align*}
\end{exmp}
\begin{exmp}[$3$-bits majority]  The level-by-level translation of the boolean circuit  $C$ in Figure~\ref{figure:maj3} is the following:
\label{exmp:maj3}
\allowdisplaybreaks
\begin{align*}
\mathtt{level}_C^{-1}&\triangleq \langle m, g \rangle           
\\
\mathtt{level}_C^{0}&\triangleq  (\lambda m. \lambda g. \mathtt{level}_C^{-1})(\mathtt{and}\, z_1\, z_2)(\mathtt{and}\, y''_2\, x''_3)\\
\mathtt{level}_C^{1}&\triangleq  (\lambda z_1. \lambda z_2. \mathtt{level}_C^0)(\mathtt{or}\, y'_1\, y'_3)(\mathtt{and}\, y''_1\, y''_3)\\
\mathtt{level}_C^{2}&\triangleq (\lambda y_1. \lambda y_3. \mathtt{let}\, (\mathtt{out}\, y_1)\mathtt{\ be \   } y'_1, y''_1 \mathtt{ \ in \ }\\
& 
\phantom{\ \triangleq \quad}
\mathtt{let}\, (\mathtt{out} \, y_3)\mathtt{ \ be \ }y'_3, y''_3 \mathtt{ \ in \ }\mathtt{level}_C^1          )(\mathtt{or}\, x'_1\, x'_2 )(\mathtt{or}\, y'_2\, x'_3)\\
\mathtt{level}_C^3&\triangleq (\lambda y_2.\lambda x_3. \mathtt{let}\, (\mathtt{out}\, y_2)\mathtt{\ be \ }y'_2, y''_2 \mathtt{ \ in \  }\\
&  
\phantom{\ \triangleq\quad}
\mathtt{let}(\mathtt{out} \, x_3)\mathtt{\ be \ }x'_3, x''_3 \mathtt{ \ in \ }\mathtt{level}_C^{2}                  )(\mathtt{and}\, x''_1\, x''_2 )\, x_3\\
\mathtt{level}_C^4&\triangleq (\lambda x_1. \lambda x_2.  \mathtt{let}\, (\mathtt{out}\, x_1)\mathtt{\ be \ }x'_1, x''_1 \mathtt{ \ in \  }\\
& 
\phantom{\ \triangleq \quad}
\mathtt{let}\, (\mathtt{out}\, x_2)\mathtt{\ be \ }x'_2, x''_2 \mathtt{ \ in \  }\mathtt{level}_C^{3}      )\, x_1\, x_2
\end{align*}
where we set $\lambda(C)\triangleq \lambda x.  \mathtt{let \ }x \mathtt{\ be \ }x_1, x_2, x_3 \mathtt{ \ in \ }\mathtt{level}_C^4$  which reduces to:
\begin{align*}
&\lambda x.  \mathtt{let \ }x \mathtt{\ be \ }x_1, x_2, x_3 \mathtt{ \ in \ } \mathtt{let}\, (\mathtt{out}\, x_1)\mathtt{\ be \ }x'_1,x''_1 \mathtt{ \ in \ }\\
&\mathtt{let}\,(\mathtt{out} \, x_2 )\mathtt{\ be \ }x'_2, x''_2 \mathtt{\ in\ } ( \mathtt{let}\,(\mathtt{out}\, x_3)\mathtt{\ be \ }x'_3, x''_3 \mathtt{ \ in \ }\\
&\mathtt{let}\,(\mathtt{out}\,(\mathtt{or}\, x'_1\, x'_2))\mathtt{\ be \ }y'_1, y''_1 \mathtt{ \ in \ } (
 \mathtt{let}\,(\mathtt{out}\,(\mathtt{and}\, x''_1\, x''_2)) \mathtt{\ be \ }y'_2, y''_2 \mathtt{ \ in \ }\\
& \mathtt{let }\,(\mathtt{out}\,(\mathtt{or}\, y'_2\, x'_3))\mathtt{\ be \ }y'_3, y''_3 \mathtt{ \ in \ }
\langle \mathtt{ and }\,(\mathtt{or}\, y'_1\, y'_3)(\mathtt{and}\, \, y''_1\, y''_3) , \mathtt{and}\,y''_2\, x''_3 \rangle \ ))
\enspace . \qed
\end{align*}
\end{exmp}

\noindent
The size of the term coding an internal node depends on its fan-in.
Likewise, the size of the term coding a fan-out node depends on the number of conclusions. The size of the circuit bounds both values.  Moreover, by Theorem~\ref{thm:Subject Reduction for IMLL2shpos}, reducing a typable term yields a typable term. These observations imply:
\begin{prop}[Simulation of circuit evaluation]\label{prop: simulation for boolean circuits}
If $ C$ is an unbounded fan-in boolean circuit over $\mathcal{B}_1$ with $n$ inputs and $m$ outputs then $\lambda(C)$  is such that:
\begin{enumerate}
\item  its size  is $ O(|C|) $, 
\item it has type $(\shpos \mathbf{B}\otimes \overset{n}{\ldots}\otimes \shpos \mathbf{B})\multimap ( \mathbf{B}\otimes \overset{m}{\ldots}\otimes  \mathbf{B})$ in $\mathsf{LEM}$, and
\item for all $(i_1, \ldots, i_n) \in \lbrace 0, 1 \rbrace^n$, the evaluation of  $C$  on input $(i_1, \ldots, i_n)$ outputs the tuple $(i'_1, \ldots, i'_{m})\in \lbrace 0, 1 \rbrace^m$ iff  $\lambda(C) \, \langle \mathtt{i}_1, \dots, \mathtt{i}_n \rangle  \rightarrow^*\langle \mathtt{i}'_1, \ldots, \mathtt{i}'_m\rangle$.
\end{enumerate}
\end{prop}
\noindent
It should not be surprising that the translation cannot preserve the depth of a given circuit,  since $ \mathsf{LEM}$ has only \emph{binary} logical operators. That is why we use nested instances 
\enquote{\texttt{let}}
(Definition~\eqref{eqn:datatype unity and product}) to access single elements of
$A_1\otimes\ldots\otimes A_n $.  We could preserve the depth by extending \textsf{LEM} with unbounded tensor products as done, for example, in~\cite{DBLP:conf/lics/Terui04} for the multiplicative fragment of linear logic $\mathsf{MLL}$.

\subsection{Numerals in $ \mathsf{LEM} $}
\label{section:Church numerals}
We introduce  a class $ \mathcal{N} $ of terms in $\mathsf{LEM}$,   called  \emph{numerals}, that represent natural numbers.
We give a successor \texttt{S} and an addition \texttt{A} on numerals, both typable in $\mathsf{LEM}$,  and we show  that they  
behave as expected using  Subject reduction. Moreover, the numerals can operate as
iterators on a class of terms in $\mathsf{LEM}$ that form a group with respect to the application.

\begin{defn}[Terms and types for $ \mathcal{N} $]
\label{defn:Terms and types for mathcalN}
Let us recall that $\mathbf{1}$ is $\forall \alpha. \alpha \multimap \alpha$ and $I$ is $\lambda x.x$ (Section~\ref{sec: background}).
The numerals of $ \mathcal{N} $ have form:
\begin{align}
\nonumber
\overline{0} &\triangleq 
 \lambda fx. \texttt{discard}_{\mathbf{1}}\, f\, \texttt{in}\, x: \mathbb{N} 
\\
\nonumber
\overline{1} &\triangleq
 \lambda fx. fx: \mathbb{N}
\\
\overline{n+2} &\triangleq
 \lambda fx. \texttt{copy}_{\mathbf{1}}^{I}\, f\, \texttt{as}\,
             f_1\ldots f_n\,\texttt{in}\, f_1(\ldots(f_n\, x)\ldots): \mathbb{N}
\enspace ,
\label{eqn: infinito}
\end{align}
\noindent
where, for any $M$,
$ \texttt{copy}_{\mathbf{1}}^{I}\, f_0\, \texttt{as}\, f_1\ldots f_n\,\texttt{in}\ M $ in~\eqref{eqn: infinito} stands for:
\begin{align*}
\texttt{copy}_{\mathbf{1}}^{I}\, f_0\, \texttt{as}\, f_1, f'_2 \,\texttt{in}\,
 (\texttt{copy}_{\mathbf{1}}^{I}\, f'_2\, \texttt{as}\, f_2, f'_3 \,\texttt{in}\, \ldots
 (\texttt{copy}_{\mathbf{1}}^{I}\, f'_{n-1}\, \texttt{as}\, f_{n-1}, f_n \,\texttt{in}\, M)\ldots)
 \enspace ,
\end{align*}
\noindent
and $ \mathbb{N} \triangleq \mathbf{N}[\mathbf{1}/\alpha] $ with
$ \mathbf{N} \triangleq
(\shpos \alpha) \multimap \alpha $.
In order to identify terms that represent the same natural number, 
we take numerals up to the following equivalences:
\begin{align*}
&\texttt{copy}_{\boldmath{1}}^{I}\, f\, \texttt{as}\, f_1, f_2 \,\texttt{in}\, 
(\texttt{copy}_{\boldmath{1}}^{I}\, f_2\, \texttt{as}\, f_3, f_4 \,\texttt{in}\, M)
\\
&
\ \,
\qquad
\qquad
\qquad
\qquad
\qquad
= \texttt{copy}_{\boldmath{1}}^{I}\, f\, \texttt{as}\, f_2, f_4 \,\texttt{in}\, (
 \texttt{copy}_{\boldmath{1}}^{I}\, f_2\, \texttt{as}\, f_1, f_3 \,\texttt{in}\, M)
\\
&f (\texttt{copy}_{\boldmath{1}}^{I}\, f'\, \texttt{as}\, g, h \ \texttt{in}\, M)
=
\texttt{copy}_{\boldmath{1}}^{I}\, f'\, \texttt{as}\, g, h \,\texttt{in}\, f\,M
\enspace .
\qquad 
\qquad 
\qquad 
\qquad 
\qed
\end{align*}
\end{defn}
The elements of $ \mathcal{N} $ are the analogue of the Church numerals. Let us compare $\mathbb{N}$ to the type $\mathbf{int}$ of the Church numerals in Linear Logic: 
\begin{equation*}
\begin{split}
\mathbf{int}&\triangleq \forall \alpha .(\oc (\alpha \multimap \alpha)\multimap (\alpha \multimap \alpha))\\
\mathbb{N} &\triangleq (\shpos \forall \alpha .(\alpha \multimap \alpha))\multimap \forall \alpha . (\alpha \multimap \alpha) \enspace .
\end{split}
\end{equation*}
In the former the universal quantification is in positive position, while in the latter it  occurs on both sides of the main implication. This is because we can apply the modality $\shpos$  only to  ground types, which are closed. Also, observe that the lack of an external quantifier in $\mathbb{N}$ limits  the use of numerals as iterators.

The analogy with the Church numerals can be pushed further by defining a successor $\texttt{S}$ and an addition $\texttt{A}$:
\begin{align*}
\texttt{S} & \triangleq
\lambda nfx.
\texttt{copy}_{\boldmath{1}}^{I}\,f\,\texttt{as}\, f_1, f_2 \,\texttt{in}\, f_1(n f_2 x): \mathbb{N}\multimap\mathbb{N}
\\
\texttt{A} & \triangleq
\lambda mnfx.
\texttt{copy}_{\boldmath{1}}^{I}\,f\,\texttt{as}\, f_1, f_2 \,\texttt{in}\, 
m f_1(n f_2 x): \mathbb{N}\multimap \mathbb{N}\multimap\mathbb{N}
\enspace .
\end{align*}
\noindent
Sticking to the computational behaviour of the terms
(Figures~\ref{fig: term reduction rules} 
and~\ref{fig: term commuting conversions}), not of the underlying derivations, 
the Subject reduction 
(Theorem~\ref{thm:Subject Reduction for IMLL2shpos}) implies:

\begin{prop} \label{prop: succ add}For all $n , m\geq 0$, $ \mathtt{S}\,\overline{n} \rightarrow^* \overline{n+1} $ and $ \mathtt{A}\,\overline{m}\,\overline{n}\rightarrow^* \overline{m+n} $.
\end{prop}
\noindent
For example, the following reduction is legal:
\begin{align*}
\texttt{S} \,\overline{2} & \triangleq
(\lambda nfx.
\texttt{copy}_{\boldmath{1}}^{I}\,f\,\texttt{as}\, f_1, f_2 \,\texttt{in}\, f_1(n f_2 x))
( \lambda gy. \texttt{copy}_{\mathbf{1}}^{I}\,g\, \texttt{as}\,
             g_1,  g_2\,\texttt{in}\, g_1(g_2\, y))
\\
& \rightarrow
\lambda fx.
\texttt{copy}_{\boldmath{1}}^{I}\,f\,\texttt{as}\, f_1, f_2 \,\texttt{in}\, f_1(( \lambda gy. \texttt{copy}_{\mathbf{1}}^{I}\,g\, \texttt{as}\,
             g_1,  g_2\,\texttt{in}\, g_1(g_2\, y)) f_2 x)
\\
& \rightarrow
\lambda fx.
\texttt{copy}_{\boldmath{1}}^{I}\,f\,\texttt{as}\, f_1, f_2 \,\texttt{in}\, f_1(( \lambda y. \texttt{copy}_{\mathbf{1}}^{I}\,f_2\, \texttt{as}\,
             g_1,  g_2\,\texttt{in}\, g_1(g_2\, y)) x)
\\
& \rightarrow
\lambda fx.
\texttt{copy}_{\boldmath{1}}^{I}\,f\,\texttt{as}\, f_1, f_2 \,\texttt{in}\, f_1( \texttt{copy}_{\mathbf{1}}^{I}\,f_2\, \texttt{as}\,
             g_1,  g_2\,\texttt{in}\, g_1(g_2\, x)) 
      \\
& =
\lambda fx.
\texttt{copy}_{\boldmath{1}}^{I}\,f\,\texttt{as}\, f_1, f_2 \,\texttt{in}\, ( \texttt{copy}_{\mathbf{1}}^{I}\,f_2\, \texttt{as}\,
             g_1,  g_2\,\texttt{in}\, f_1(g_1(g_2\, x)))        
              \triangleq \overline{3}
\end{align*}
\noindent
Observe that Proposition~\ref{prop: succ add} considers typable terms and term reductions  by exploiting Theorem~\ref{thm:Subject Reduction for IMLL2shpos}. A similar result cannot be restated for the related  derivations and the  lazy-cut elimination  (Definition~\ref{defn:Lazy cut-elimination strategy}).   For example, the here above term $\texttt{S} \,\overline{2}$ has  type $\mathbb{N}$, that is not lazy (Definition~\ref{defn: lazyness}), due to the presence of a universal quantification in negative position.  Indeed, the  lazy cut-elimination strategy of a derivation of $\texttt{S} \,\overline{2}$  runs  into  deadlocks before producing a cut-free derivation of $\overline{3}$. 


As far as we could see, the  \enquote{zero-test}, the predecessor  and the subtraction on numerals cannot have type in $\mathsf{LEM}$.  
The problem is the position of the universal quantifiers of $ \mathbb{N} $. Consider, for example, the following predecessor:
\begin{align*}
\texttt{P} & \triangleq
\lambda nsz.n\, S[s]\, B[z]
&&
(\textrm{Predecessor})
\\
S[M] & \triangleq
\lambda p.\texttt{let}\ p\ \texttt{be}\ l, r\ \texttt{in}\ 
\langle M, lr \rangle
&&
(\textrm{Step function})
\\
B[N] & \triangleq \langle I, N\rangle
&&
(\textrm{Base function})
\end{align*}
introduced by Roversi \cite{Roversi:1999-CSL}. Giving a type to \texttt{P} would require to substitute $ (\alpha\multimap\alpha)\otimes\alpha$ for $ \alpha $ in $ \mathbb{N} $, as suggested by the application of $ n:\mathbb{N} $ to $ S[s] $. The position of the universal quantifiers in $\mathbb{N}$ forbids it. Were such an instance legal, we could iterate functions, contradicting the cubic bound on the cut-elimination
(Theorem~\ref{thm: cut elimination for downarrow IMLL2}.)

Further, we can generalize $\mathbb{N}$ to
$\mathbf{N}[\overline{(A \multimap A)} /\gamma]$,
where $\overline{(A \multimap A)}$ is the closure of a 
quantifier-free type $A \multimap A$,
and find that \emph{Hereditarily finite permutations} (\textsf{HFP})
by Dezani
\cite{DBLP:journals/tcs/Dezani-Ciancaglini76} 
inhabit $\overline{A \multimap A}$. 
An  \textsf{HFP}  is a $\lambda$-term of the form $ P \triangleq \lambda z x_1 \ldots x_n. z(P_1 x_{\rho(1)})\ldots (P_n x_{\rho(n)}) $,
for some $ n\geq 0 $,  where $\rho \in S_n$ (the symmetric group of $\lbrace 1, \ldots, n\rbrace$)  and $ P_1, \ldots, P_n$ are \textsf{HFP}. The class $\mathcal{H}_{\text{lin}}$ of linear $\lambda$-terms which are \textsf{HFP}  (considered modulo $\beta \eta$-equivalence) forms a group:
\begin{enumerate}
\item The binary operation is $\lambda fgx.f(g\,x)$;
\item The identity is $I$;
\item If $P=\lambda z x_1 \ldots x_n. z(P_1 x_{\rho(1)})\ldots (P_n x_{\rho(n)})$  is in  $\mathcal{H}_{\text{lin}}$, the inductively defined inverse is:
 \[P^{-1} \triangleq \lambda z' x_1 \ldots x_n. z'(P^{-1}_{\rho^{-1}(1)} x_{\rho^{-1}(1)})\ldots (P^{-1}_{\rho^{-1}(n)} x_{\rho^{-1}(n)})\] where, for all $1 \leq i \leq n$,  $\rho^{-1}_i$ is the inverse permutation of $\rho_i$  and $P^{-1}_{i}$ is the inverse of $P_{i}$.
\end{enumerate}
For example, let $ P =\lambda wabc.w(\lambda xy. ayx)(\lambda xy. bxy)c $ which belongs to \textsf{HFP} since
$ \lambda xy. ayx =_{\beta } (\lambda zxy. zyx) a $,
$ \lambda xy. bxy =_{\beta } (\lambda zxy. zxy) b $,
$ c =_{\beta } Ic $, where $ (\lambda zxy. zyx)$, $ (\lambda zxy. zxy) $ and $ I $ are in \textsf{HFP}. Then, $P$ has type
$\forall \alpha_x \alpha_y\alpha_a\alpha_b\alpha_c\alpha. 
A \multimap A$, 
which is a ground type
where $A$ is 
$(\alpha_x\multimap\alpha_y\multimap\alpha_a )
 \multimap 
 (\alpha_x\multimap\alpha_y\multimap\alpha_b ) 
 \multimap \alpha_c 
 \multimap \alpha$.
These observations show an unexpected link between $\mathsf{LEM}$ and reversible computation (see Perumalla \cite{perumalla2013chc} for a thorough introduction) that can well be expressed in terms of monoidal structures where permutations play a central role \cite{DBLP:conf/types/PaoliniPR15, paolini16ictcs, paolini2018ngc}.

\section{Conclusions}
\label{section:Conclusions}\noindent
We introduce $\mathsf{LEM}$. It is a type-assignment for the 
linear $ \lambda $-calculus extended with new constructs that can 
duplicate or erase values, i.e.~closed and normal linear 
$ \lambda $-terms. $\mathsf{LEM}$ enjoys a mildly weakened 
cut-elimination, and the Subject reduction. The internalization of the 
mechanism of linear weakening  and contraction by means of modal rules 
allows to exponentially compress derivations of $\mathsf{IMLL}_2$.  
On one side, this enables to represent boolean circuits more compactly, as compared to previous ones, based on the multiplicative fragment of Linear Logic. On the other, $\mathsf{LEM}$ can represent Church-like encoding of the natural numbers, with successor and the addition for them.

We conclude by briefly discussing possible future works.

In Section~\ref{sec: duplication and erasure in a linear setting}, we 
conjecture that a version of the general separation property holds in 
the linear $\lambda$-calculus (Conjecture~\ref{conj: linear separation}.)
Were it true, we could show that a
 duplicator exists for all the finite sets of closed terms in 
$\beta\eta$-normal form, and connect
linear duplication with the standard notion of separation.

In Section~\ref{sec: the system shpos, with, oplus IMLL2 cut elimination and bound},  we design $\mathsf{LEM}$ to express linear weakening and 
contraction in the same spirit as of the exponential rules of  
$\mathsf{LL}$. We are working to push the analogy further,
by formulating a type-assignment that extends
$\mathsf{IMLL}_2$ with \enquote{linear additives}. A candidate rule is:
\begin{prooftree}
\AxiomC{$x_1:A \vdash M_1: A_1 $}
\AxiomC{$x_2: A \vdash M_2: A_2$}
\AxiomC{$ \vdash V:A$}
\RightLabel{$\with$R}
\TrinaryInfC{$x: A \vdash \mathtt{copy}^V_A \, x \mathtt{\ as \ }x_1,x_2 \mathtt{\ in \ } \langle M_1, M_2\rangle : A_1 \with A_2$}
\end{prooftree}
where $V$ is a value and $A, A_1, A_2$ are closed and without negative occurrences of $\forall$. The intuition behind this rule is the one 
we discuss in Section~\ref{sec: the system shpos, with, oplus IMLL2 cut elimination and bound} for the contraction rule of $\mathsf{LEM}$. 
We also claim that the new system  would keep the normalization cost  linear, unlike standard additives (see~\cite{mairson2003computational}). 

Section~\ref{sec: encoding boolean circuits}  presents an encoding of the  boolean circuits not preserving their depth. Moving to unbounded fan-in proof nets for $\mathsf{LEM}$ would improve the correspondence,  where the rules $ p, w, c $ and $ d $ would be expressed  by nodes and boxes. 
Operations on them would compactly perform duplication and get rid of garbage, possibly improving~\cite{DBLP:conf/lics/Terui04,DBLP:conf/lfcs/MogbilR07,DBLP:journals/corr/abs-1201-1120}.  
A reasonable question would then be whether the use of alternative and weaker exponential rules in $\mathsf{LEM}$ could be the right approach to capture circuit complexity classes like $\textsf{NC}, \textsf{AC}$, and $\textsf{P}/_{\operatorname{poly}}$ in analogy with the implicit characterizations of the Polynomial and Elementary-time computational complexity classes by means of Light Logics \cite{lafont2004soft,danos2003linear}.

Section~\ref{section:Church numerals} contributes to the problem of defining numeral systems in linear settings. 
In~\cite{Mackie2018}, Mackie has recently introduced linear variants of numeral systems.
He shows that successor, addition, predecessor, and subtraction have representatives in the linear $\lambda$-calculus. We could not find how giving type in $\mathsf{LEM}$ to  some of  the terms of Mackie's numeral systems.
However, by merging Mackie's encoding and Scott numerals \cite{CF:58}, numeral systems seem to exist which  $ \mathsf{LEM} $ can give a type to. The cost would be to extend $\mathsf{LEM}$ with recursive types, following 
Roversi\&Vercelli \cite{Roversi+Vercelli:2010-DICE10}.

\bibliographystyle{plain}
\section*{\refname}
\bibliography{biblio}

\begin{thebibliography}{10}

\bibitem{DBLP:conf/csl/AlvesFFM06}
Sandra Alves, Maribel Fern{\'{a}}ndez, M{\'{a}}rio Florido, and Ian Mackie.
\newblock The {P}ower of {L}inear {F}unctions.
\newblock In Zolt{\'{a}}n {\'{E}}sik, editor, {\em Computer Science Logic, 20th
  International Workshop, {CSL} 2006, 15th Annual Conference of the EACSL,
  Szeged, Hungary, September 25-29, 2006, Proceedings}, volume 4207 of {\em
  Lecture Notes in Computer Science}, pages 119--134. Springer, 2006.

\bibitem{DBLP:journals/corr/abs-1201-1120}
Cl{\'{e}}ment Aubert.
\newblock Sublogarithmic uniform boolean proof nets.
\newblock In Jean{-}Yves Marion, editor, {\em Proceedings Second Workshop on
  Developments in Implicit Computational Complexity, Saarbr{\"{u}}cken, {DICE}
  2011, Germany, April 2nd and 3rd, 2011.}, volume~75 of {\em {EPTCS}}, pages
  15--27, 2011.

\bibitem{barendregt1984lambda}
Hendrik~Pieter Barendregt.
\newblock The lambda calculus: Its syntax and semantics. 1984.
\newblock {\em Studies in Logic and the Foundations of Mathematics}, 1984.

\bibitem{coppo1978semi}
Mario Coppo, Mariangiola Dezani-Ciancaglini, and S~Ronchi Della~Rocca.
\newblock (semi)-separability of finite sets of terms in scott's
  d$\infty$-models of the $\lambda$-calculus.
\newblock In {\em International Colloquium on Automata, Languages, and
  Programming}, pages 142--164. Springer, 1978.

\bibitem{CF:58}
H.~B. Curry and R.~Feys.
\newblock {\em Combinatory Logic, Volume I}.
\newblock North-Holland, 1958.
\newblock Second printing 1968.

\bibitem{danos2003linear}
Vincent Danos and Jean-Baptiste Joinet.
\newblock Linear logic and elementary time.
\newblock {\em Information and Computation}, 183(1):123--137, 2003.

\bibitem{DBLP:journals/tcs/Dezani-Ciancaglini76}
Mariangiola Dezani{-}Ciancaglini.
\newblock Characterization of normal forms possessing inverse in the
  \emph{lambda-beta-eta}-calculus.
\newblock {\em Theor. Comput. Sci.}, 2(3):323--337, 1976.

\bibitem{gaboardi2009light}
Marco Gaboardi and Simona Ronchi~Della Rocca.
\newblock From light logics to type assignments: a case study.
\newblock {\em Logic Journal of the {IGPL}}, 17(5):499--530, 2009.

\bibitem{Girard:TCS87}
Jean{-}Yves Girard.
\newblock Linear {L}ogic.
\newblock {\em Theor. Comput. Sci.}, 50:1--102, 1987.

\bibitem{girard1987linear}
Jean-Yves Girard.
\newblock Linear logic.
\newblock {\em Theoretical computer science}, 50(1):1--101, 1987.

\bibitem{hindley1989bck}
J~Roger Hindley.
\newblock Bck-combinators and linear $\lambda$-terms have types.
\newblock {\em Theoretical Computer Science}, 64(1):97--105, 1989.

\bibitem{klop1980combinatory}
Jan~Willem Klop.
\newblock Combinatory reduction systems.
\newblock 1980.

\bibitem{lafont2004soft}
Yves Lafont.
\newblock Soft linear logic and polynomial time.
\newblock {\em Theoretical Computer Science}, 318(1):163--180, 2004.

\bibitem{Mackie2018}
Ian Mackie.
\newblock Linear {N}umeral {S}ystems.
\newblock {\em Journal of Automated Reasoning}, Feb 2018.

\bibitem{mairsonlinear}
Harry~G. Mairson.
\newblock Linear {L}ambda {C}alculus and {PTIME}-completeness.
\newblock {\em J. Funct. Program.}, 14(6):623--633, November 2004.

\bibitem{mairson2003computational}
Harry~G. Mairson and Kazushige Terui.
\newblock {O}n the {C}omputational {C}omplexity of {C}ut-{E}limination in
  {L}inear {L}ogic.
\newblock In Carlo Blundo and Cosimo Laneve, editors, {\em Theoretical Computer
  Science}, pages 23--36, Berlin, Heidelberg, 2003. Springer Berlin Heidelberg.

\bibitem{MATSUOKA200737}
Satoshi Matsuoka.
\newblock Weak typed {B}{o}hm {T}heorem on {IMLL}.
\newblock {\em Annals of Pure and Applied Logic}, 145(1):37--90, 2007.

\bibitem{DBLP:conf/lfcs/MogbilR07}
Virgile Mogbil and Vincent Rahli.
\newblock Uniform circuits, {\&} boolean proof nets.
\newblock In Sergei~N. Art{\"{e}}mov and Anil Nerode, editors, {\em Logical
  Foundations of Computer Science, International Symposium, {LFCS} 2007, New
  York, NY, USA, June 4-7, 2007, Proceedings}, volume 4514 of {\em Lecture
  Notes in Computer Science}, pages 401--421. Springer, 2007.

\bibitem{DBLP:conf/types/PaoliniPR15}
Luca Paolini, Mauro Piccolo, and Luca Roversi.
\newblock A certified study of a reversible programming language.
\newblock In {\em {TYPES}}, volume~69 of {\em LIPIcs}, pages 7:1--7:21. Schloss
  Dagstuhl - Leibniz-Zentrum fuer Informatik, 2015.

\bibitem{paolini16ictcs}
Luca Paolini, Mauro Piccolo, and Luca Roversi.
\newblock A class of reversible primitive recursive functions.
\newblock {\em Electronic Notes in Theoretical Computer Science}, 322:227--242,
  2016.
\newblock Elsevier, Netherlands.

\bibitem{paolini2018ngc}
Luca Paolini, Mauro Piccolo, and Luca Roversi.
\newblock On a class of reversible primitive recursive functions and its
  turing-complete extensions.
\newblock {\em New Generation Computing}, 36(3):233--256, Jul 2018.

\bibitem{perumalla2013chc}
Kalyan~S. Perumalla.
\newblock {\em Introduction to Reversible Computing}.
\newblock Chapman \& Hall/CRC Computational Science. Taylor \& Francis, 2013.

\bibitem{Roversi:1999-CSL}
Luca Roversi.
\newblock {A} {P}-{T}ime {C}ompleteness {P}roof for {L}ight {L}ogics.
\newblock In {\em Ninth Annual Conference of the EACSL (CSL'99)}, volume 1683
  of {\em Lecture Notes in Computer Science}, pages 469 -- 483, Madrid (Spain),
  September 1999. Springer-Verlag.

\bibitem{Roversi+Vercelli:2010-DICE10}
Luca Roversi and Luca Vercelli.
\newblock {Safe Recursion on Notation into a Light Logic by Levels}.
\newblock In {\em {Proceedings of the Workshop on Developments in Implicit
  Computational complexity (DICE 2010)}}, volume~23 of {\em {Electronic
  Proceedings in Theoretical Computer Science}}, pages 63 -- 77. On-line, March
  2010.

\bibitem{schubert2001complexity}
Aleksy Schubert.
\newblock The complexity of $\beta$-reduction in low orders.
\newblock In {\em International Conference on Typed Lambda Calculi and
  Applications}, pages 400--414. Springer, 2001.

\bibitem{DBLP:conf/lics/Terui04}
Kazushige Terui.
\newblock Proof nets and boolean circuits.
\newblock In {\em 19th {IEEE} Symposium on Logic in Computer Science {(LICS}
  2004), 14-17 July 2004, Turku, Finland, Proceedings}, pages 182--191. {IEEE}
  Computer Society, 2004.

\bibitem{troelstra2000basic}
Anne~Sjerp Troelstra and Helmut Schwichtenberg.
\newblock {\em Basic proof theory}.
\newblock Number~43. Cambridge University Press, 2000.

\bibitem{Vollmer:1999:ICC:520668}
Heribert Vollmer.
\newblock {\em Introduction to Circuit Complexity: A Uniform Approach}.
\newblock Springer-Verlag, Berlin, Heidelberg, 1999.

\end{thebibliography}
\appendix
\section{The proof of Theorem~\ref{thm: pi1 are duplicable}} 
\label{sec: the d-soundness theorem DICE}
In this section we give a detailed proof of Theorem~\ref{thm: pi1 are duplicable} for $\mathsf{IMLL}_2$, which states that if $A$ is an inhabited ground type, i.e~an inhabited  closed $\Pi_1$-type, then $A$ is also a duplicable type. By Definition~\ref{defn: duplicable and erasable types}, this amounts to show that a linear $\lambda$-term $\mathtt{D}_{A}: A\multimap A \otimes A$ exists such that  $\mathtt{D}_{A}\, V \rightarrow^*_{\beta \eta} \langle V, V \rangle$ holds for every value  $V$ of $A$. We shall construct $ \mathtt{D}_{A} $ as the composition of three linear $\lambda$-terms, as diagrammatically displayed in Figure~\ref{fig:DA diagram}. For each such component we dedicate a specific subsection. For the sake of presentation, in this section we  focus on  terms of $\mathsf{IMLL}_2$ rather than on derivations, so that when we say that a term $M$ has type $A$ with context $\Gamma$, we clearly mean that a derivation $\mathcal{D}$ exists such that $\mathcal{D}\triangleleft \Gamma \vdash M: A$. Moreover, as assumed in Section~\ref{sec: background}, terms are considered modulo $\alpha$-equivalence. 
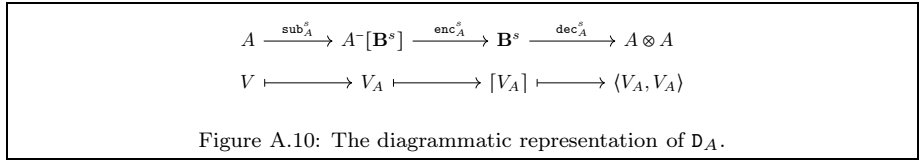
\begin{figure}
    \begin{tikzpicture}[baseline= (a).base]
    \node[scale=.725] (a) at (0,0){
        \begin{tikzcd}[column sep=scriptsize, row sep=tiny]
        A \ar[rr, "\mathtt{sub}^s_{A}"] 
        && 
        A^-[\mathbf{B}^s] \ar[rr, "\mathtt{enc}^s _{A}"] 
        && 
        \mathbf{B}^s 
        \ar[rr, "\mathtt{dec}^s _{A} "] 
        &&
        A \otimes A
        \\
        V \ar[rr, mapsto]  
        &&  
        V_{A} \ar[rr, mapsto] 
        && 
        \lceil V_{A} \rceil \ar[rr, mapsto]
        &&
        \langle V_{A}, V_{A}  \rangle
        \end{tikzcd}
    };
    \end{tikzpicture}
    \caption{The diagrammatic representation of $ \mathtt{D}_{A}$.}
    \label{fig:DA diagram} 
\end{figure}
\subsection{The linear $ \lambda $-term $ \mathtt{sub}^s_{A} $}
Roughly, the $\lambda$-term $\mathtt{sub}^s_{A}$, when applied to a value  $V$ of ground type $A$, produces its $\eta$-long normal form $V_{A}$ whose type is obtained from $A$ as follows: we strip away every occurrence of $\forall$ and we substitute each type variable  with the $s$-ary tensor of boolean datatypes $\mathbf{B}^s=\mathbf{B}\otimes \overset{s}{\ldots}\otimes \mathbf{B}$, for some $s>0$.\\
Before introducing the $\lambda$-term $\mathtt{sub}^s_{A}$, we need the definition of $\eta$-long normal form:
\begin{defn}[$\eta$-long normal forms]\label{defn: eta long nf}
Let $\mathcal{D}\triangleleft \Gamma \vdash M: B$ be cut-free.
We define the $\eta$-\textit{expansion} of $\mathcal{D}$, denoted  $\mathcal{D}^\Gamma_B$,  as the derivation obtained from $\mathcal{D}$ by substituting every occurrence of:
\begin{prooftree}
\AxiomC{}
\RightLabel{$ax$}
\UnaryInfC{$x: A \vdash x: A$}
\end{prooftree}
 with a derivation of $ x: A \vdash M':A $, for some $ M' $, whose axioms have form $y: \alpha \vdash y: \alpha$. The  $\eta$-expansion is unique and transforms the $\lambda$-term $M$ to   its \emph{$\eta$-long normal form}, denoted by $M^{\Gamma}_B$ and such that $M^\Gamma_B\rightarrow^*_\eta M$.
If the context $\Gamma$ of an $ \eta $-expanded $ \mathcal{D} $ is 
$x_1: A_1, \ldots, x_n:A$ we may write 
$\mathcal{D}^{A_1, \ldots, A_n}_B$ and $M^{A_1, \ldots, A_n}_B$.
If $\Gamma$ is empty, we feel free to write $\mathcal{D}_B$ and $M_B$.
\end{defn}
\begin{lem}\label{lem: canonical definition for eta long normal forms} Let $\mathcal{D}\triangleleft \Gamma \vdash M: A$ be a cut-free derivation in $\mathsf{IMLL}_2$, and let  $M^\Gamma_A$ denote the $\eta$-long normal form obtained by $\eta$-expanding $\mathcal{D}$. Then:
\begin{enumerate}[1]
\item  \label{eqn: 1 canonical definition for eta long normal forms} If $M=x$, $A= \alpha$, and  $\Gamma= x: \alpha$ then  $x^\alpha_\alpha=x$.
\item \label{eqn: 2 canonical definition for eta long normal forms} If $M=x$, $A= \forall \alpha. B$, and  $\Gamma= x: \forall \alpha. B$ then  $x^{\forall \alpha. B}_{\forall \alpha. B}= x^{B}_B$.
\item  \label{eqn: 3 canonical definition for eta long normal forms} If $M=x$, $A= B \multimap C$, and  $\Gamma= x: B \multimap C$ then  $x^{B \multimap C}_{B \multimap C}=\lambda y.(xy^{B}_B)^C_C$.
\item \label{eqn: 4 canonical definition for eta long normal forms} If $A= \forall \alpha. B$ then $M^{\Gamma}_{\forall \alpha. B}= M^\Gamma_{B\langle \gamma/\alpha \rangle}$, for some $\gamma$. 
\item \label{eqn: 5 canonical definition for eta long normal forms} If $M= \lambda x. N$ and $A= B \multimap C$ then $(\lambda x. N)^{\Gamma}_{B \multimap C}= \lambda x. N^{\Gamma, x: B}_C$.
\item \label{eqn: 7 canonical definition for eta long normal forms} If $M=P[yN/x]$ and $\Gamma=\Delta, \Sigma, y: B \multimap C$, where $P$ has type $A$ with context $\Delta, x:C$  and   $N$ has type $B$ with context $\Sigma$,  then $(P[yN/x])^{\Gamma}_A=P^{\Delta, x: C}_A[yN^\Sigma _B/x]$.
\item \label{eqn: 6 canonical definition for eta long normal forms} If $M=P[yN/x]$ and  $\Gamma= \Gamma', y: \forall \alpha. B$ then we have $(P[yN/x])^{\Gamma', y: \forall \alpha.B}_A=(P[yN/x])^{\Gamma', y: B\langle D/\alpha \rangle}_A$, for some type $D$. 
\end{enumerate}
\end{lem}
\begin{proof}
Just follow the definition of $\eta$-long normal form.
\end{proof}

\begin{defn} \label{defn: stripping forall} Let $A$ be a type in $\mathsf{IMLL}_2$. We define $A^-$ by induction on the complexity of the type:
\begin{align*}
\alpha^- &\triangleq \alpha \\
(A \multimap B)^- &\triangleq  A^- \multimap B^-\\
(\forall \alpha. A)^-&\triangleq  A^-\langle \gamma/\alpha \rangle \enspace, 
\end{align*}
where $\gamma$ is taken from the head of an infinite list of fresh type variables. 
The notation $A[B]$ denotes the type obtained by replacing $B$ for every free type variable of $A$. Moreover, if $\Gamma= x_1:A_1, \ldots, x_n:A_n$, then $\Gamma^-$ stands for $x_1:A_1^-, \ldots, x_n:A^-_n$, and $\Gamma[B]$  stands for $x_1:A_1[B], \ldots, x_n:A_n[B]$.\qed
\end{defn}
\begin{defn}[The linear $\lambda$-term $\mathtt{sub}^s_{A}$]\label{defn: sub} Let $s>0$. We define the linear $\lambda$-terms $\mathtt{sub}^s_{A}: A[\mathbf{B}^s]\multimap A^-[\mathbf{B}^s]$, where $A $ is a $ \Pi_1$-type, and $\overline{\mathtt{sub}}^s_{A}: A^-[\mathbf{B}^s] \multimap A[\mathbf{B}^s]$, where $A $ is a  $\Sigma_1$-type, by simultaneous induction on the size of $A$:
\begin{align*}
	&\mathtt{sub}^s_{\alpha}\triangleq  \lambda x.x 
	&&\overline{\mathtt{sub}}^s_{\alpha} \triangleq  \lambda x.x \\
	&\mathtt{sub}^s_{\forall \alpha. B}  \triangleq  \mathtt{sub}^s_{B} 
&&  \\
	&\mathtt{sub}^s_{B \multimap C}  \triangleq  \lambda x.\lambda y.  \mathtt{sub}^s_{C}(x\, (\overline{\mathtt{sub}}^s_{B}\,y)) 
	&&  \overline{\mathtt{sub}}^s_{B \multimap C}
	      \triangleq  
	       \lambda x.\lambda 
	       y.\overline{\mathtt{sub}}^s_C(x\,(\mathtt{sub}^s_{B}\,y)).	
           \qed
\end{align*}
\end{defn}
The following will be used to compact the proof of some of the coming lemmas.
\begin{defn}\label{defn: sub substitution}
Let $s>0$. Let $A$ be a $\Pi_1$-type. Let $\Gamma=x_1:A_1, \ldots, x_n:A_n$ be a context of $\Sigma_1$-types. Let $M$ be an inhabitant of 
$A[\mathbf{B}^s]$ with context $\Gamma[\mathbf{B}^s]$.
Then $M[\Gamma]$ denotes the substitution:
\begin{align*}
&M[\overline{\mathtt{sub}}^s_{A_1}\,x'_1/x_1, \ldots,\overline{\mathtt{sub}}^s_{A_n}\,x'_n/x_n ]
\end{align*} 
for some $x'_1, \ldots, x'_n$.
\qed
\end{defn}

\begin{lem}\label{lem: sub variable to eta long normal form} Let $s>0$ and $z$ be of type $A[\mathbf{B}^s]$.
\begin{enumerate}[(1)]
\item \label{eqn: 1 sub variable to eta long normal form} If $A$ is a $\Pi_1$-type, then $\mathtt{sub}^s_{A}\, z \rightarrow_\beta^* z_{A}^A$.
\item \label{eqn: 2 sub variable to eta long normal form}  If $A$ is a  $\Sigma_1$-type, then $\overline{\mathtt{sub}}^s_A\, z \rightarrow_\beta^* z^A_A$.
\end{enumerate}
\end{lem}
\begin{proof}
We prove both points by simultaneous induction on $\vert A \vert$:
\begin{enumerate}
\item Case $A= \alpha$. Both the statements are straightforward since we have  $z^\alpha _\alpha=z$ by Lemma~\ref{lem: canonical definition for eta long normal forms}.\ref{eqn:  1 canonical definition for eta long normal forms}.
\item Case $A= \forall \alpha. B$. This case applies to point~\ref{eqn: 1 sub variable to eta long normal form} only. By induction hypothesis, for every variable $x$ of type $B[\mathbf{B}^s]$, $\mathtt{sub}^s_{B}\,x \rightarrow_\beta^* x^B_B$. The $\lambda$-term $\mathtt{sub}^s_B$ has type  $B[\mathbf{B}^s]\multimap B^-[\mathbf{B}^s]$, which is equal to $(B\langle \mathbf{B}^s/\alpha\rangle )[\mathbf{B}^s]\multimap (\forall \alpha. B)^-[\mathbf{B}^s]$. Hence, $\mathtt{sub}^s_B$ has also type $(\forall \alpha. B)[\mathbf{B}^s]\multimap (\forall \alpha. B)^-[\mathbf{B}^s]$. Moreover,  by Definition~\ref{defn: sub} we have $\mathtt{sub}^s_{B}= \mathtt{sub}^s_{\forall \alpha. B}$. Therefore, for every variable $z$ of type $(\forall \alpha. B)[\mathbf{B}^s]$ we have  $\mathtt{sub}^s_{\forall \alpha. B}\, z= \mathtt{sub}^s_B\, z \rightarrow_\beta^* z^{B}_B$. But $z^{B}_B=z^{\forall \alpha. B}_{\forall \alpha. B}$ by Lemma~\ref{lem: canonical definition for eta long normal forms}.\ref{eqn:  2 canonical definition for eta long normal forms}.
\item Case $A= B \multimap C$. We prove point~\ref{eqn: 1 sub variable to eta long normal form}   only (point~\ref{eqn: 2 sub variable to eta long normal form} is similar). Let $z$ be of type $(B\multimap C)[\mathbf{B}^s]=B[\mathbf{B}^s]\multimap C[\mathbf{B}^s]$. Then we have
\allowdisplaybreaks
\begin{align*}
\mathtt{sub}^s_{B \multimap C}\, z &=(\lambda x.  \lambda y.  \mathtt{sub}^s_{C}(x(\overline{\mathtt{sub}}^s_{B}\, y)))z &&\text{Definition}~\ref{defn: sub} \\
&\rightarrow_\beta \lambda y.  \mathtt{sub}^s_{C}(z(\overline{\mathtt{sub}}^s_{B}\, y))&&\\ 
&=\lambda y. ( \mathtt{sub}^s_{C}\, w)[z(\overline{\mathtt{sub}}^s_{B}\, y)/w] &&\\
& \rightarrow_\beta^* \lambda y. w^C_C [z (\overline{\mathtt{sub}}^s_{B}\, y)/w]&& \text{induction hyp., point}~\ref{eqn: 1 sub variable to eta long normal form}\\
& \rightarrow_\beta^* \lambda y. w^C_C [z y^B_B/w]&& \text{induction hyp., point}~\ref{eqn: 2 sub variable to eta long normal form}\\
&= \lambda y. (zy^B_B)^C_C &&\\
&= z^{B \multimap C}_{B\multimap C}&& \text{Lemma}~\ref{lem: canonical definition for eta long normal forms}.\ref{eqn: 3 canonical definition for eta long normal forms}.
\end{align*}
\end{enumerate}
\end{proof}
\begin{lem}\label{lem: sub preservation of eta long normal form} Let $s>0$. 
If $z:A[\mathbf{B}^s]$, where $A$ is a  $\Pi_1$-type, then $\mathtt{sub}^s_A\, z^A_A \rightarrow_\beta^* z^A_A $.
\end{lem}
\begin{proof}
 We prove it by induction on $\vert A \vert$:
\begin{enumerate}
\item Case $A= \alpha$. The statement is straightforward since we have $z^\alpha _\alpha=z$ by Lemma~\ref{lem: canonical definition for eta long normal forms}.\ref{eqn: 1 canonical definition for eta long normal forms}
\item Case $A= \forall \alpha. B$.  By Definition~\ref{defn: sub}, $\mathtt{sub}^s_{\forall \alpha. B}= \mathtt{sub}^s_{B}$ and we use the induction hypothesis.
\item Case $A= B \multimap C$. Then we have
\allowdisplaybreaks
\begin{align*}
 \mathtt{sub}^s_{B \multimap C}\, z^{B \multimap C}_{B\multimap C}&= (\lambda x.  \lambda y.  \mathtt{sub}^s_{C}(x(\overline{\mathtt{sub}}^s_{B}\, y)))z^{B \multimap C}_{B\multimap C}&& \text{Definition}~\ref{defn: sub}\\
 &= (\lambda x.  \lambda y.  \mathtt{sub}^s_{C}(x(\overline{\mathtt{sub}}^s_{B}\, y)))( \lambda w. (zw^{B}_B)^C_C)&&\text{Lemma}~\ref{lem: canonical definition for eta long normal forms}.\ref{eqn: 3 canonical definition for eta long normal forms}\\
 &\rightarrow_\beta  \lambda y.  \mathtt{sub}^s_{C}((\lambda w. (zw^{B}_B)^C_C)(\overline{\mathtt{sub}}^s_{B}\, y))\\
  &\rightarrow _\beta \lambda y.  \mathtt{sub}^s_{C} (z(\overline{\mathtt{sub}}^s_{B}\, y)^{B}_B)^C_C \\
   &\rightarrow_\beta^*  \lambda y.  \mathtt{sub}^s_{C} (z y^{B}_B)^C_C&&\text{Lemma}~\ref{lem: sub variable to eta long normal form}.\ref{eqn: 2 sub variable to eta long normal form}\\
    &=  \lambda y.  (\mathtt{sub}^s_{C}\, w^{C}_C)[ z y^{B}_B/w] \\ 
       &\rightarrow_\beta^*  \lambda y.  w^C_C[ z y^{B}_B/w] &&\text{induction hyp.} \\ 
          &=  \lambda y. (zy^B_B)^C_C \\ 
            &=_\alpha z^{B\multimap C}_{B \multimap C}&&\text{Lemma}~\ref{lem: canonical definition for eta long normal forms}.\ref{eqn: 3 canonical definition for eta long normal forms}
\end{align*} 
\end{enumerate}
\end{proof}
\begin{lem}\label{lem: sub} Let $s>0$. Let $A$ be a $\Pi_1$-type, and let  
$\Gamma=x_1:A_1, \ldots, x_n:A_n$ be a context of $\Sigma_1$-types. 
If $\Gamma[\mathbf{B}^s]\vdash M:A[\mathbf{B}^s]$, with $M$ normal, then:
\begin{equation*}
\mathtt{sub}^s_{A}\, M[\Gamma]\rightarrow_\beta^* M^\Gamma _A  \enspace .
\end{equation*}
\end{lem}
\begin{proof}
Let $Q_{\Gamma, A}$ be the number of universal quantifications in $A_1, \ldots, A_n, A$. We prove the result by induction on $\vert M \vert+ Q_{\Gamma, A}$. If $M=z$ then $\Gamma= z:A$ and $\mathtt{sub}^s_A\, M[\Gamma]= \mathtt{sub}^s_{A}(\overline{\mathtt{sub}}^s_{A}\, z)$. By point~\ref{eqn: 2 sub variable to eta long normal form} of Lemma~\ref{lem: sub variable to eta long normal form} and by Lemma~\ref{lem: sub preservation of eta long normal form} we have $\mathtt{sub}^s_A( \overline{\mathtt{sub}}^s_A \, z) \rightarrow_\beta^*  \mathtt{sub}^s_A \, z^A_A \rightarrow_\beta^* z^A_A$. If $M= \lambda z.N$ then we have two cases depending on the type of $M$:
\begin{enumerate}
\item Case $A= \forall \alpha. B$.   The $\lambda$-term $\mathtt{sub}^s_B$ has type $B[\mathbf{B}^s]\multimap B^-[\mathbf{B}^s]$, which is equal to $(B\langle \mathbf{B}^s/\alpha\rangle )[\mathbf{B}^s] \multimap (\forall \alpha. B)^-[\mathbf{B}^s]$, so that $\mathtt{sub}^s_B$ has also type $(\forall \alpha. B)[\mathbf{B}^s]\multimap (\forall \alpha. B)^-[\mathbf{B}^s]$. By Definition~\ref{defn: sub} we have $\mathtt{sub}^s_{B}= \mathtt{sub}^s_{\forall \alpha. B}$. By using the induction hypothesis, for every $M$ of type $(\forall \alpha. B)[\mathbf{B}^s]$ with context $\Gamma[\mathbf{B}^s]$, we have $\mathtt{sub}^s_{\forall \alpha. B}\, M[\Gamma]= \mathtt{sub}^s_B\, M[\Gamma]  \rightarrow_\beta^* M^{\Gamma}_{B}$. Moreover, by Lemma~\ref{lem: canonical definition for eta long normal forms}.\ref{eqn: 4 canonical definition for eta long normal forms}, $M^{\Gamma}_{B}= M^\Gamma_{\forall \alpha. A}$.
\item Case $A= B \multimap C$. Then we have:
\allowdisplaybreaks
\begin{align*}
\mathtt{sub}^s_{B \multimap C}\, M[\Gamma]&=(\lambda x.\lambda y.  \mathtt{sub}^s_{C}(x\, (\overline{\mathtt{sub}}^s_{B}\,y)))(\lambda z. N)[\Gamma]&&\text{Definition}~\ref{defn: sub}\\
&\rightarrow_\beta \lambda y.  \mathtt{sub}^s_{C}((\lambda z.N)[\Gamma] \, (\overline{\mathtt{sub}}^s_{B}\,y))\\
&\rightarrow_\beta  \lambda y.  \mathtt{sub}^s_{C}((N[\Gamma])[\overline{\mathtt{sub}}^s_{B}\,y/z])\\
&=  \lambda y.  \mathtt{sub}^s_{C}(N[\Gamma, y:B] )&&\text{Definition}~\ref{defn: sub substitution}\\
&\rightarrow_\beta^*  \lambda y.N^{\Gamma,y: B}_C&&\text{induction hyp.}\\
&=(\lambda y.N)^{\Gamma}_{B \multimap C}&&\text{Lemma}~\ref{lem: canonical definition for eta long normal forms}.\ref{eqn: 5 canonical definition for eta long normal forms}\\
&=_{\alpha} M^\Gamma_{B \multimap C}.
\end{align*}
\end{enumerate}
If $M=P[zN/w]$ then the type of $z$ cannot have an outermost universal quantification, because $\Gamma$ is a context of $\Sigma_1$-types. So $z$ has type of the form $B \multimap C$ in $\Gamma$.  Let $\Gamma'$ and $\Gamma''$ be contexts  such that $\Gamma= \Gamma', \Gamma'', z: B \multimap C$, $\operatorname{dom}(\Gamma')=FV(P)$, and $\operatorname{dom}(\Gamma'')=FV(N)$. Then we have:
\allowdisplaybreaks
\begin{align*}
\mathtt{sub}^s _{A}\, M[\Gamma]&= \mathtt{sub}^s_{A}(P[zN/w])[\Gamma]\\
&= \mathtt{sub}^s_{A}(P[\Gamma']  [(zN)[\Gamma'', z: B \multimap C]/w])\\
&=\mathtt{sub}^s_A (P[\Gamma'][(\overline{\mathtt{sub}}^s_{B \multimap C}\, z)(N[\Gamma''])/w])\\	    
&\rightarrow_\beta^* \mathtt{sub}^s_A (P[\Gamma'][ \overline{\mathtt{sub}}^s_C(z \,(\mathtt{sub}^s_{B}\,(N[\Gamma''])))/w])&&\text{Definition}~\ref{defn: sub}\\	       
&\rightarrow_\beta^* \mathtt{sub}^s_A (P[\Gamma'][ \overline{\mathtt{sub}}^s_C(z N^{\Gamma''}_B)/w] )&&\text{induction hyp.}\\	
&= (\mathtt{sub}^s_A \, P[\Gamma'][\overline{\mathtt{sub}}^s_C\, w/w] ) [z N^{\Gamma''}_B/w]\\	
&= (\mathtt{sub}^s_A \, P[\Gamma', w: C])[z N^{\Gamma''}_B/w]&&\text{Definition}~\ref{defn: sub substitution}\\
&\rightarrow_\beta^* P^{\Gamma', w: C}_A [z N^{\Gamma''}_B/w]&&\text{induction hyp.}\\	
&=(P[zN/w])^{\Gamma}_A &&\text{Lemma}~\ref{lem: canonical definition for eta long normal forms}.\ref{eqn: 7 canonical definition for eta long normal forms}
\end{align*}
\end{proof}

\subsection{The linear $\lambda$-term $ \mathtt{enc}^s_{A} $}
One missing ingredient  in the previous subsection is the value of $ s $, which is fixed to some strictly positive integer.  To determine $s$ we need the following property:
\begin{lem} \label{lem: pi1type bound} For every cut-free derivation $\mathcal{D}\triangleleft \Gamma \vdash M: B$ in $\mathsf{IMLL}_2$ which does not contain applications of $\forall$L, the following inequations hold:
\begin{equation}\label{eqn: 1 inequation}
\vert M \vert  \leq  \vert M^{\Gamma}_B \vert \leq \vert \Gamma^-  \vert + \vert B^- \vert \leq 2 \cdot \vert M_B^{\Gamma} \vert \enspace ,
\end{equation}
where $(\_ )^-$ is as in Definition~\ref{defn: stripping forall}, and $M^\Gamma_A$ is as in Definition~\ref{defn: eta long nf}.
\end{lem}
\begin{proof}
The inequation $\vert M \vert  \leq  \vert M^{\Gamma}_B \vert$ is by definition of $\eta$-long normal form. Now, let $\mathcal{D}^{\Gamma}_B$ be the $\eta$-expansion of $\mathcal{D}$, so that $\mathcal{D}^{\Gamma}_B \triangleleft \Gamma \vdash M^{\Gamma}_B: B$. We prove the remaining two inequations by induction on $\mathcal{D}^{\Gamma}_B$. If it is an axiom then, by definition of $\eta$-expansion, it must be of the form $x: \alpha \vdash x:\alpha$, where   $M^{\Gamma}_B=x$. Hence, $\vert x \vert \leq 2 \cdot \vert \alpha \vert \leq  2 \cdot \vert x \vert$. Both the rules $\multimap$R and $\multimap$L,  increase by one the overall size of the types in a judgment and of the corresponding term, so the inequalities still hold. Last, the rules for $\forall$ do not affect the size of both  $\Gamma^-, B^-$ and  $M^{\Gamma}_B$.
\end{proof}
\noindent
Notice that Lemma \ref{lem: pi1type bound} does not hold in general whenever $\mathcal{D}$ contains instances of the inference rule $\forall$L,   since one can exploit the inference rule $\forall$L to \enquote{compress} the size of a type. \\
Now, consider a cut-free derivation $	\mathcal{D}\triangleleft \vdash M:A$, where $A$ is  a ground type. Since negative occurrences of $\forall$ are not allowed in $A$, $\mathcal{D}$ contains no application of $\forall$L and,  by Lemma~\ref{lem: pi1type bound}, this implies  that $\vert M \vert \leq \vert A^-\vert$. This limits the number of variables a generic inhabitant of $A$ has, so that we can safely say that the variables of $ M $ must certainly belong to a fixed set $\lbrace \mathrm{x}_1, \ldots, \mathrm{x}_{\vert A^- \vert} \rbrace$. The next step is to show that we can encode every normal form as a tuple of booleans, i.e.~as elements in $ \mathbf{B}^s $ with a sufficiently large $ s $. Actually, we are interested in $\eta$-long normal forms only, due to the way the linear $\lambda$-term $\mathtt{sub}^s_{A}$ acts on inhabitants of $A$ as shown in the previous subsection. So, given a ground type $ A$, we can represent the $\eta$-long normal forms of type $A $ with tuples of type $ \mathbf{B}^{\mathcal{O}(\vert A^-\vert \, \cdot \, \log \vert A^-\vert)} $, since each such linear $\lambda$-term has at most $\vert A^- \vert$ symbols, each one encoded using around $\log \vert A^- \vert$ bits. By setting $s=c \cdot (\vert A^-\vert \, \cdot \,  \log \vert A^- \vert)$  for some $c>0$ large enough, there must exist a coding function $\lceil \_ \rceil: \Lambda_s \longrightarrow \mathbf{B}^s$, where  $\Lambda_{s} $ is the set of all normal linear $\lambda$-terms having size bounded by $s$. The role of the $\lambda$-term $\mathtt{enc}^s_{A}$ is to internalize the coding function   $\lceil \_ \rceil$ in $\mathsf{IMLL}_2$ as far as the $\eta$-long normal forms of a fixed type $A$ are concerned. 

The coming Lemma~\ref{lem: existence abs app}   relies on an 
iterated selection mechanism, i.e.~a nested  \texttt{if}-\texttt{then}-\texttt{else} construction. In order to define selection, we first we need to extend the projection in~\eqref{eqn: boolean projection} (Section~\ref{sec: duplication and erasure in a linear setting}).

\begin{defn}[Generalized projection]\label{defn: generalized projection}
Let  $A$  be a ground type. For all $k \geq 0$ and $\vec{m}=m_1, \ldots, m_ k \geq 0$,  the  linear $\lambda$-term $\pi^{\vec{m}}_1$  is defined below:
\begin{equation*}
\pi_1^{\vec{m} } \triangleq \begin{cases}    
\lambda z. \mathtt{let\ }z \mathtt{\ be\ }x,y \mathtt{\ in\ 
}(\mathtt{let\ } \mathtt{E}_{A}\,y \mathtt{\ be\ }I \mathtt{\ in\ }x) &\text{if }k=0\\ 
 \lambda z. \mathtt{let\ }z \mathtt{\ be\ }x,y \mathtt{\ in\ }(\mathtt{let\ } \mathtt{E}_{A}\,(y\, \mathtt{tt}^{m_1}  \ldots \mathtt{tt}^{m_k}) \mathtt{\ be\ }I \mathtt{\ in\ }x) &\text{if }k>0
 \end{cases} 
\end{equation*}
with type $B \otimes B \multimap B$, where $B \triangleq \mathbf{B}^{m_1} \multimap \ldots \multimap  \mathbf{B}^{m_k} \multimap A$. When $k=0$ we simply  write $\pi_1$ in place of $\pi_1^{\vec{m}}$, whose type is $A \otimes A \multimap A$.
\end{defn}

\begin{defn}[Generalized selection] \label{defn: generalized selection} Let  $A$ be a ground type and let $M_{\mathtt{tt}^n}$, $ M_{\langle \mathtt{tt}^{n-1},\mathtt{ff}\rangle}$, \ldots, $M_{\langle \mathtt{tt},\mathtt{ff}^{n-1}\rangle}$, $M_{\mathtt{ff}^n}$ be (not necessarily distinct) normal inhabitants of $\mathbf{B}^{m_1} \multimap \ldots \multimap   \mathbf{B}^{m_k} \multimap A$, for some   $n \geq 1$,  $k \geq 0 $, and  $\vec{m}=m_1, \ldots, m_k \geq 0$.  We define the linear $\lambda$-term:
\begin{equation}\label{eqn: generalized selection}
\mathtt{if \ }x \mathtt{\ then \ }[M_{\mathtt{tt}^n},  M_{\langle \mathtt{tt}^{n-1}, \mathtt{ff}\rangle}, \ldots, M_{\langle \mathtt{tt},\mathtt{ff}^{n-1}\rangle},M_{\mathtt{ff}^n} ]^{\vec{m} }
\end{equation}
with type $ \mathbf{B}^{m_1} \multimap \ldots \multimap \mathbf{B}^{m_k} \multimap A$ and context $x:\mathbf{B}^n$ by induction on $n$:
\begin{itemize}
\item $n=1$: $ \mathtt{if \ }x \mathtt{\ then \ }[M_{\mathtt{tt}},M_{\mathtt{ff}} ]^{\vec{m} } \triangleq  \pi_1^{\vec{m}} (x\, M_{\mathtt{tt}} \, M_{\mathtt{ff}})$.
\item $n>1$:  $ \mathtt{if \ }x \mathtt{\ then \ }[M_{\mathtt{tt}^n},  M_{\langle \mathtt{tt}^{n-1}, \mathtt{ff}\rangle}, \ldots, M_{\langle \mathtt{tt},\mathtt{ff}^{n-1}\rangle},M_{\mathtt{ff}^n} ]^{\vec{m}} \triangleq $
\begin{align*}
&\mathtt{ let \ }x \mathtt{ \ be \ } x_1, x_2  \mathtt{\ in\ } (\mathtt{if\ }x_2 \mathtt{\ then \ } \\
& \big[ (\lambda y_1. \mathtt{if \ }y_1 \mathtt{\ then \ } [P_{\mathtt{tt}^{n-1}},  P_{\langle \mathtt{tt}^{n-2}, \mathtt{ff}\rangle}, \ldots, P_{\langle \mathtt{tt},\mathtt{ff}^{n-2}\rangle},P_{\mathtt{ff}^{n-1}}]^{\vec{m}}),\\
&\phantom{\big[}(\lambda y_2. \mathtt{if \ }y_2 \mathtt{\ then \ } [Q_{\mathtt{tt}^{n-1}},  Q_{\langle \mathtt{tt}^{n-2}, \mathtt{ff}\rangle}, \ldots, Q_{\langle \mathtt{tt},\mathtt{ff}^{n-2 }\rangle},Q_{\mathtt{ff}^{n-1}}]^{\vec{m} })  \big]^{n-1, \vec{m}}) \, x_1 
\end{align*}
where, $\pi_1^{\vec{m}}$ is as in Definition~\ref{defn: generalized projection} and, for every $n$-tuple  $\langle \mathtt{b}_1, \ldots, \mathtt{b}_n \rangle$ of booleans, 
$ P_{\langle \mathtt{b}_1, \ldots, \mathtt{b}_n \rangle}\triangleq M_{\langle \langle \mathtt{b}_1, \ldots, \mathtt{b}_n \rangle, \mathtt{tt} \rangle}$, $
Q_{\langle \mathtt{b}_1, \ldots, \mathtt{b}_n \rangle}\triangleq M_{\langle  \langle \mathtt{b}_1, \ldots, \mathtt{b}_n \rangle,\mathtt{ff} \rangle}$.
\end{itemize}
when $k=0$ we feel free of ruling out the apex $\vec{m}$ in~\eqref{eqn: generalized selection}.
\end{defn}
\begin{lem} Let  $A$ be a ground type and let $M_{\mathtt{tt}^n}$, $ M_{\langle \mathtt{tt}^{n-1},\mathtt{ff}\rangle}$, \ldots, $M_{\langle \mathtt{tt},\mathtt{ff}^{n-1}\rangle}$, $M_{\mathtt{ff}^n}$ be (not necessarily distinct) normal inhabitants of $\mathbf{B}^{m_1} \multimap \ldots \multimap   \mathbf{B}^{m_k} \multimap A$, for some   $n \geq 1$,  $k \geq 0 $, and  $\vec{m}=m_1, \ldots, m_k \geq 0$.  
For every $n$-tuple of booleans  $\langle \mathtt{b}_1, \ldots, \mathtt{b}_n \rangle$ it holds that:
\begin{equation*}
\mathtt{if \ } \langle \mathtt{b}_1, \ldots, \mathtt{b}_n \rangle \mathtt{\ then \ }(M_{\mathtt{tt}^n},  M_{\langle \mathtt{tt}^{n-1}, \mathtt{ff}\rangle}, \ldots, M_{\langle \mathtt{tt},\mathtt{ff}^{n-1}\rangle},M_{\mathtt{ff}^n} ) \rightarrow_\beta^* M_{\langle \mathtt{b}_1, \ldots, \mathtt{b}_n \rangle}\enspace .
\end{equation*}
\end{lem}
\begin{proof}
Straightforward.
\end{proof}
Notice that, if $n=1$ and $k=0$  in Definition~\eqref{defn: generalized selection}, we get  the usual  \texttt{if}-\texttt{then}-\texttt{else} construction defined in~\cite{gaboardi2009light} as:
\begin{equation}\label{eqn: pi1m}
\mathtt{if}\ x \ \mathtt{then}\ M_1\ \mathtt{else}\ M_2
\triangleq  \pi_1(x\,M_1\,M_2)
\end{equation}
with type $A$ and context $x: \mathbf{B}$, where $\pi_1: A \otimes A \multimap A$ is as in Definition~\ref{defn: generalized projection}. Clearly, if  $\mathtt{b}_1\triangleq\mathtt{tt}$ and $\mathtt{b}_2\triangleq \mathtt{ff}$, then $\mathtt{if}\ \mathtt{b}_i\ \mathtt{then}\ M_1\ \mathtt{else}\ M_2 
\rightarrow_\beta^* M_i$ for $i=1,2$.

Before defining the linear $\lambda$-term $\mathtt{enc}_A^s$ we need to encode the $ \lambda $-abstractions and the applications in $\mathsf{IMLL}_2$.

\begin{lem} \label{lem: existence abs app} Let $s>0$. The following statements hold:
\begin{enumerate}[(1)]
\item \label{eqn: abs}  A linear $\lambda$-term $\mathtt{abs}^s: \mathbf{B}^s \multimap \mathbf{B}^{s}\multimap \mathbf{B}^s$ exists such that $\mathtt{abs}\lceil x \rceil \lceil M \rceil \rightarrow^*_\beta \lceil \lambda x. M \rceil$, if $\vert \lambda x .M \vert \leq s $ and $ x \in \lbrace \mathrm{x}_1, \ldots, \mathrm{x}_{s} \rbrace$.
\item\label{eqn: app}  A linear $\lambda$-term $\mathtt{app}^s: \mathbf{B}^s \multimap \mathbf{B}^{s}\multimap \mathbf{B}^s$ exists such that $\mathtt{app} \lceil M \rceil \lceil N \rceil \rightarrow^* _\beta \lceil MN 
\rceil $, if $\vert MN  \vert \leq s$.
\end{enumerate}
\end{lem}
\begin{proof}
We sketch the proof of Point~\ref{eqn: abs} only, since Point~\ref{eqn: app} is similar. Recall the notation in Definition~\ref{defn: generalized selection}. We let boolean values range over $b_1, b_2, \ldots$ and with $\mathtt{b}$ we denote the corresponding encoding of the boolean value $b$ in $\mathsf{IMLL}_2$.   The linear $\lambda$-term $\mathtt{abs}$ is  of the form:
\begin{equation*}
\lambda x. \lambda y. (\mathtt{if \ }x \mathtt{\ then\ } [ P_{\mathtt{tt}^s}, P_{\langle \mathtt{tt}^{s-1}, \mathtt{ff}\rangle}, \ldots, P_{\langle \mathtt{ff}, \mathtt{tt}^{s-1}\rangle},P_{\mathtt{ff}^s}]^s)\, y
\end{equation*}
where, for all $s$-tuple of booleans $T=\langle \mathtt{b}_1, \ldots, \mathtt{b}_s \rangle$, the linear $\lambda$-term  $P_{T}$ with type $\mathbf{B}^s \multimap \mathbf{B}^s$ is as follows:
\begin{equation*}
\lambda y.\mathtt{if \ }y \mathtt{\ then\ } [Q^T_{\mathtt{tt}^s}, Q^T_{\langle \mathtt{tt}^{s-1}, \mathtt{ff}\rangle}, \ldots, Q^T_{\langle \mathtt{ff}, \mathtt{tt}^{s-1}\rangle},Q^T_{\mathtt{ff}^s}]\enspace .
\end{equation*}
For all $T=\langle \mathtt{b}_1, \ldots, \mathtt{b}_s \rangle$ and for all $T'=\langle \mathtt{b}'_1, \ldots, \mathtt{b}'_s \rangle$ we define:
\begin{equation*}
 Q^T_{T'}= \begin{cases} \lceil \lambda x. M \rceil  &\text{if } \langle \mathtt{b}_1, \ldots, \mathtt{b}_s \rangle= \lceil x\rceil, \  \langle \mathtt{b}'_1, \ldots, \mathtt{b}'_s \rangle= \lceil M \rceil, \\ &  \text{and }\vert \lambda x. M \vert \leq  s\\ \\
\langle \mathtt{tt},\overset{s}{ \ldots}, \mathtt{tt} \rangle &\text{otherwise}.
\end{cases}
\end{equation*}
\end{proof}
The $\lambda$-term $\mathtt{enc}^s_{A}$, given a value  $V_A$  in $\eta$-long normal form and of type $A$, combines the $\lambda$-terms $\mathtt{abs}^s$ and $\mathtt{app}^s$ to construct its encoding.
\begin{defn}[The linear $\lambda$-term $\mathtt{enc}^s_{A}$] \label{defn: enc} Let $s>0$. We define the linear $\lambda$-terms $\mathtt{enc}^s_{A}: A^-[\mathbf{B}^s]\multimap \mathbf{B}^s$, where $A$ is a  $\Pi_1$-type, and $\overline{\mathtt{enc}}^s_A: \mathbf{B}^s \multimap A^-[\mathbf{B}^s]$, where $A$ is a $\Sigma_1$-type, by simultaneous induction on the size of $A$:
    \begin{align*}
& \mathtt{enc}^s_{\alpha}  \triangleq  \lambda z. z &&     \mathtt{enc}^s_{B \multimap C}
       \triangleq \lambda z. \mathtt{abs}^s \lceil x \rceil\,   (\mathtt{enc}^s_{C}\,   (z\, (\overline{\mathtt{enc}}^s_{B}\,   \lceil x \rceil))) \\
       &    \overline{\mathtt{enc}}^s_{\alpha} \triangleq   \lambda z. z
       &&  \overline{\mathtt{enc}}^s_{B \multimap C}\triangleq \lambda z. \lambda x.   \overline{\mathtt{enc}}^s_{C}\, (\mathtt{app}^s    z\, (\mathtt{enc}^s_{B}\,x))
    \end{align*}
    \noindent
    with $x$  chosen fresh in $\lbrace \mathrm{x}_1, \ldots, \mathrm{x}_{s} \rbrace$.\qed
    \end{defn}
    The following will be used to compact the proof of some of the coming lemmas.
    \begin{defn} 
Let $s>0$, and let $A$ be a $\Pi_1$-type and $\Gamma=x_1:A_1, \ldots, x_n:A_n$ be a context of $\Sigma_1$-types. If  $M$ is an inhabitant of type $A^-[\mathbf{B}^s]$   with context $\Gamma^-[\mathbf{B}^s]$ then  $M[\Gamma]$ denotes the substitution:
\begin{equation*}
  M[\overline{\mathtt{enc}}^s_{A_1}\,x'_1 /x_1 ,\ldots, \overline{\mathtt{enc}}^s_{A_n}\, x'_n/x_n] 
\end{equation*}    
for some $x'_1, \ldots, x'_n$. \qed
\end{defn} 
To prove that $\mathtt{enc}^s_{A}$ is able to encode a value $V_A$ of type $A$ we need an intermediate step. We first prove that $\mathtt{enc}^s_{A}$   substitutes every $\lambda$-abstraction in $V_{A}$ with an instance of  $\mathtt{abs}^s$, and every application with an instance of  $\mathtt{app}^s$, thus producing a \enquote{precode}. Then we  prove that, when every free variable in it has been substituted with its respective encoding,  the precode reduces to $\lceil V_{A}\rceil$.
\begin{defn} Let $s>0$. If $M$ is a linear $\lambda$-term in normal form such that $\vert M \vert \leq s$, we define $M^s$ by induction on $\vert M \vert$:
\begin{enumerate}
\item $M=x$ if and only if $M^s=x$,
\item $M= \lambda x. N$ if and only if $M^s= \mathtt{abs}^s \, \lceil x' \rceil \, N^s[\lceil x' \rceil/x] $,
\item $M=PQ$ if and only if $M^s= \mathtt{app}^s \, P^s \, Q^s$,
\end{enumerate} 
where $x'$ is fresh, chosen in $\lbrace \mathrm{x}_1, \ldots, \mathrm{x}_{s} \rbrace$. \qed
\end{defn}
\begin{lem} \label{lem: substitution lemma for s} Let $s>0$. If $M$ and $N$ are linear $\lambda$-terms, then $M^s[N^s/x]=(M[N/x])^s$.
\end{lem}
\begin{proof}
By induction on $\vert M \vert$. If $M=x$ then $x^s[N^s/x]= x[N^s/x]= N^s= (x[N/x])^s$. If $M=PQ$ then either $x$ occurs in $P$ or it occurs in $Q$, and let us consider the case  $x \in FV(P)$, the other case being similar: by using the induction hypothesis we have $(PQ)^s [N^s/x]= \mathtt{app}^s \, P^s[N^s/x] \, Q^s= \mathtt{app}^s (P[N/x])^s \, Q^s= (P[N/x]Q)^s=((PQ)[N/x])^s$. If $M= \lambda y. P$ then we have that $(\lambda y. P)^s [N^s /x]$ $= \mathtt{abs}^s \, \lceil y' \rceil \, P^s[N^s /x][\lceil y' \rceil/y]$ $= \mathtt{abs}^s \, \lceil y' \rceil   \, (P[N /x])^s [\lceil y' \rceil /y] $ $= (\lambda y. P[N/x])^s$ $= ((\lambda y. P)[N/x])^s$.
\end{proof}
\begin{lem}\label{lem: precoding lemma} Let $s>0$. If $M$ is a linear $\lambda$-term in normal form such that $\vert M \vert \leq s$ with free variables $x_1, \ldots, x_n$ then 
\[M^s [ \vec{\lceil x' \rceil}]\rightarrow_\beta^* \lceil M[x'_1/x_1, \ldots, x'_n/x_n] \rceil\]
where $ \vec{\lceil x' \rceil}=[\lceil x'_1 \rceil/x_1, \ldots, \lceil x'_n \rceil/x_n]$ and  $x'_1, \ldots, x'_n$  are distinct and fresh in $\lbrace \mathrm{x}_1, \ldots, \mathrm{x}_{s} \rbrace$.
\end{lem}
\begin{proof}
By induction on $\vert M \vert$. If $M=x$ then $\exists i \leq n$ $x_i=x$, so that  
$x^s [\lceil x'_i \rceil/x]= x [\lceil x'_i \rceil/x]=\lceil x'_i\rceil= \lceil  x[x'_i/x] \rceil$. If $M= \lambda y. N$ then, using the induction hypothesis, we have: 
\allowdisplaybreaks
\begin{align*}
(\lambda y. N)^s  [ \vec{\lceil x' \rceil}]&= (\mathtt{abs}^s \, \lceil y' \rceil \, N^s [\lceil y' \rceil/y]) [ \vec{\lceil x' \rceil} ]\\
&=\mathtt{abs}^s \, \lceil y' \rceil \, N^s[ \vec{\lceil x' \rceil}, \lceil y' \rceil /y]\\
&\rightarrow_\beta^* \mathtt{abs}^s \, \lceil y' \rceil \, \lceil N[x_1'/x_1, \ldots, x'_n/x_n, y'/y] \rceil\\
&\rightarrow_\beta^* \lceil \lambda y'. N[x_1'/x_1, \ldots, x'_n/x_n, y'/y] \rceil&& \text{Lemma~\ref{lem: existence abs app}}\\
&=_{\alpha}\lceil  (\lambda y. N)[x_1'/x_1, \ldots, x'_n/x_n] \rceil.
\end{align*}
If $M=PQ$ then let $y_1, \ldots, y_m$ (resp. $z_1, \ldots, z_k$) be the free variables of $P$ (resp. $Q$), and let $\vec{x'}=y'_1, \ldots, y'_m, z'_1, \ldots, z'_k$. Then we have: 
\allowdisplaybreaks
\begin{align*}
&(PQ)^s [ \vec{\lceil x' \rceil}]=\\
&= \mathtt{app}^s \, P^s  [ \vec{\lceil y' \rceil}] \, Q^s [ \vec{\lceil z' \rceil}] \\
&\rightarrow_\beta^* \mathtt{app}^s \,  \lceil P[y_1'/y_1, \ldots, y'_m/y_m] \rceil \, \lceil Q[z_1'/z_1, \ldots, z'_k/z_k] \rceil \\
&\rightarrow_\beta^*  \lceil P[y_1'/y_1, \ldots, y'_m/y_m] Q[z_1'/z_1, \ldots, z'_k/z_k]  \rceil && \text{Lemma~\ref{lem: existence abs app}}\\
&= \lceil (PQ)[x_1'/x_1, \ldots, x'_n/x_n] \rceil.  
\end{align*}
\end{proof}
It is easy to check that  if $M$ is an inhabitant of a $\Pi_1$-type $A$ with context $\Gamma=x_1:A_1, \ldots, x_n:A_n$ of $\Sigma_1$-types, then $M$ has also type $A^-[\mathbf{B}^s]$   with context $\Gamma^-[\mathbf{B}^s]$.
\begin{lem}\label{lem: enc} Let $M$ be a $\eta$-long normal form of type $A$  with context $\Gamma=x_1:A_1, \ldots, x_n:A_n$, where $A$ is a $\Pi_1$-type and $\Gamma$ is a context of $\Sigma_1$-types, and let $\sum_{i=1}^m\vert A_i^- \vert+ \vert A^-  \vert=k$ and $s=c \cdot (k\cdot \log k)$, for some $c$ large enough. Then:
\begin{equation*}
 (\mathtt{enc}^s_A\, M[\Gamma]) [\vec{\lceil x' \rceil}] \rightarrow_\beta^*\lceil M[x'_1/x_1, \ldots, x'_n/x_n]  \rceil \enspace , 
\end{equation*}
 where $ \vec{\lceil x' \rceil}=[\lceil x'_1 \rceil/x_1, \ldots, \lceil x'_n \rceil/x_n]$, with $x'_1, \ldots, x'_n$  distinct and chosen fresh in $\lbrace \mathrm{x}_1, \ldots, \mathrm{x}_{s} \rbrace$.
\end{lem}
\begin{proof} By Lemma~\ref{lem: precoding lemma} it suffices to prove by induction on $\vert M \vert $ that the reduction  $\mathtt{enc}^s_A\, M[\Gamma]  \rightarrow_\beta^* M^s$ holds.  If $M=x$ then $A= \alpha$ and $\Gamma= x:\alpha$, because   $M$ is in $\eta$-long normal form,  so that we have $\mathtt{enc}_{\alpha}^s \, x[x: \alpha]= \mathtt{enc}^s_{\alpha}(\overline{\mathtt{enc}}_\alpha^s \, x) \rightarrow_\beta^* x= x^s $.  If $M= \lambda y.N$ then $A= B \multimap C$, so that:
\allowdisplaybreaks
\begin{align*}
&\mathtt{enc}_{B \multimap C}^s ((\lambda y.N)[\Gamma]) \\
&\rightarrow_\beta \mathtt{abs}^s\, \lceil x' \rceil (\mathtt{enc}^s_C((\lambda y. N[\Gamma])(\overline{\mathtt{enc}}^s_B \, \lceil y' \rceil)))  &&\text{Definition~\ref{defn: enc}}\\
&\rightarrow_\beta \mathtt{abs}^s\, \lceil y' \rceil (\mathtt{enc}^s_C(N[\Gamma][\overline{\mathtt{enc}}^s_B \, \lceil y' \rceil/y])) \\
&= \mathtt{abs}^s\, \lceil y' \rceil  (\mathtt{enc}^s_C(N[\Gamma][\overline{\mathtt{enc}}^s_B \, x/x])) [ \lceil y' \rceil /y]\\
&= \mathtt{abs}^s\, \lceil y' \rceil (\mathtt{enc}^s_C(N[\Gamma,  y:B]))[ \lceil y'\rceil/y]&&\text{Definition~\ref{defn: enc}}\\
&\rightarrow_\beta^* \mathtt{abs}^s\, \lceil y ' \rceil (N^s [\lceil y'\rceil /y]) && \text{induction hyp.}\\
&=  (\lambda y. N)^s .
\end{align*}
Last, suppose $M=P[yN/x]$, and let $\Sigma$, $\Delta$ be contexts such that  $\Gamma=\Sigma, \Delta, y: B \multimap C$, $\operatorname{dom}(\Sigma)=FV(P)$, and $\operatorname{dom}(\Delta)=FV(N)$.   Then we have:
\allowdisplaybreaks
\begin{align*}
&\mathtt{enc}^s_{A} (P[yN/x])[\Gamma]\\
&= \mathtt{enc}^s_{A} (P[\Sigma][(yN)[\Delta, y:B \multimap C]   /x]) \\
&= \mathtt{enc}^s_{A} (P[\Sigma][(\overline{\mathtt{enc}}^s_{B \multimap C}\,  y  ) N[\Delta]   /x])  \\
&\rightarrow_\beta^* \mathtt{enc}^s_{A} (P[\Sigma][ \overline{\mathtt{enc}}^s_{C}(\mathtt{app}^s \,   y  (\mathtt{enc}^s_B \, N[\Delta])) /x]) &&\text{Definition~\ref{defn: enc}}\\
& \rightarrow_\beta^* \mathtt{enc}^s_{A} (P[\Sigma][ \overline{\mathtt{enc}}^s_{C}(\mathtt{app}^s \,    y \, N^s) /x])&&\text{induction hyp.}\\
&= \mathtt{enc}^s_{A} (P[\Sigma][ \overline{\mathtt{enc}}^s_{C}(yN)^s  /x])\\
&=\mathtt{enc}^s_{A} (P[\Sigma, x: C])  [ (yN)^s   /x]&&\text{Definition~\ref{defn: enc}}\\
&\rightarrow_\beta^* P^s  [ (yN)^s    /x] &&\text{induction hyp.}\\
&=(P[yN/x])^s  && \text{Lemma~\ref{lem: substitution lemma for s}}. 
\end{align*}
\end{proof}

\subsection{The linear $\lambda$-term $ \mathtt{dec}^s_{A}$} \label{subsection: dec and DBn}
The linear $ \lambda $-term $ \mathtt{dec}^s_{A} $ is the component of $\mathtt{D}_{A}$ requiring the type inhabitation. Roughly, it takes in input a tuple of boolean values encoding the $\eta$-long normal form $V_{A}$ of a ground type $A$, and it produces the pair $\langle V_{A}, V_{A}\rangle$. To ensure that $ \mathtt{dec}^s_{A} $ is defined on all possible inputs, it is built in such a way that it returns a default inhabitant of $A$ whenever the tuple of booleans in input does not encode any $\lambda$-term. 
\begin{defn}[The linear $\lambda$-term $\mathtt{dec}^s_{A}$]\label{defn: dec} Let $A$ be a ground type  and let  $U$ be a value of type $A$. If  for some $c$ large enough $s=c \cdot (\vert A^- \vert \cdot \log \vert A^- \vert )$,  then we  define the linear $\lambda$-term $\mathtt{dec}^s_{A}: \mathbf{B}^s \multimap A\otimes A$ as follows:
\begin{equation*}
\lambda x. \mathtt{if \ }x \mathtt{\ then\ } [P_{\mathtt{tt}^s}, P_{\langle \mathtt{tt}^{s-1}, \mathtt{ff}\rangle}, \ldots, P_{\langle \mathtt{ff}, \mathtt{tt}^{s-1}\rangle},P_{\mathtt{ff}^s}]
\end{equation*}
where, for all $T=\langle \mathtt{b}_1, \ldots, \mathtt{b}_s \rangle$ of type $\mathbf{B}^s$:
\begin{equation*}
 P_{T}= \begin{cases} \langle V_{A}, V_{A} \rangle  &\text{if }\langle \mathtt{b}_1, \ldots, \mathtt{b}_s \rangle= \lceil V_{A} \rceil \\
\langle U, U \rangle &\text{otherwise}.
\end{cases}
\end{equation*}
\end{defn}
We are now able to prove the fundamental result of this section: 
\begin{thm}[Duplication~\cite{mairson2003computational}]
Every inhabited ground type is duplicable.
\end{thm}
\begin{proof}
 The duplicator $\mathtt{D}_{A}$ of a inhabited ground type is defined as follows: we fix $s=c \cdot (\vert A^- \vert  \cdot \log \vert A^- \vert )$,  we fix a default value  $U$ of $A$ (see Definition~\ref{defn: dec}), and we set:
\begin{equation*}
\mathtt{D}_{A}\triangleq \mathtt{dec}^s_{A}\circ \mathtt{enc}^s_{A}\circ \mathtt{sub}^s_{A}
\end{equation*}
which has type $A \multimap A \otimes A$. By Lemma~\ref{lem: sub}, Lemma~\ref{lem: enc}, and Definition~\ref{defn: dec} the conclusion follows. Moreover, for all values $V$ of type $A$, we have:
\begin{equation*}
\mathtt{D}_A\, V \rightarrow_{\beta}^* \langle V_A, V_A\rangle \rightarrow_{\eta}^* \langle V, V\rangle \enspace .
\end{equation*}
\end{proof}
\begin{rem}\label{rem: duplicator} If $A$ is a ground type inhabited by the value $U$, we shall  write $\mathtt{D}_{A}^{U}$ to stress that the default inhabitant of $A$ used in  constructing the duplicator $\mathtt{D}_{A}$ of $A$ is $U$.
\end{rem}

\end{document}